\documentclass[12pt,english,reqno]{smfart}
\usepackage{mathptmx}
\usepackage{color}
\definecolor{mycolor}{rgb}{1, 0, 0}
\definecolor{mycolorG}{rgb}{0,0.25,0}
\definecolor{mycolorP}{rgb}{0.5,0,1}
\theoremstyle{plain}
\newtheorem{thm}{Theorem}[section]
\newtheorem{lem}[thm]{Lemma}
\newtheorem{prop}[thm]{Proposition}
\newtheorem{defn}[thm]{Definition}
\newtheorem{cor}[thm]{Corollary}
\theoremstyle{definition}
\newtheorem*{rem}{Remark}
\renewcommand{\theequation}{\arabic{section}.\arabic{equation}}
\begin{document}
\title[Superbosonization]
{Superbosonization of\\ invariant random matrix ensembles}
\author{P.\ Littelmann}
\address{Mathematisches Institut, Universit\"at zu K\"oln, Weyertal
86-90, D-50931 K\"oln, Germany; e-mail: littelma@math.uni-koeln.de}
\author{H.-J.\ Sommers}
\address{Fachbereich Physik, Universit\"at Duisburg-Essen, Lotharstr.\
1, D-47048 Duisburg, Germany; e-mail: h.j.sommers@uni-due.de}
\author{M.R.\ Zirnbauer}
\address{Institut f\"ur Theoretische Physik, Universit\"at zu K\"oln,
Z\"ulpicher Str.\ 77, D-50937 K\"oln, Germany; e-mail:
zirn@thp.uni-koeln.de}
\date{July 18, 2007}
\begin{abstract}
  'Superbosonization' is a new variant of the method of commuting and
  anti-commuting variables as used in studying random matrix models of
  disordered and chaotic quantum systems. We here give a concise
  mathematical exposition of the key formulas of superbosonization.
  Conceived by analogy with the bosonization technique for Dirac
  fermions, the new method differs from the traditional one in that
  the superbosonization field is dual to the usual
  Hubbard-Stratonovich field. The present paper addresses invariant
  random matrix ensembles with symmetry group $\mathrm{U}_n \,$,
  $\mathrm{O}_n \,$, or $\mathrm {USp}_n\,$, giving precise
  definitions and conditions of validity in each case. The method is
  illustrated at the example of Wegner's $n$-orbital model.
  Superbosonization promises to become a powerful tool for
  investigating the universality of spectral correlation functions for
  a broad class of random matrix ensembles of non-Gaussian and/or
  non-invariant type.
\end{abstract}
\maketitle

\section{Introduction and overview}
\label{sect:intro}

The past 25 years have seen substantial progress in the physical
understanding of insulating and metallic behavior in disordered
quantum Hamiltonian systems of the random Schr\"odinger and random
band matrix type. A major role in this development, bearing
especially on the metallic regime and the metal-insulator transition,
has been played by the method of commuting and anti-commuting
variables, or supersymmetry method for short. Assuming a Gaussian
distribution for the disorder, this method proceeds by making a
Hubbard-Stratonovich transformation followed by a saddle-point
approximation (or elimination of the massive modes) to arrive at an
effective field theory of the non-linear sigma model type. This
effective description has yielded many new results including, e.g.,
the level statistics in small metallic grains, localization in thick
disordered wires, and a scaling theory of critical systems in higher
dimension \cite{efetov-book}.

While the method has been widely used and successfully so, there
exist some limitations and drawbacks. For one thing, the method works
well only for systems with normal-distributed disorder; consequently,
addressing the universality question for non-Gaussian distributions
(like the invariant ensembles studied via the orthogonal polynomial
method) has so far been beyond reach. For another, the symmetries of
the effective theory are not easy to keep manifest when using the
mathematically well-founded approach of Sch\"afer and Wegner
\cite{SW}. A problem of lesser practical relevance is that the
covariance matrix of the random variables, which is to be inverted by
the Hubbard-Stratonovich transformation, does not always have an
inverse \cite{SW}.

In this paper we introduce a new variant of the supersymmetry
approach which is complementary to the traditional one. Inspired by
the method of bosonization of Dirac fermions, and following an
appellation by Efetov and coworkers, we refer to the new method as
`superbosonization'. As we will see, in order for superbosonization
to be useful the distribution of the random Hamiltonians must be
invariant under some symmetry group, and this group cannot be `too
small' in a certain sense. We expect the method to be at its best for
random matrix ensembles with a local gauge symmetry such as Wegner's
$n$-orbital model \cite{N-orbital} with gauge group $K =
\mathrm{U}_n\,$, $\mathrm{O}_n\,$, or $\mathrm{USp}_n\,$.

Superbosonization differs in several ways from the traditional method
of Sch\"afer and Wegner \cite{SW} based on the Hubbard-Stratonovich
(HS) transformation: (i) the superbosonization field has the physical
dimension of 1/energy (whereas the HS field has the dimension of
energy). (ii) For a fixed symmetry group $K$, the target space of the
superbosonization field is always the \emph{same} product of compact
and non-compact symmetric spaces regardless of where the energy
parameters are, whereas the HS field of the Sch\"afer-Wegner method
changes as the energy parameters move across the real axis. (iii) The
method is not restricted to Gaussian disorder distributions. (iv) The
symmetries of the effective theory are manifest at all stages of the
calculation.

A brief characterization of what is meant by the physics word of
`superbosonization' is as follows. The point of departure of the
supersymmetry method (in its old variant as well as the new one) is
the Fourier transform of the probability measure of the given
ensemble of disordered Hamiltonians. This Fourier transform is
evaluated on a supermatrix built from commuting and anti-commuting
variables and thus becomes a superfunction; more precisely a
function, say $f\,$, which is defined on a complex vector space $V_0$
and takes values in the exterior algebra $\wedge (V_1^\ast)$ of
another complex vector space $V_1\,$. If the probability measure is
invariant under a group $K$, so that the function $f$ is equivariant
with respect to $K$ acting on $V_0$ and $V_1\,$, then a standard
result from invariant theory tells us that $f$ can be viewed as the
pullback of a superfunction $F$ defined on the quotient of $V = V_0
\oplus V_1$ by the group $K\,$. The heart of the superbosonization
method is a formula which reduces the integral of $f$ to an integral
of the lifted function $F\,$. Depending on how the dimension of $V$
compares with the rank of $K$, such a reduction step may or may not
be useful for further analysis of the integral. Roughly speaking,
superbosonization gets better with increasing value of $\mathrm{rank}
(K)$. From a mathematical perspective, superbosonization certainly
promises to become a powerful tool for the investigation and proof of
the universality of spectral correlations for a whole class of random
matrix ensembles that are not amenable to treatment by existing
techniques.

Let us now outline the plan of the paper. In Sects.\ \ref{sect:1.1}
and \ref{sect:1.2} we give an informal introduction to our results,
which should be accessible to physicists as well as mathematicians. A
concise summary of the motivation (driven by random matrix
applications) for the mathematical setting of this paper is given in
Sect.\ \ref{sect:susy-method}.

In Sect.\ \ref{sect:BB-sector} we present a detailed treatment of the
special situation of $V_1 \equiv 0$ (the so-called Boson-Boson
sector), where anti-commuting variables are absent. This case was
treated in an inspiring paper by Fyodorov \cite{fyodorov}, and our
final formula -- the transfer of the integral of $f$ to an integral
of $F$ -- coincides with his. The details of the derivation, however,
are different. While Fyodorov employs something he calls the
Ingham-Siegel integral, our approach proceeds directly by push
forward to the quotient $V_0 // K^\mathbb{C}$. Another difference is
that our treatment covers each of the three classical symmetry groups
$K = \mathrm{U}_n\,$, $\mathrm{O}_n\,$, and $\mathrm{USp}_n \,$, not
just the first two.

Sect.\ \ref{sect:FF-sector} handles the complementary situation $V_0
= 0$ (the so-called Fermion-Fermion sector). In this case the
starting point is the Berezin integral of $f \in \wedge(V_1^\ast)$,
i.e., one differentiates once with respect to each of the
anti-commuting variables, or projects on the top-degree component of
$\wedge(V_1^\ast)$. From a theoretical physicist's perspective, this
case is perhaps the most striking one, as it calls for the mysterious
step of transforming the Berezin integral of $f$ to an integral of
the lifted function $F$ over a compact symmetric space. The
conceptual difficulty here is that many choices of $F$ exist, and any
serious theoretical discussion of the matter has to be augmented by a
proof that the final answer does not depend on the specific choice
which is made.

Finally, Sect.\ \ref{sect:full-susy} handles the full situation where
$V_0 \not= 0$ and $V_1 \not= 0$ (i.e., both types of variable,
commuting and anti-commuting, are present). Heuristic ideas as to how
one might tackle this situation are originally due to Lehmann, Saher,
Sokolov and Sommers \cite{LSSS} and to Hackenbroich and
Weidenm\"uller \cite{HW}. These ideas have recently been pursued by
Efetov and his group \cite{efetov-bos} and by Guhr \cite{guhr}, but
their papers are short of mathematical detail -- in particular, the
domain of integration after the superbosonization step is left
unspecified -- and address only the case of unitary symmetry. In
Sect.\ \ref{sect:full-susy} we supply the details missing from these
earlier works and prove the superbosonization formula for the cases
of $K = \mathrm{U}_n \,$, $\mathrm{O}_n\,$, and $\mathrm{USp}_n\,$,
giving sufficient conditions of validity in each case. While it
should certainly be possible to construct a proof based solely on
supersymmetry and invariant-theoretic notions including Howe dual
pairs, Lie superalgebra symmetries and the existence of an invariant
Berezin measure, our approach here is different and more
constructive: we use a chain of variable transformations reducing the
general case to the cases dealt with in Sects.\ \ref{sect:BB-sector}
and \ref{sect:FF-sector}.

\textit{Acknowledgment}. This paper is the product of a
mathematics-physics research collaboration funded by the Deutsche
Forschungsgemeinschaft via SFB/TR 12.

\subsection{Basic setting}
\label{sect:1.1}

Motivated by the method of commuting and anti-commuting variables as
reviewed in Sect.\ \ref{sect:susy-method}, let there be a set of
complex variables $Z_c^i$ with complex conjugates $\tilde{Z}_i^c :=
\overline{Z_c^i}\,$, where indices are in the range $i = 1, \ldots,
n$ and $c = 1, \ldots, p\,$. Let there also be two sets of
anti-commuting variables $\zeta_e^i$ and $\tilde{\zeta}_i^e$ with
index range $i = 1, \ldots, n$ and $e = 1, \ldots, q\,$. (Borrowing
the language from the physics context where the method is to be
applied, one calls $n$ the number of orbitals and $p$ and $q$ the
number of bosonic resp.\ fermionic replicas.) It is convenient and
useful to arrange the variables $Z_c^i\,$, $\zeta_e^i$ in the form of
rectangular matrices $Z\,$, $\zeta$ with $n$ rows and $p$ resp.\ $q$
columns. A similar arrangement as rectangular matrices is made for
the variables $\tilde{Z}_i^c$, $\tilde{\zeta}_i^e$, but now with $p$
resp.\ $q$ rows and $n$ columns.

We are going to consider integrals over these variables in the sense
of Berezin. Let
\begin{equation}\label{eq:1.1}
    D_{Z,\bar{Z};\zeta,\tilde{\zeta}} := 2^{pn} \prod_{c=1}^p\,
    \prod_{i=1}^n \vert d\mathfrak{Re}(Z_c^i) \, d\mathfrak{Im}
    (Z_c^i) \vert \otimes (2\pi)^{-qn} \prod_{e=1}^q \prod_{j=1}^n
    \frac{\partial^2} {\partial\zeta_e^j \, \partial\tilde{\zeta}_j^e}
    \;,
\end{equation}
where the derivatives are left derivatives, i.e., we use the sign
convention $\frac{\partial^2} {\partial\zeta \partial\tilde{\zeta}}
\tilde{\zeta} \zeta = 1$, and the product of derivatives projects on
the component of maximum degree in the anti-commuting variables. The
other factor is Lebesgue measure for the commuting complex variables
$Z$. We here denote such integrals by
\begin{equation}\label{eq:fund-int}
    \int f \equiv \int_{\mathrm{Mat}_{n,p}(\mathbb{C})} D_{Z,\bar{Z}
    ; \zeta,\tilde{\zeta}} \, f(Z,\tilde{Z};\zeta,\tilde{\zeta})
\end{equation}
for short. The domain of integration will be the linear space of all
complex rectangular $n \times p$ matrices $Z$, with $\tilde{Z} =
Z^\dagger \in \mathrm{Mat}_{p,\,n}(\mathbb{C})$ being the Hermitian
adjoint of $Z\,$. We assume that our integrands $f$ decrease at
infinity so fast that the integral $\int f$ exists.

In the present paper we will be discussing such integrals for the
particular case where the integrand $f$ has a Lie group symmetry.
More precisely, we assume that a Lie group $K$ is acting on
$\mathbb{C}^n$ and this group is either the unitary group
$\mathrm{U}_n\,$, or the real orthogonal group $\mathrm{O}_n\,$, or
the unitary symplectic group $\mathrm{USp}_n\,$. The fundamental
$K$-action on $\mathbb{C}^n$ gives rise to a natural action by
multiplication on the left resp.\ right of our rectangular matrices:
$Z \mapsto g Z\,$, $\zeta \mapsto g \zeta$ and $\tilde{Z} \mapsto
\tilde{Z} g^{-1}$, $\tilde{\zeta} \mapsto \tilde{\zeta} g^{-1}$
(where $g \in K$). The functions $f$ to be integrated shall have the
property of being $K$-invariant:
\begin{equation}\label{eq:K-symmetry}
    f(Z,\tilde{Z};\zeta,\tilde{\zeta}) = f(g Z,\tilde{Z}g^{-1} ;
    g\zeta,\tilde{\zeta}g^{-1}) \quad (g \in K) \;.
\end{equation}
We wish to establish a reduction formula for the Berezin integral
$\int f$ of such functions. This formula will take a form that varies
slightly between the three cases of $K = \mathrm{U}_n\,$, $K =
\mathrm{O}_n\,$, and $K = \mathrm{USp}_n\,$.

\subsubsection{The case of $\,\mathrm{U}_n$-symmetry}

Let then $f$ be an analytic and $\mathrm{U}_n$-invariant function of
our basic variables $Z, \tilde{Z}, \zeta, \tilde{\zeta}$ for
$\tilde{Z} = Z^\dagger$. We now make the further assumption that $f$
extends to a $\mathrm{GL}_n(\mathbb{C})$-invariant \emph{holomorphic}
function when $Z$ and $\tilde{Z}$ are viewed as \emph{independent}
complex matrices; which means that the power series for $f$ in terms
of $Z$ and $\tilde{Z}$ converge everywhere and that the symmetry
relation (\ref{eq:K-symmetry}) for the extended function $f$ holds
for all $g \in \mathrm{GL}_n(\mathbb{C})$, the complexification of
$\mathrm{U}_n\,$. The rationale behind these assumptions about $f$ is
that they guarantee the existence of another function $F$ which lies
`over' $f$ in the following sense.

It is a result of classical invariant theory \cite{Howe1989} that the
algebra of $\mathrm{GL}_n(\mathbb{C})$-invariant polynomial functions
in $Z$, $\tilde{Z}$, $\zeta$, $\tilde{\zeta}$ is generated by the
quadratic invariants
\begin{displaymath}
    (\tilde{Z} Z)_c^{c^\prime} \equiv \tilde{Z}_i^{c^\prime} Z_c^i \;,
    \quad (\tilde{Z} \zeta)_e^{c^\prime} \equiv \tilde{Z}_i^{c^\prime}
    \zeta_e^i \;, \quad (\tilde{\zeta} Z)_c^{e^\prime} \equiv
    \tilde{\zeta}_i^{e^\prime} Z_c^i \;, \quad (\tilde{\zeta}
    \zeta)_e^{e^\prime} \equiv \tilde{\zeta}_i^{e^\prime}\zeta_e^i\;.
\end{displaymath}
Here we are introducing the summation convention: an index that
appears twice, once as a subscript and once as a superscript, is
understood to be summed over.

How does this invariant-theoretic fact bear on our situation? To
answer that, let $F$ be a holomorphic function of complex variables
$x_c^{c^\prime}$, $y_e^{e^\prime}$ and anti-commuting variables
$\sigma_e^{c^\prime}$, $\tau_c^{e^\prime}$ with index range $c,
c^\prime = 1, \ldots, p$ and $e, e^\prime = 1, \ldots, q\,$. Again,
let us organize these variables in the form of matrices, $x = (x_c^{
c^\prime})$, $y = (y_e^{e^\prime})$, etc., and write $F$ as
\begin{displaymath}
    F(x_c^{c^\prime},y_e^{e^\prime};\sigma_e^{c^\prime},\tau_c^{e^\prime})
    \equiv F \begin{pmatrix} x &\sigma \\ \tau &y \end{pmatrix} \;.
\end{displaymath}
Then the relevant statement from classical invariant theory in
conjunction with \cite{Luna1976} is this: given any $\mathrm{GL}_n
(\mathbb{C})$-invariant holomorphic function $f$ of the variables $Z,
\tilde{Z}, \zeta, \tilde{\zeta}$, it is possible to find a
holomorphic function $F$ of the variables $x, y, \sigma, \tau$ so
that
\begin{equation}\label{eq:F-over-f}
    F \begin{pmatrix} \tilde{Z} Z &\tilde{Z} \zeta \\
    \tilde{\zeta} Z &\tilde{\zeta} \zeta \end{pmatrix} =
    f(Z,\tilde{Z};\zeta,\tilde{\zeta})\;.
\end{equation}
To be sure, there exists no unique choice of such function $F$.
Indeed, since the top degree of the Grassmann algebra generated by
the anti-commuting variables $\zeta_e^i$ and $\tilde{\zeta}_i
^{e^\prime}$ is $2qn\,$, any monomial in the matrix variables $y$ of
degree higher than $qn$ vanishes identically upon making the
substitution $y_e^{e^\prime}= \tilde{\zeta}_i^{e^\prime}\zeta_e^i\,$.

In the following, we will use the abbreviated notation $F = F(Q)$
where the symbol $Q$ stands for the supermatrix built from the
matrices $x, \sigma, \tau, y:$
\begin{equation}\label{eq:def-Q}
    Q = \begin{pmatrix} x &\sigma\\ \tau &y \end{pmatrix} \;.
\end{equation}

\subsubsection{Orthogonal and symplectic symmetry}
\label{sect:orthsymp}

In the case of the symmetry group being $K = \mathrm{O}_n$ the
complex vector space $\mathbb{C}^n$ is equipped with a non-degenerate
symmetric tensor $\delta_{ij} = \delta_{j\,i}$ (which you may think
of as the Kronecker delta symbol). By definition, the elements $k$ of
the orthogonal group $\mathrm{O}_n$ satisfy the conditions $k^{-1} =
k^\dagger$ and $k^\mathrm{t} \delta k = \delta$ where $k^\mathrm{t}$
means the transpose of the matrix $k\,$. Let $\delta^{ij}$ denote the
components of the inverse tensor, $\delta^{-1}$. In addition to
$\tilde{Z} Z\,$, $\tilde{Z}\zeta\,$, $\tilde{\zeta} Z\,$,
$\tilde{\zeta}\zeta$ we now have the following independent quadratic
$K$-invariants:
\begin{eqnarray*}
    &&(Z^\mathrm{t}\delta Z)_{c^\prime c} = Z_{c^\prime}^i \delta_{ij}
    Z_c^j \;, \quad (\tilde{Z} \delta^{-1} \tilde{Z}^\mathrm{t}
    )^{c^\prime c}= \tilde{Z}_i^{c^\prime} \delta^{ij} \tilde{Z}_j^c\;,
    \\ &&(Z^\mathrm{t} \delta \zeta)_{c^\prime e} = Z_{c^\prime}^i
    \delta_{ij} \zeta_e^j \;, \quad (\tilde{\zeta} \delta^{-1}
    \tilde{Z}^\mathrm{t})^{e^\prime c} = \tilde{\zeta}_i^{e^\prime}
    \delta^{ij} \tilde{Z}_j^c \;, \\&&(\zeta^\mathrm{t} \delta \zeta
    )_{e^\prime e} = \zeta_{e^\prime}^i\delta_{ij} \zeta_e^j \;, \quad
    (\tilde{\zeta} \delta^{-1}\tilde{\zeta}^\mathrm{t})^{e^\prime e}
    = \tilde{\zeta}_i^{e^\prime}\delta^{ij} \tilde{\zeta}_j^e \;.
\end{eqnarray*}

In the case of symplectic symmetry, the dimension $n$ has to be an
even number and $\mathbb{C}^n$ is equipped with a non-degenerate
skew-symmetric tensor $\varepsilon_{ij} = - \varepsilon_{j\,i}\,$.
Elements $k$ of the unitary symplectic group $\mathrm{USp}_n$ satisfy
the conditions $k^{-1} = k^\dagger$ and $k^\mathrm{t} \varepsilon k =
\varepsilon\,$. If $\varepsilon^{ij} = - \varepsilon^{j\,i}$ are the
components of $\varepsilon^{-1}$, the extra quadratic invariants for
this case are
\begin{displaymath}
    (Z^\mathrm{t} \varepsilon Z)_{c^\prime c} = Z_{c^\prime}^i
    \varepsilon_{ij} Z_c^j \;, \quad (\tilde{Z} \varepsilon^{-1}
    \tilde{Z}^\mathrm{t})^{c^\prime c} = \tilde{Z}_i^{c^\prime}
    \varepsilon^{ij} \tilde{Z}_j^c \;, \quad \text{etc.}
\end{displaymath}

To deal with the two cases of orthogonal and symplectic symmetry in
parallel, we introduce the notation $\beta := \delta$ for $K =
\mathrm{O}_n$ and $\beta := \varepsilon$ for $K = \mathrm{USp}_n\,$,
and we organize all quadratic invariants as a supermatrix:
\begin{equation}\label{eq:invariants}
    \begin{pmatrix}
    \tilde{Z} Z &\tilde{Z}\beta^{-1}\tilde{Z}^\mathrm{t} &\tilde{Z}
    \zeta &\tilde{Z}\beta^{-1}\tilde{\zeta}^\mathrm{t} \\
    Z^\mathrm{t}\beta Z &Z^\mathrm{t} \tilde{Z}^\mathrm{t}
    &Z^\mathrm{t}\beta \zeta &Z^\mathrm{t}\tilde{\zeta}^\mathrm{t}
    \\ \tilde{\zeta}Z &\tilde{\zeta}\beta^{-1} \tilde{Z}^\mathrm{t}
    &\tilde{\zeta} \zeta &\tilde{\zeta}\beta^{-1} \tilde{\zeta}
    ^\mathrm{t} \\ -\zeta^\mathrm{t}\beta Z &-\zeta^\mathrm{t}
    \tilde{Z}^\mathrm{t} &-\zeta^\mathrm{t}\beta \zeta
    &-\zeta^\mathrm{t} \tilde{\zeta}^\mathrm{t} \end{pmatrix} \;.
\end{equation}
This particular matrix arrangement is motivated as follows.

Let $Q$ be the supermatrix (\ref{eq:def-Q}) of before, but now double
the size of each block; thus $x$ here is a matrix of size $2p \times
2p\,$, the rectangular matrix $\sigma$ is of size $2p \times 2q\,$,
and so on. Then impose on $Q$ the symmetry relation $Q = T_\beta
Q^\mathrm{st} (T_\beta)^{-1}$ where
\begin{displaymath}
    T_\delta = \begin{pmatrix}0 &1_p &{} &{}\\ 1_p &0 &{} &{}\\
    {} &{} &0 &-1_q\\{} &{} &1_q &0 \end{pmatrix}\;,
    \quad T_\varepsilon = \begin{pmatrix}
    0 &-1_p &{} &{}\\ 1_p &0 &{} &{}\\
    {} &{} &0 &1_q\\{} &{} &1_q &0 \end{pmatrix} \;,
\end{displaymath}
and $Q^\mathrm{st}$ means the supertranspose:
\begin{displaymath}
    Q^\mathrm{st} = \begin{pmatrix} x^\mathrm{t}&\tau^\mathrm{t}
    \\ - \sigma^\mathrm{t} &y^\mathrm{t} \end{pmatrix} \;.
\end{displaymath}
It is easy to check that the supermatrix (\ref{eq:invariants}) obeys
precisely this relation $Q T_\beta = T_\beta Q^\mathrm{st}$.

For the symmetry groups $K = \mathrm{O}_n$ and $K = \mathrm{USp}_n$
-- with the complexified groups being $G = \mathrm{O}_n(\mathbb{C})$
and $G = \mathrm{Sp}_n(\mathbb{C})$ -- it is still true that the
algebra of $G$-invariant holomorphic functions $f$ of $Z, \tilde{Z},
\zeta, \tilde{\zeta}$ is generated by the invariants that arise at
the quadratic level. Thus, if $f$ is any function of such kind, then
there exists (though not uniquely so) a holomorphic function $F(Q)$
which pulls back to the given function $f:$
\begin{equation}\label{eq:F-OVER-f}
    F\begin{pmatrix}
    \tilde{Z} Z &\tilde{Z}\beta^{-1}\tilde{Z}^\mathrm{t} &\tilde{Z}
    \zeta &\tilde{Z}\beta^{-1}\tilde{\zeta}^\mathrm{t} \\
    Z^\mathrm{t}\beta Z &Z^\mathrm{t} \tilde{Z}^\mathrm{t}
    &Z^\mathrm{t}\beta \zeta &Z^\mathrm{t}\tilde{\zeta}^\mathrm{t}
    \\ \tilde{\zeta}Z &\tilde{\zeta}\beta^{-1} \tilde{Z}^\mathrm{t}
    &\tilde{\zeta} \zeta &\tilde{\zeta} \beta^{-1} \tilde{\zeta}^
    \mathrm{t} \\ -\zeta^\mathrm{t}\beta Z &-\zeta^\mathrm{t}
    \tilde{Z}^\mathrm{t} &-\zeta^\mathrm{t} \beta \zeta
    &-\zeta^\mathrm{t} \tilde{\zeta}^\mathrm{t} \end{pmatrix} =
    f(Z,\tilde{Z};\zeta,\tilde{\zeta}) \;.
\end{equation}

\subsection{Superbosonization formula}
\label{sect:1.2}

A few more definitions are needed to state our main result, which
transfers the integral of $f$ to an integral of $F$.

In (\ref{eq:fund-int}) the definition of the Berezin integral $\int
f$ was given. Let us now specify how we integrate the `lifted'
function $F$, beginning with the case of $K = \mathrm{U}_n\,$. There,
the domain of integration will be $D = D_p^0 \times D_q^1$ where
$D_p^0$ is the symmetric space of positive Hermitian $p \times p$
matrices and $D_q^1$ is the group of unitary $q \times q$ matrices,
$D_q^1 = \mathrm {U}_q\,$. The Berezin superintegral form to be used
for $F(Q)$ is
\begin{equation}\label{eq:DQ-u}
    DQ := d\mu_{D_p^0}(x) \, d\mu_{D_q^1}(y)\, (2\pi)^{-pq}
    \Omega_{W_1} \circ \mathrm{Det}^q(x - \sigma y^{-1} \tau) \,
    \mathrm{Det}^p(y - \tau x^{-1} \sigma) \;,
\end{equation}
where the meaning of the various symbols is as follows. The Berezin
form $\Omega_{W_1}$ is defined as the product of all derivatives with
respect to the anti-commuting variables:
\begin{equation}\label{eq:1.10}
    \Omega_{W_1} = \prod_{c=1}^p \, \prod_{e=1}^q
    \frac{\partial^2}{\partial\sigma_e^c \, \partial\tau_c^e }\;.
\end{equation}
The symbol $d\mu_{D_q^1}$ denotes a suitably normalized Haar measure
on $D_q^1 = \mathrm{U}_q$ and $d\mu_{D_p^0}$ means a positive measure
on $D_p^0$ which is invariant with respect to the transformation $X
\mapsto g X g^\dagger$ for all invertible complex $p \times p$
matrices $g \in \mathrm{GL}_p (\mathbb{C})$. Our precise
normalization conventions for these measures are defined by the
Gaussian limits
\begin{displaymath}
    \lim_{t \to +\infty} \sqrt{t/ \pi}^{\, p^2} \int_{D_p^0}
    \mathrm{e}^{-t\, \mathrm{Tr}\, (x-\mathrm{Id})^2} d\mu_{D_p^0}(x)
    = 1 = \lim_{t \to +\infty} \sqrt{t / \pi}^{\, q^2} \int_{D_q^1}
    \mathrm{e}^{t\, \mathrm{Tr}\, (y-\mathrm{Id})^2} d\mu_{D_q^1}(y)\;.
\end{displaymath}

Now assume that $p \le n\,$. Then we assert that the
\emph{superbosonization formula}
\begin{equation}\label{bosonize}
    \int f = \frac{\mathrm{vol}(\mathrm{U}_n)}
    {\mathrm{vol}(\mathrm{U}_{n-p+q})} \int_D DQ \,\,
    \mathrm{SDet}^n(Q) \, F(Q)
\end{equation}
holds for a large class of analytic functions with suitable falloff
behavior at infinity. (In the body of the paper we state and prove
this formula for the class of Schwartz functions, i.e., functions
that decrease faster than any power. This, however, is not yet the
optimal formulation, and we expect the formula (\ref{bosonize}) to
hold in greater generality.) Here $\mathrm{vol}(\mathrm{U}_n) := \int
d\mu_{D_n^1}(y)$ is the volume of the unitary group, the integrands
$f$ and $F$ are assumed to be related by (\ref{eq:F-over-f}), and
$\mathrm{SDet}$ is the superdeterminant function,
\begin{displaymath}
    \mathrm{SDet} \begin{pmatrix} x &\sigma\\ \tau &y \end{pmatrix} =
    \frac{\mathrm{Det}(x)}{\mathrm{Det}(y - \tau x^{-1} \sigma)}\;.
\end{displaymath}

It should be mentioned at this point that ideas toward the existence
of such a formula as (\ref{bosonize}) have been vented in the recent
literature \cite{efetov-bos, guhr}. These publications, however, do
not give an answer to the important question of which integration
domain to choose for $Q\,$. Noting that the work of Efetov et al.\ is
concerned with the case of $n = 1$ and $p = q \gg 1$, let us
emphasize that the inequality $p \le n$ is in fact necessary in order
for our formula (\ref{bosonize}) to be true. (The situation for $p
> n$ is explored in a companion paper \cite{RBM-all}.) Moreover, be
advised that analogous formulas for the related cases of $K =
\mathrm{O}_n \,, \mathrm {USp}_n$ have not been discussed at all in
the published literature.

Turning to the latter two cases, we introduce two $2r \times 2r$
matrices $t_s$ and $t_a:$
\begin{displaymath}
    t_s = \begin{pmatrix} 0 &1_r \\ 1_r &0 \end{pmatrix} \;, \quad
    t_a = \begin{pmatrix} 0 &-1_r \\ 1_r &0 \end{pmatrix} \;,
\end{displaymath}
where $r = p$ or $r = q$ depending on the context. Then let a linear
space $\mathrm{Sym}_b (\mathbb{C}^{2r})$ for $b := s$ or $b := a$ be
defined by
\begin{displaymath}
    \mathrm{Sym}_b (\mathbb{C}^{2r}) := \left\{ M \in \mathrm{Mat}_{
    2r,2r}(\mathbb{C}) \mid M = t_b\,M^\mathrm{t}(t_b)^{-1}\right\}\;.
\end{displaymath}
Thus the elements of $\mathrm{Sym}_b (\mathbb{C}^{2r})$ are complex
$2r \times 2r$ matrices which are symmetric with respect to
transposition followed by conjugation with $t_b\,$. With this
notation, we can rephrase the condition $Q = T_\beta Q^\mathrm{st}
(T_\beta)^{-1}$ for the blocks $x$ resp.\ $y$ as saying they are in
$\mathrm{Sym}_s(\mathbb{C}^{2p})$ resp.\ $\mathrm{Sym}_a
(\mathbb{C}^{2q})$ for $\beta = \delta$ and in $\mathrm{Sym}_a
(\mathbb{C}^{2p})$ resp.\ $\mathrm{Sym}_s(\mathbb{C}^{2q})$ for
$\beta = \varepsilon\,$. The domain of integration for $Q$ will now
be $D_\beta := D_{\beta,p}^0 \times D_{\beta,\,q}^1$ where
\begin{displaymath}
    D_{\delta,p}^0 = \mathrm{Herm}^+ \cap \mathrm{Sym}_s
    (\mathbb{C}^{2p}) \;, \quad D_{\delta,\,q}^1 =
    \mathrm{U} \cap \mathrm{Sym}_a(\mathbb{C}^{2q}) \;,
\end{displaymath}
in the case of $\beta = \delta$ (or $K = \mathrm{O}_n$), and
\begin{displaymath}
    D_{\varepsilon,p}^0 = \mathrm{Herm}^+ \cap \mathrm{Sym}_a
    (\mathbb{C}^{2p}) \;, \quad D_{\varepsilon, \,q}^1 =
    \mathrm{U} \cap \mathrm{Sym}_s(\mathbb{C}^{2q}) \;,
\end{displaymath}
in the case of $\beta = \varepsilon$ (or $K = \mathrm{USp}_n$). Thus
in both cases, $\beta = \delta$ and $\beta = \varepsilon$, the
integration domains $D_{\beta,p}^0$ and $D_{\beta,\,q}^1$ are
constructed by taking the intersection with the positive Hermitian
matrices and the unitary matrices, respectively.

The Berezin superintegral form $DQ$ for the cases $\beta = \delta,
\varepsilon$ has the expression
\begin{equation}\label{eq:DQ-os}
    DQ := d\mu_{D_{\beta,p}^0}(x) \, d\mu_{D_{\beta,\,q}^1}(y) \,
    \Omega_{W_1} \circ \frac{\mathrm{Det}^q(x - \sigma y^{-1} \tau)\,
    \mathrm{Det}^p(y - \tau x^{-1} \sigma)}{(2\pi)^{2pq}\,
    \mathrm{Det}^{\frac{1}{2}m_\beta}(1-x^{-1}\sigma y^{-1}\tau)}\;,
\end{equation}
where $m_\delta = 1$ and $m_\varepsilon = -1$. The Berezin form
$\Omega_{W_1}$ now is simply a product of derivatives w.r.t.\ all of
the anti-commuting variables in the matrix $\sigma\,$. (The entries
of $\tau$ are determined from those of $\sigma$ by the relation $Q =
T_\beta Q^\mathrm{st} (T_\beta)^{-1}$). For $\beta = \delta$ one
defines
\begin{equation}\label{eq:1.12}
    \Omega_{W_1} = \prod_{c=1}^{p} \, \prod_{e=1}^q
    \frac{\partial^2}{\partial\sigma_e^c \, \partial\sigma_{e+q}^{c+p}}
    \; \prod_{c=1}^{p} \, \prod_{e=1}^q \frac{\partial^2}
    {\partial\sigma_e^{c+p} \, \partial\sigma_{e+q}^c}\;,
\end{equation}
while for $\beta = \varepsilon$ the definition is the same except
that the ordering of the derivatives $\partial/\partial\sigma_e^{c
+p}$ and $\partial / \partial\sigma_{e+q}^c$ in the second product
has to be reversed.

It remains to define the measures $d\mu_{D_{\beta,p}^0}$ and
$d\mu_{D_{ \beta,q}^1}\,$. To do so, we first observe that the
complex group $\mathrm{GL}_{2p}(\mathbb{C})$ acts on $\mathrm{Sym}_b
(\mathbb{C}^{2p})$ by conjugation in a twisted sense:
\begin{displaymath}
    x \mapsto g x \tau_b(g^{-1}) \;, \quad \tau_b(g^{-1})
    = t_b\, g^\mathrm{t} (t_b)^{-1} \quad (b = s, a)\;.
\end{displaymath}
A derived group action on the restriction to the positive Hermitian
matrices is then obtained by restricting to the subgroup $G^\prime
\subset \mathrm{GL}_{2p}(\mathbb{C})$ defined by the condition
\begin{displaymath}
    \tau_b(g^{-1}) = g^\dagger \;.
\end{displaymath}
This subgroup $G^\prime$ turns out to be $G^\prime \simeq
\mathrm{GL}_{2p}(\mathbb{R})$ for $b = s$ and $G^\prime \simeq
\mathrm{GL}_p(\mathbb{H})$, the invertible $p \times p$ matrices
whose entries are real quaternions, for $b = a\,$. In the sector of
$y$, the unitary group $\mathrm{U}_{2q}$ acts on $D_{\beta,\,q}^1 =
\mathrm{U} \cap \mathrm{Sym}_b(\mathbb{C}^{2q})$ by the same twisted
conjugation,
\begin{displaymath}
    y \mapsto g y \tau_b(g^{-1}) \quad (b = a, s) \;.
\end{displaymath}
Now in all cases, $d\mu_{D_{\beta,p}^0}$ and $d\mu_{D_{\beta,\,q}^1}$
are measures on $D_{\beta,p}^0$ and $D_{\beta,\,q}^1$ which are
invariant by the pertinent group action. Since the group actions at
hand are transitive, all of our invariant measures are unique up to
multiplication by a constant. As before, we consider a Gaussian limit
in order to fix the normalization constant:
\begin{displaymath}
    \lim_{t \to +\infty} \sqrt{t / \pi}^{\, p(2p + m_\beta)}
    \int_{D_{\beta,p}^0} \mathrm{e}^{-\frac{t}{2}\, \mathrm{Tr}
    \,(x-\mathrm{Id})^2} d\mu_{D_{\beta,p}^0}(x) = 1\;.
\end{displaymath}
The normalization of $d\mu_{D_{\beta,\,q}^1}$ is specified by the
corresponding formula where we make the replacements $p \to q$, and
$m_\beta \to - m_\beta$, and $- \mathrm{Tr}\,(x-\mathrm{Id})^2 \to +
\mathrm{Tr}\,(y-\mathrm{Id})^2$. An explicit expression for each of
these invariant measures is given in the Appendix.

We are now ready to state the superbosonization formula for the cases
of orthogonal and symplectic symmetry. Let the inequality of
dimensions $n \ge 2p$ be satisfied. We then assert that the following
is true.

Let the Berezin integral $\int f$ still be defined by
(\ref{eq:fund-int}), but now assume the holomorphically extended
integrand $f$ to be $G$-invariant with complexified symmetry group $G
= \mathrm{O}_n(\mathbb{C})$ for $\beta = \delta$ and $G =
\mathrm{Sp}_n (\mathbb{C})$ for $\beta = \varepsilon$. Let $K_n =
\mathrm{O}_n$ in the former case and $K_n = \mathrm{USp}_n$ in the
latter case. Then, choosing any holomorphic function $F(Q)$ related
to the given function $f$ by (\ref{eq:F-OVER-f}), the integration
formula
\begin{equation}\label{bos-other}
    \int f = 2^{(q-p)m_\beta}
    \frac{\mathrm{vol}(K_n)}{\mathrm{vol}(K_{n-2p+2q})}
    \int_{D_\beta} DQ \,\, \mathrm{SDet}^{n/2}(Q) \, F(Q)
\end{equation}
holds true, provided that $f$ falls off sufficiently fast at
infinity.

Thus the superbosonization formula takes the same form as in the
previous case $K = \mathrm{U}_n\,$, except that the exponent $n$ now
is reduced to $n/2\,$. The latter goes hand in hand with the size of
the supermatrix $Q$ having been expanded by $p \to 2p$ and $q \to
2q\,$.

Another remark is that the square root of the superdeterminant,
\begin{displaymath}
    \mathrm{SDet}^{n/2}(Q) = \sqrt{\mathrm{Det}^n(x^{\vphantom{-1}})}
    \, \big/ \, \sqrt{\mathrm{Det}^n(y - \tau x^{-1} \sigma)} \;,
\end{displaymath}
is always analytic in the sector of the matrix $y$. For the case of
orthogonal symmetry this is because $D_{\delta,\,q}^1 = \mathrm{U}
\cap \mathrm{Sym}_a(\mathbb{C}^{2q})$ is isomorphic to the unitary
skew-symmetric $2q \times 2q$ matrices and for such matrices the
determinant has an analytic square root known as the Pfaffian. (In
the language of random matrix physics, $D_{\delta,\,q}^1$ is the
domain of definition of the Circular Symplectic Ensemble, which has
the feature of Kramers degeneracy.) In the case of symplectic
symmetry, where the number $n$ is always even, no square root is
being taken in the first place.

As another remark, let us mention that each of our integration
domains is isomorphic to a symmetric space of compact or non-compact
type. These isomorphisms $D_q^1 \simeq \tilde{D}_q^1$ and $D_p^0
\simeq \tilde{D}_p^0$ are listed in Table \ref{fig:2}. Detailed
explanations are given in the main text.

\begin{table}
\begin{center}
\begin{tabular}{|p{0.8cm}|p{3.6cm}|p{2.7cm}|p{2.6cm}|p{1.8cm}|}
\hline
$K$ & $D_p^0$ & $\tilde{D}_p^0$ & $D_q^1$ & $\tilde{D}_q^1$ \\
\hline
  $\mathrm{U}_n$ \newline $\mathrm{O}_n$ \newline $\mathrm{USp}_n$
& $\mathrm{Herm}^+ \cap \mathrm{Mat}_{p,p}(\mathbb{C})$ \newline
  $\mathrm{Herm}^+ \cap \mathrm{Sym}_s(\mathbb{C}^{2p})$ \newline
  $\mathrm{Herm}^+ \cap \mathrm{Sym}_a(\mathbb{C}^{2p})$
& $\mathrm{GL}_p(\mathbb{C})/\mathrm{U}_p$ \newline
  $\mathrm{GL}_{2p}(\mathbb{R})/\mathrm{O}_{2p}$ \newline
  $\mathrm{GL}_p(\mathbb{H})/\mathrm{USp}_{2p}$
& $\mathrm{U} \cap \mathrm{Mat}_{q,\,q}(\mathbb{C})$ \newline
  $\mathrm{U} \cap \mathrm{Sym}_a(\mathbb{C}^{2q})$ \newline
  $\mathrm{U} \cap \mathrm{Sym}_s(\mathbb{C}^{2q})$
& $\mathrm{U}_{q}$ \newline
  $\mathrm{U}_{2q}/\mathrm{USp}_{2q}$ \newline
  $\mathrm{U}_{2q}/\mathrm{O}_{2q}$ \\
\hline
\end{tabular}\vskip 10pt
\caption{Isomorphisms between integration domains and symmetric
spaces.} \label{fig:2}
\end{center}
\end{table}
Let us also mention that the expressions (\ref{eq:DQ-u}) and
(\ref{eq:DQ-os}) for the Berezin integration forms $DQ$ can be found
from a supersymmetry principle: each $DQ$ is associated with one of
three Riemannian symmetric superspaces in the sense of \cite{suprev}
(to be precise, these are the supersymmetric non-linear sigma model
spaces associated with the random matrix symmetry classes $A$III,
$BD$I, and $C$II) and is in fact the Berezin integration form which
is invariant w.r.t.\ the action of the appropriate Lie superalgebra
$\mathfrak{gl}$ or $\mathfrak{osp}\,$. We will make no use of this
symmetry principle in the present paper. Instead, we will give a
direct proof of the superbosonization formulas (\ref{bosonize}) and
(\ref{bos-other}), deriving the expressions (\ref{eq:DQ-u}) and
(\ref{eq:DQ-os}) by construction, not from a supersymmetry argument.

Finally, we wish to stress that in random matrix applications, where
$n$ typically is a large number, the reduction brought about by the
superbosonization formulas (\ref{bosonize}) and (\ref{bos-other}) is
a striking advance: by conversion from its original role as the
number of integrations to do, the big integer $n$ has been turned
into an exponent, whereby asymptotic analysis of the integral by
saddle-point methods becomes possible.

\subsection{Illustration}

To finish this introduction, let us illustrate the new method at the
example of Wegner's $n$-orbital model with $n$ orbitals per site and
unitary symmetry.

The Hilbert space $V$ of that model is an orthogonal sum, $V =
\oplus_{i \in \Lambda} V_i\,$, where $i$ labels the sites (or
vertices) of a lattice $\Lambda$ and the $V_i \simeq \mathbb{C}^n$
are Hermitian vector spaces of dimension $n\,$. The Hamiltonians of
Wegner's model are random Hermitian operators $H : \, V \to V$
distributed according to a Gaussian measure $d\mu(H)$. To specify the
latter, let $\Pi_i : V \to V_i$ be the orthogonal projector on
$V_i\,$. The probability measure of Wegner's model is then given as a
Gaussian distribution $d\mu(H)$ with Fourier transform
\begin{displaymath}
    \int \mathrm{e}^{-\mathrm{i}\, \mathrm{Tr}\,(H K)} d\mu(H)
    = \mathrm{e}^{- \frac{1}{2n} \sum_{ij} C_{ij} \mathrm{Tr}\,
    (K \Pi_i K \Pi_j)} \;,
\end{displaymath}
where $K \in \mathrm{End}(V)$, and the variances $C_{ij} = C_{j\,i}$
are non-negative real numbers. We observe that $d\mu(H)$ is invariant
under conjugation $H \mapsto g H g^{-1}$ by unitary transformations
$g \in \prod_{i \in \Lambda} \mathrm{U}(V_i)$ (indeed, the Fourier
transform has the corresponding invariance under $K \mapsto g K
g^{-1}$); such an invariance is called a \emph{local gauge symmetry}
in physics.

Let us now be interested in, say, the average ratio of characteristic
polynomials:
\begin{displaymath}
    R(E_0\,,E_1) := \int \frac{\mathrm{Det}(E_1 - H)}{\mathrm{Det}
    (E_0 - H)} \, d\mu(H) \quad (\mathfrak{Im}\, E_0 > 0) \;.
\end{displaymath}
To compute $R(E_0\, , E_1)$ one traditionally uses a supersymmetry
method involving the so-called Hubbard-Stratonovich transformation.
In order for this approach to work, one needs to assume that the
positive quadratic form with matrix coefficients $C_{ij}$ has an
inverse. If it does, then the traditional approach leads to the
following result \cite{DPS}:
\begin{displaymath}
    R(E_0\,,E_1) = \int \mathrm{e}^{-\frac{n}{2} \sum_{ij}
    (C^{-1})_{ij} \, (x_i \, x_j - y_i\,  y_j)} D_\mathrm{HS}(x\,,y)
    \prod_{k \in \Lambda} \frac{(y_k-E_1)^{n-1}}{(x_k - E_0)^{n+1}}
    \, \frac{dy_k\,dx_k}{2\pi/\mathrm{i}} \;,
\end{displaymath}
where the integral is over $x_k \in \mathbb{R}$ and $y_k \in
\mathrm{i} \mathbb{R}\,$. The factor $D_\mathrm{HS} (x\, ,y)$ is a
fermion determinant resulting from integration over the
anti-commuting components of the Hubbard-Stratonovich field; it is
the determinant of the matrix with elements
\begin{displaymath}
    n\left(\delta_{ij}-(C^{-1})_{ij}\,(x_i-E_0)(y_j-E_1)\right)\;.
\end{displaymath}
Notice that the integration variables $x_k$ and $y_k$ carry the
physical dimension of energy.

In contrast, using the new approach opened up by the
superbosonization formula of the present paper, we obtain
\begin{displaymath}
    R(E_0\, , E_1) = \int \mathrm{e}^{-\frac{n}{2} \sum_{ij} C_{ij}
    \, (x_i \, x_j - y_i\, y_j)} D_\mathrm{SB}(x\, ,y) \prod_{k \in
    \Lambda}\left( \frac{x_k \,\mathrm{e}^{\mathrm{i}E_0 x_k}}{y_k\,
    \mathrm{e}^{\mathrm{i}E_1 y_k}}\right)^n \,
    \frac{dx_k \, dy_k}{2\pi\mathrm{i} \, x_k y_k} \;.
\end{displaymath}
Here the integral is over $x_k \in \mathbb{R}_+$ and $y_k \in
\mathrm{U}_1$ (the unit circle in $\mathbb{C}$). These integration
variables have the physical dimension of $(\mathrm{energy})^{-1}$.
The factor $D_\mathrm{SB}(x\, ,y)$ still is a fermion determinant,
which now arises from integration over the anti-commuting variables
of the superbosonization field; it is the determinant of the matrix
with elements
\begin{displaymath}
    n \left( \delta_{ij} + C_{ij}\, x_i \, y_j \right) \;.
\end{displaymath}
When both methods (Hubbard-Stratonovich and superbosonization) are
applicable, our two formulas for $R(E_0\, , E_1)$ are exactly
equivalent to each other. Please be warned, however, that this
equivalence is by no means easy to see directly.

From a practical viewpoint, the main difference between the two
formulas is that one of them is expressed by the quadratic form of
variances $C_{ij}$ whereas the other one is expressed by the
\emph{inverse} of that quadratic form. A rigorous analysis based on
the formula from Hubbard-Stratonovich transformation (or, rather, the
resulting formula for the density of states) for the case of
long-range $C_{ij}$ and $n = 1$, was made in \cite{DPS}. A similar
analysis based on the formula from superbosonization has not yet been
done.

\section{Motivation: supersymmetry method}
\label{sect:susy-method}\setcounter{equation}{0}

Imagine some quantum mechanical setting where the Hilbert space is
$\mathbb{C}^n$ equipped with its standard Hermitian structure. On
that finite-dimensional space, let us consider Hermitian operators
$H$ that are drawn at random from a probability distribution or
ensemble $d\mu(H)$. We might wish to compute the spectral correlation
functions of the ensemble or some other observable quantity of
interest in random matrix physics.

One approach to this problem is by the so-called supersymmetry method
\cite{efetov-book}. In that method the observables one wishes to
compute are written in terms of Green's functions, i.e., matrix
coefficients of the resolvent operator of $H$, which are then
expressed as Gaussian integrals over commuting and anti-commuting
variables.

The key object of the method is the characteristic function of
$d\mu(H)$,
\begin{displaymath}
    \mathcal{F}(K) = \int \mathrm{e}^{- \mathrm{i}\,
    \mathrm{Tr}\, (HK)} d\mu(H) \;,
\end{displaymath}
where the exact meaning of the Fourier variable $K$ depends on what
observable is to be computed. In the general case (with $p$ bosonic
and $q$ fermionic `replicas'), defining the matrix entries of $K$
with respect to the standard basis $\{ e_1 , \ldots, e_n \}$ of
$\mathbb{C}^n$ by $K e_j = e_i \, K_j^i\,$, one lets
\begin{displaymath}
    K_j^i := Z_c^i \tilde{Z}_j^c + \zeta_e^i \, \tilde{\zeta}_j^e
\end{displaymath}
where $Z_c^i\,$, $\tilde{Z}_j^c$ and $\zeta_e^i\,$, $\tilde{
\zeta}_j^e$ are the commuting and anti-commuting variables of Sect.\
\ref{sect:intro} and the summation convention is still in force. To
compute, say, the spectral correlation functions of $d\mu(H)$, one
multiplies $\mathcal{F}(K)$ by the exponential function
\begin{equation}\label{eq:energies}
    \exp \left(\mathrm{i} Z_c^i E_{c^\prime}^c \tilde{Z}_i^{c^\prime}
    + \mathrm{i}\zeta_e^j F_{e^\prime}^{e} \tilde{\zeta}_j^{e^\prime}
    \right) \;, \quad E_{c^\prime}^c = E_c \delta_{c^\prime}^c \;,
    \quad F_{e^\prime}^e = F_e \delta_{e^\prime}^e \;,
\end{equation}
where the parameters $E_c$ and $F_e$ have the physical meaning of
energies, and one integrates the product against the flat Berezin
integration form $D_{Z,\bar{Z}; \zeta, \tilde{\zeta}}$ over the real
vector space defined by $\tilde{Z}_i^c = \mathrm{sgn} (\mathfrak{Im}
\, E_c) \overline{Z_c^i}$ (for $c = 1, \ldots, p$ and $i = 1, \ldots,
n$). The desired correlation functions are then generated by a
straightforward process of taking derivatives with respect to the
energy parameters at coinciding points. Note that for all $g \in
\mathrm{GL}_n(\mathbb{C})$ the exponential (\ref{eq:energies}) is
invariant under
\begin{displaymath}
    Z_c^i \mapsto g_j^i Z_c^j \;, \quad \tilde{Z}_i^c \mapsto
    \tilde{Z}_j^c (g^{-1})_i^j \;, \quad \zeta_e^i \mapsto g_j^i
    \zeta_e^j \;, \quad \tilde{\zeta}_e^i \mapsto \tilde{\zeta}_e^j
    (g^{-1})_i^j \;.
\end{displaymath}

Let us now pass to a basis-free formulation of this setup. For that
we are going to think of the sets of complex variables $Z_c^i$ and
$\tilde{Z}_i^c$ as bases of holomorphic linear functions for the
complex vector spaces $\mathrm{Hom}(\mathbb{C}^p , \mathbb{C}^n)$
resp.\ $\mathrm{Hom}(\mathbb{C}^n , \mathbb{C}^p)$, and we interpret
the anti-commuting variables $\zeta_e^i$ and $\tilde{\zeta}_i^e$ as
generators for the exterior algebras of the vector spaces
$\mathrm{Hom}( \mathbb{C}^q , \mathbb{C}^n)^\ast$ resp.\ $
\mathrm{Hom}( \mathbb{C}^n , \mathbb{C}^q )^\ast$. Let
\begin{displaymath}
    V_0 := \mathrm{Hom}(\mathbb{C}^p , \mathbb{C}^n) \oplus
    \mathrm{Hom}(\mathbb{C}^n , \mathbb{C}^p) \;, \quad
    V_1 := \mathrm{Hom}(\mathbb{C}^q , \mathbb{C}^n) \oplus
    \mathrm{Hom}(\mathbb{C}^n , \mathbb{C}^q) \;.
\end{displaymath}
If we now choose some fixed Hermitian operator $H$ drawn from our
random matrix ensemble, the exponential $\mathrm{e}^{-\mathrm{i} \,
\mathrm{Tr}\, (HK)}$ is seen to be a holomorphic function on $V_0$
with values in the exterior algebra $\wedge( V_1^\ast)$. Under mild
assumptions on $d\mu(H)$ (e.g., bounded support, or rapid decay at
infinity) the holomorphic property carries over to the integral $\int
\exp(-\mathrm{i} \mathrm{Tr}\, HK)\, d\mu(H)$. The characteristic
function $\mathcal{F}(K)$ in that case is a holomorphic function
\begin{displaymath}
    \mathcal{F}(K)\, : \,\, V_0 \to \wedge (V_1^\ast) \;,
\end{displaymath}
and so is the function resulting from $\mathcal{F}(K)$ by
multiplication with the Gaussian factor (\ref{eq:energies}). We
denote this product of functions by $f$ for short.

Combining $V_0$ and $V_1$ to a $\mathbb{Z}_2$-graded vector space $V
:= V_0 \oplus V_1\,$, we denote the graded-commutative algebra of
holomorphic functions $V_0 \to \wedge(V_1^\ast)$ by $\mathcal{A}_V$.
The main task in the supersymmetry method is to compute the Berezin
superintegral of $f \in \mathcal{A}_V\,$.

This task is rather difficult to carry out for functions $f$
corresponding to a general probability measure $d\mu(H)$. Let us
therefore imagine that $f$ has some symmetries. Thus, let a group $G
\subset \mathrm{GL}(\mathbb{C}^n)$ be acting on $\mathbb{C}^n$, and
let
\begin{displaymath}
    g . (L \oplus \tilde{L}) = L \, g^{-1} \oplus g \tilde{L}
    \quad (g \in G) \;,
    \end{displaymath}
for $L \in \mathrm{Hom}(\mathbb{C}^n,\mathbb{C}^p)$ and $\tilde{L}
\in \mathrm{Hom}(\mathbb{C}^p,\mathbb{C}^n)$ be the induced action of
$G$ on $V_0\,$. We also have the same $G$-action on $V_1\,$, and the
latter induces a $G$-action on $\wedge(V_1^\ast)$.

Now let the given probability measure $d\mu(H)$ be in\-variant with
respect to conjugation by the elements of (a unitary form of) such a
group $G$. Via the Fourier transform, this symmetry gets transferred
to the characteristic function $\mathcal{F}(K)$, and also to the
product of $\mathcal{F}(K)$ with the exponential (\ref{eq:energies}).
Our function $f$ then satisfies the relation $f(v) = g . f(g^{-1} v)$
for all $g \in G$ and $v \in V_0$ and thus is an element of the
subalgebra $\mathcal{A}_V^G \subset \mathcal{A}_V$ of
$G$-\emph{equivariant} holomorphic functions.

Following Dyson \cite{dyson} the complex symmetry groups $G$ of prime
interest in random matrix theory are $G = \mathrm{GL}_n(\mathbb{C})$,
$\mathrm {O}_n(\mathbb{C})$, and $\mathrm{Sp}_n(\mathbb{C})$. These
are the complexifications of the compact symmetry groups
$\mathrm{U}_n\,$, $\mathrm{O}_n\,$, and $\mathrm{USp}_n\,$,
corresponding to ensembles of Hermitian matrices with unitary
symmetry, real symmetric matrices with orthogonal symmetry, and
quaternion self-dual matrices with symplectic symmetry.

To summarize, in the present paper we will be concerned with the
algebra $\mathcal{A}_V^G$ of $G$-equivariant holomorphic functions
\begin{equation}
    f \, : \,\, V_0 \to \wedge(V_1^\ast) \;, \quad
    v \mapsto f(v) = g . f(g^{-1} v) \quad (g \in G) \;,
\end{equation}
for the classical Lie groups $G = \mathrm{GL}_n(\mathbb{C})$,
$\mathrm{O}_n(\mathbb{C})$, and $\mathrm{Sp}_n(\mathbb{C})$, and the
vector spaces
\begin{eqnarray}
    &&V_0 = \mathrm{Hom}(\mathbb{C}^n , \mathbb{C}^p) \oplus
    \mathrm{Hom}(\mathbb{C}^p , \mathbb{C}^n) \;, \label{eq:2.3}\\
    &&V_1 = \mathrm{Hom}(\mathbb{C}^n , \mathbb{C}^q) \oplus
    \mathrm{Hom}(\mathbb{C}^q , \mathbb{C}^n) \;. \label{eq:2.4}
\end{eqnarray}

Our strategy will be to lift $f \in \mathcal{A}_V^G$ to another
algebra $\mathcal{A}_W$ of holomorphic functions $F : \, W_0 \to
\wedge(W_1^\ast)$, using a surjective homomorphism $\mathcal{A}_W \to
\mathcal{A}_V^G$. The thrust of the paper then is to prove a
statement of reduction -- the superbosonization formula --
transferring the Berezin superintegral of $f \in \mathcal{A}_V^G$ to
such an integral of $F \in \mathcal{A}_W\,$.

The advantage of our treatment (as compared to the orthogonal
polynomial method) is that it readily extends to the case of symmetry
groups $G \times G \times \ldots \times G$ with direct product
structure. This will make it possible in the future to treat such
models as Wegner's gauge-invariant model \cite{N-orbital} with $n$
orbitals per site and gauge group $G$.

\subsection{Notation}

We now fix some notation which will be used throughout the paper.

If $A$ and $B$ are vector spaces and $L : \, A \to B$ is a linear
mapping, we denote the canonical adjoint transformation between the
dual vector spaces $B^\ast$ and $A^\ast$ by $L^\mathrm{t} : \, B^\ast
\to A^\ast$. We call $L^\mathrm{t}$ the `transpose' of $L\,$. A
Hermitian structure $\langle \, , \, \rangle$ on a complex vector
space $A$ determines a complex anti-linear isometry $c_A : \, A \to
A^\ast$ by $v \mapsto \langle v , \cdot \rangle$. If both $A$ and $B$
carry Hermitian structure, then $L : \, A \to B$ has a Hermitian
adjoint $L^\dagger : \, B \to A$ defined by $L^\dagger = c_A^{-1}
\circ L^\mathrm{t} \circ c_B^{\vphantom{-1}}\,$. The operator
$(L^\dagger)^\mathrm{t} : \, A^\ast \to B^\ast$ is denoted by
$(L^\dagger)^\mathrm{t} \equiv \overline{L}\,$. Note $\overline{L} =
c_B^{\vphantom{-1}} \circ L \circ c_A^{-1}$. Finally, if each of $A$
and $B$ is equipped with a non-degenerate pairing $A \times A \to
\mathbb{C}$ and $B \times B \to \mathbb{C}\,$, so that we are given
complex linear isomorphisms $\alpha : \, A \to A^\ast$ and $\beta :
\, B \to B^\ast$, then there exists a transpose $L^T : \, B \to A$ by
\begin{displaymath}
    L^T \,:\,\, B \stackrel{\beta}{\to} B^\ast\stackrel{L^\mathrm{t}}
    {\to} A^\ast \stackrel{\alpha^{-1}}{\to} A \;.
\end{displaymath}

To emphasize that this is really a story about linear operators
rather than basis-dependent matrices, we use such notations as
$\mathrm{Hom}(\mathbb{C}^n , \mathbb{C}^p)$ for the vector space of
complex linear transformations from $\mathbb{C}^n$ to $\mathbb{C}^p$.
In a small number of situations we will resort to the alternative
notation $\mathrm{Mat}_{p,\,n} (\mathbb{C})$.

From here on in this article, Lie groups by default will be complex
Lie groups; thus $\mathrm{GL}_n \equiv \mathrm{GL}(\mathbb{C}^n)$,
$\mathrm{O}_n \equiv \mathrm{O}(\mathbb{C}^n)$, and $\mathrm{Sp}_n
\equiv \mathrm{Sp}(\mathbb{C}^n)$, with $n \in 2\mathbb{N}$ in the
last case.

\section{Pushing forward in the boson-boson sector}
\label{sect:BB-sector} \setcounter{equation}{0}

In this section, we shall address the special situation of $V_1 =
0\,$, or fermion replica number $q = 0\,$. Thus we are now facing the
commutative algebra $\mathcal{A}_V^G \equiv \mathcal{O}(V)^G$ of
$G$-invariant holomorphic functions on the complex vector space
\begin{displaymath}
    V \equiv V_0 = \mathrm{Hom}(\mathbb{C}^n , \mathbb{C}^p) \oplus
    \mathrm{Hom}(\mathbb{C}^p , \mathbb{C}^n) \;.
\end{displaymath}
In order to deal with this function space we will use the fact that
$\mathcal{O}(V)$ can be viewed as a completion of the symmetric
algebra $\mathrm{S}(V^\ast)$. Since the $G$-action on $\mathrm{S}
(V^\ast)$ preserves the $\mathbb{Z}$-grading
\begin{displaymath}
    \mathrm{S}(V^\ast) = \oplus_{k \ge 0}\, \mathrm{S}^k (V^\ast)
\end{displaymath}
and is reductive on each symmetric power $\mathrm{S}^k(V^\ast)$, one
has a subalgebra $\mathrm{S}^k (V^\ast)^G$ of $G$-fixed elements in
$\mathrm{S}^k(V^\ast)$ for all $k\,$.

\subsection{$G$-invariants at the quadratic level}
\label{sect:BB-invs}

It is a known fact of classical invariant theory (see, e.g.,
\cite{Howe1995}) that for each of the cases $G = \mathrm{GL}_n\,$,
$\mathrm{O}_n\,$, and $\mathrm {Sp}_n\,$, all $G$-invariants in
$\mathrm{S}(V^\ast)$ arise at the quadratic level, i.e., $\mathrm{S}
(V^\ast)^G$ is generated by $\mathrm{S}^2(V^\ast)^G$. Let us
therefore sharpen our understanding of these quadratic invariants
$\mathrm{S}^2(V^\ast)^G$.

\subsubsection{The case of $G = \mathrm{GL}_n\,$.}

All quadratic invariants are just of a single type here: they arise
by composing the elements of $\mathrm{Hom}(\mathbb{C}^p ,
\mathbb{C}^n)$ with those of $\mathrm{Hom}(\mathbb{C}^n ,
\mathbb{C}^p)$.
\begin{lem}\label{lem:BB-GL-invs}
$\mathrm{S}^2(V^\ast)^G$ is isomorphic as a complex vector space to
$W^\ast = \mathrm{End}(\mathbb{C}^p)^\ast$.
\end{lem}
\begin{proof}
Using the canonical transpose $\mathrm{Hom}(A,B) \simeq \mathrm{Hom}
(B^\ast , A^\ast)$ we have
\begin{displaymath}
  V \simeq \mathrm{Hom}(\mathbb{C}^n,\mathbb{C}^p) \oplus
  \mathrm{Hom}((\mathbb{C}^n)^\ast,(\mathbb{C}^p)^\ast) \;.
\end{displaymath}
For $G = \mathrm{GL}_n$ there exists no non-zero $G$-invariant tensor
in $\mathbb{C}^n \otimes \mathbb{C}^n$ or $(\mathbb{C}^n)^\ast
\otimes (\mathbb{C}^n)^\ast$. Therefore $\mathrm{S}^2 (\mathrm{Hom}
(\mathbb{C}^n,\mathbb{C}^p))^G = 0$ and $\mathrm{S}^2 (\mathrm{Hom}
((\mathbb{C}^n)^\ast , (\mathbb{C}^p)^\ast))^G = 0\,$, resulting in
\begin{displaymath}
    \mathrm{S}^2(V)^G \simeq \big(\mathrm{Hom}(\mathbb{C}^n ,
    \mathbb{C}^p) \otimes \mathrm{Hom}( (\mathbb{C}^n)^\ast,
    (\mathbb{C}^p)^\ast) \big)^G \;.
\end{displaymath}
The space of $G$-invariants in $(\mathbb{C}^n)^\ast \otimes
\mathbb{C}^n$ is one-dimensional (with generator $\varphi^i \otimes
e_i$ given by the canonical pairing between a vector space and its
dual). Since the action of $G$ on $\mathbb{C}^p$ is trivial, it
follows that
\begin{displaymath}
    \mathrm{S}^2 (V)^G \simeq \left( \mathbb{C}^p \otimes
    (\mathbb{C}^n)^\ast \otimes \mathbb{C}^n \otimes (\mathbb{C}^p)^\ast
    \right)^G \simeq \mathbb{C}^p \otimes (\mathbb{C}^p)^\ast
    \simeq \mathrm{End}(\mathbb{C}^p) \equiv W \;.
\end{displaymath}

The action of $G = \mathrm{GL}_n$ on $\mathrm{S}^2(V)$ and
$\mathrm{S}^2(V^\ast)$ is reductive. Therefore there exists a
canonical pairing $\mathrm{S}^2 (V)^G \otimes \mathrm{S}^2(V^\ast)^G
\to \mathbb{C}$ and the isomorphism $\mathrm{S}^2(V)^G \to W$
dualizes to an isomorphism $W^\ast \to (\mathrm{S}^2(V)^G)^\ast
\simeq \mathrm{S}^2(V^\ast)^G$.
\end{proof}

\subsubsection{The cases of $G = \mathrm{O}_n\,$, $\mathrm{Sp}_n\,$.}

Here $\mathbb{C}^n$ is equipped with a $G$-invariant non-degenerate
bilinear form or, equivalently, with a $G$-equivariant isomorphism
\begin{displaymath}
    \beta \, : \,\, \mathbb{C}^n \to (\mathbb{C}^n)^\ast \;,
\end{displaymath}
which is symmetric for $G = \mathrm{O}_n$ and alternating for $G =
\mathrm{Sp}_n\,$. To distinguish between these two, we sometimes
write $\beta = \delta$ in the former case and $\beta = \varepsilon$
in the latter case.

To describe $\mathrm{S}^2(V^\ast)^G$ for both cases, we introduce the
following notation. On $U := \mathbb{C}^p \oplus (\mathbb{C}^p)^\ast$
we have two canonical bilinear forms: the symmetric form
\begin{displaymath}
    s(v \oplus \varphi, v^\prime \oplus \varphi^\prime)
    = \varphi^\prime(v) + \varphi(v^\prime) \;,
\end{displaymath}
and the alternating form
\begin{displaymath}
    a(v \oplus \varphi, v^\prime \oplus \varphi^\prime)
    = \varphi^\prime(v) - \varphi(v^\prime) \;.
\end{displaymath}
\begin{defn}
Let $b = s$ or $b = a\,$. An endomorphism $L : \, U \to U$ of the
complex vector space $U = \mathbb{C}^p \oplus (\mathbb{C}^p)^\ast$ is
called symmetric with respect to $b$ if $L = L^T$, i.e.\ if
\begin{displaymath}
    b(L\, x\, ,y) = b(x\, ,L\,y)
\end{displaymath}
for all $x\, ,y \in U$. We denote the vector space of such
endomorphisms by $\mathrm{Sym}_b(U)$.
\end{defn}
\begin{lem}\label{lem:BB-OSP-invs}
If $U = \mathbb{C}^p \oplus (\mathbb {C}^p)^\ast$, then the space of
quadratic invariants $\mathrm{S}^2 (V^\ast)^G$ is isomorphic as a
complex vector space to $W^\ast$ where $W= \mathrm{Sym}_s(U)$ for $G=
\mathrm{O}_n$ and $W = \mathrm{Sym}_a(U)$ for $G = \mathrm{Sp}_n\,$.
\end{lem}
\begin{proof}
We still have $V \simeq \mathrm{Hom}(\mathbb{C}^p,\mathbb{C}^n)
\oplus \mathrm{Hom}((\mathbb{C}^p)^\ast,(\mathbb{C}^n)^\ast)$ but
now, via the given complex linear isomorphism $\beta : \mathbb{C}^n
\to (\mathbb{C}^n)^\ast$, we even have an identification
\begin{displaymath}
    V \simeq \mathrm{Hom}(U,\mathbb{C}^n)\simeq U^\ast \otimes
    \mathbb{C}^n\;,\quad U = \mathbb{C}^p \oplus (\mathbb{C}^p)^\ast \;.
\end{displaymath}
Also, letting $\mathrm{Sym}(V,V^\ast)$ denote the vector space of
symmetric linear transformations
\begin{displaymath}
    \sigma \,:\,\, V \to V^\ast \;, \quad \sigma(v)(v^\prime)
    = \sigma(v^\prime)(v) \;,
\end{displaymath}
there is an isomorphism $\mathrm{S}^2(V^\ast) \to
\mathrm{Sym}(V,V^\ast)$ by
\begin{displaymath}
    \varphi^\prime \varphi + \varphi \varphi^\prime
    \mapsto \left( v \mapsto \varphi^\prime(v)
    \varphi + \varphi(v) \varphi^\prime \right) \;.
\end{displaymath}
This descends to a vector space isomorphism between
$\mathrm{S}^2(V^\ast)^G$ and $\mathrm{Sym}_G(V,V^\ast)$, the
$G$-equivariant mappings in $\mathrm{Sym}(V,V^\ast)$.

Consider now $\mathrm{Hom}_G(V,V^\ast) \simeq \mathrm{Hom}_G (U^\ast
\otimes \mathbb{C}^n , U \otimes (\mathbb{C}^n)^\ast)$. As a
consequence of the $G$-action on $U$ and $U^\ast$ being trivial, one
immediately deduces that
\begin{displaymath}
    \mathrm{Hom}_G (U^\ast \otimes \mathbb{C}^n , U \otimes
    (\mathbb{C}^n)^\ast) \simeq \mathrm{Hom}(U^\ast , U) \otimes
    \mathrm{Hom}_G(\mathbb{C}^n , (\mathbb{C}^n)^\ast) \;.
\end{displaymath}
The vector space $\mathrm{Hom}_G(\mathbb{C}^n , (\mathbb{C}^n)^\ast)$
is one-dimensional with generator $\beta$. Because $\beta$ is
symmetric for $G = \mathrm{O}_n$ and alternating for $G =
\mathrm{Sp}_n\,$, it follows that
\begin{displaymath}
    \mathrm{S}^2(V^\ast)^G \simeq \mathrm{Sym}_G (U^\ast \otimes
    \mathbb{C}^n , U \otimes (\mathbb{C}^n)^\ast ) \simeq \left\{
    \begin{array}{ll} \mathrm{Sym}(U^\ast , U)\;, &\quad G =
    \mathrm{O}_n \;, \\ \mathrm{Alt}(U^\ast , U)\;, &\quad G =
    \mathrm{Sp}_n \;, \end{array} \right.
\end{displaymath}
where the notation $\mathrm{Alt}(U^\ast , U)$ means the vector space
of alternating homomorphisms $A : \, U^\ast \to U$, i.e.,
$\varphi(A(\varphi^\prime)) = - \varphi^\prime(A(\varphi))$. Note
that $\mathrm{Sym}(U^\ast,U) \simeq \mathrm{Sym}(U,U^\ast)^\ast$ and
$\mathrm{Alt}(U^\ast,U) \simeq \mathrm{Alt}(U,U^\ast)^\ast$ by the
trace form $(A,B) \mapsto \mathrm{Tr}\,(AB)$.

Now let $L\in \mathrm{Sym}_b(U)$. Since $b = s$ is symmetric and $b =
a$ is alternating, the image of $\mathrm{Sym}_b(U)$ in $\mathrm{Hom}
(U,U^\ast)$ under the mapping $L \mapsto \phi_L$ defined by
\begin{displaymath}
    \phi_L (x)(y) := b(L\, x \, , y)
\end{displaymath}
is $\mathrm{Sym}(U,U^\ast)$ for $b = s$ and $\mathrm{Alt}(U,U^\ast)$
for $b = a\,$. Moreover, since the bilinear form $b$ is
non-degenerate, this mapping is an isomorphism of $\mathbb{C}$-vector
spaces. Thus we have
\begin{displaymath}
    \mathrm{S}^2(V^\ast)^G \simeq \left\{ \begin{array}{cl}
    \mathrm{Sym}(U^\ast , U) \simeq \mathrm{Sym}(U,U^\ast)^\ast
    \simeq \mathrm{Sym}_s(U)^\ast \;, \quad &G = \mathrm{O}_n \;, \\
    \mathrm{Alt}(U^\ast , U) \simeq \mathrm{Alt}(U,U^\ast)^\ast
    \simeq \mathrm{Sym}_a(U)^\ast \;, \quad &G = \mathrm{Sp}_n \;, \\
    \end{array} \right.
\end{displaymath}
which is the statement that was to be proved.
\end{proof}

\subsection{The quadratic map $Q\,$}\label{sect:Q-map}

Summarizing the results of the previous subsection, the vector space
$W$ of quadratic $G$-invariants in $\mathrm{S}(V)$ is
 \begin{displaymath}
    W = \mathrm{S}^2(V)^G = \left\{ \begin{array}{ll}
    \mathrm{End}(\mathbb{C}^p)\;, &G = \mathrm{GL}_n \;, \\
    \mathrm{Sym}_s(U_p)\;, & G = \mathrm{O}_n \;, \\
    \mathrm{Sym}_a(U_p)\;, &G = \mathrm{Sp}_n\;, \end{array}\right.
 \end{displaymath}
where $U \equiv U_p = \mathbb{C}^p \oplus (\mathbb{C}^p)^\ast$. For
notational convenience, we will sometimes think of $W =
\mathrm{End}(\mathbb{C}^p) \hookrightarrow \mathrm{End}(\mathbb{C}^p)
\oplus \mathrm{End}((\mathbb{C}^p)^\ast)$ for the first case ($G =
\mathrm{GL}_n$) as the intersection $\mathrm{Sym}_s(U_p) \cap
\mathrm{Sym}_a(U_p)$ of the vector spaces $W$ for the last two cases
($G = \mathrm{O}_n \, , \mathrm{Sp}_n$).

In the following we will repeatedly use the decomposition of elements
$w \in W$ as
\begin{equation}\label{eq:decompose-w}
    w = \begin{pmatrix} A &B \\ C &A^\mathrm{t} \end{pmatrix}
    \in \begin{pmatrix} \mathrm{End}(\mathbb{C}^p)
    &\mathrm{Hom}( (\mathbb{C}^p)^\ast, \mathbb{C}^p) \\
    \mathrm{Hom}( \mathbb{C}^p , (\mathbb{C}^p)^\ast)
    &\mathrm{End}((\mathbb{C}^p)^\ast) \end{pmatrix} \;.
\end{equation}
Note that $B$ and $C$ are symmetric for the case of $G = \mathrm
{O}_n$ and alternating for $G= \mathrm{Sp}_n\,$. The case of $G =
\mathrm{GL}_n$ is included by setting $B = C = 0\,$. Note also the
dimensions $\mathrm{dim}\, W = p^2$, $p(2p+1)$, and $p(2p-1)$ for $G
= \mathrm{GL}_n\,$, $\mathrm{O}_n\,$, and $\mathrm{Sp}_n\,$, in this
order.

Our treatment below is based on the relationship of $\mathcal{O}
(V)^G$ with the holomorphic functions $\mathcal{O}(W)$. To make this
relation explicit, we now introduce a map
\begin{displaymath}
    Q \, : \,\, V = \mathrm{Hom}(\mathbb{C}^n , \mathbb{C}^p) \oplus
    \mathrm{Hom}(\mathbb{C}^p , \mathbb{C}^n) \to \mathrm{End}(U_p)
\end{displaymath}
by defining its blocks according to the decomposition
(\ref{eq:decompose-w}) as
\begin{displaymath}
    Q \,: \,\, L \oplus \tilde{L} \mapsto \begin{pmatrix} L\,
    \tilde{L} &L \, \beta^{-1} L^\mathrm{t}\\ \tilde{L}^\mathrm{t}
    \beta\tilde{L} &\tilde{L}^\mathrm{t} L^\mathrm{t}\end{pmatrix}\;.
\end{displaymath}
Recall that the $G$-equivariant isomorphism $\beta:\,\mathbb{C}^n \to
(\mathbb{C}^n )^\ast$ is symmetric for $G = \mathrm{O}_n\,$,
alternating for $G = \mathrm{Sp}_n\,$, and non-existent for $G =
\mathrm{GL}_n$ in which case the off-diagonal blocks $L \beta^{-1}
L^\mathrm{t}$ and $\tilde{L}^\mathrm{t} \beta \tilde{L}$ are
understood to be zero. In all three cases this mapping $Q$ is
$G$-invariant: $Q(L \oplus \tilde{L}) = Q(L g^{-1} \oplus g
\tilde{L})$ for all $g \in G$. In the last two cases this is because
$g^\mathrm{t} \beta g = \beta$ by the very notion of what it means
for $\beta$ to be $G$-equivariant.
\begin{lem}\label{lem:Q-into-W}
The $G$-invariant mapping $Q :\, V \to \mathrm{End}(U)$ is into $W$.
\end{lem}
\begin{proof}
Let $G$ be one of the groups $\mathrm{O}_n$ or $\mathrm{Sp}_n$ and
denote by $L$ and $\tilde{L}$ the elements of $\mathrm{Hom}
(\mathbb{C}^n , \mathbb{C}^p)$ resp.\ $\mathrm{Hom}(\mathbb{C}^p ,
\mathbb{C}^n)$. Introducing two isomorphisms
\begin{eqnarray*}
    &&\psi \, : \,\, V \to \mathrm{Hom}(\mathbb{C}^n , U) \;,
    \quad L \oplus \tilde{L} \mapsto L \oplus \tilde{L}^\mathrm{t}
    \beta \;, \\ &&\tilde{\psi} \, : \,\, V \to \mathrm{Hom}
    (U,\mathbb{C}^n) \;, \quad L \oplus \tilde{L} \mapsto \tilde{L}
    \oplus \beta^{-1} L^\mathrm{t} \;,
\end{eqnarray*}
we have $Q(v) = \psi(v) \tilde{\psi}(v)$ for $v = L \oplus
\tilde{L}\,$. The two maps $\psi$ and $\tilde{\psi}$ are related by
\begin{displaymath}
    \beta \tilde{\psi}(v) = \psi(v)^\mathrm{t} T_b
\end{displaymath}
where $T_b : \, U \to U^\ast$ is the isomorphism given by $x \mapsto
b(x\,,\cdot)$. Note that $T_b$ is symmetric for $b = s$ and
alternating for $b = a\,$. Using the relations above, one computes
that
\begin{displaymath}
    Q(v)^\mathrm{t} = \tilde{\psi}(v)^\mathrm{t} \psi(v)^\mathrm{t} =
    T_b^\mathrm{t} \psi(v) (\beta^\mathrm{t})^{-1} \beta \tilde{\psi}(v)
    T_b^{-1} \;.
\end{displaymath}
If parities $\sigma(\beta), \sigma(T_b) \in \{ \pm 1 \}$ are assigned
to $\beta$ and $T_b$ by $\beta^\mathrm{t} = \sigma(\beta)\, \beta$
and $T_b^\mathrm{t} = \sigma (T_b)\, T_b\,$, then $\sigma(\beta) =
\sigma(T_b)$ by construction, and it follows that $Q(v)^\mathrm{t} =
T_b Q(v) (T_b)^{-1}$. This is equivalent to saying that $Q(v) =
Q(v)^T \in \mathrm{Sym}_b (U) = W$, which proves the statement for
the groups $G = \mathrm{O}_n \, , \mathrm {Sp}_n\,$. The remaining
case of $G = \mathrm{GL}_n$ is included as a subcase by the embedding
$\mathrm{End}(\mathbb{C}^p) \hookrightarrow \mathrm {Sym}_s(U_p) \cap
\mathrm{Sym}_a(U_p)$.
\end{proof}
While the map $Q : \, V \to W$ will not always be surjective, as the
rank of $L \in \mathrm{Hom}(\mathbb{C}^n , \mathbb{C}^p)$ is at most
$\mathrm{min}(n\,,p)$, there is a pullback of algebras $Q^\ast : \,
\mathcal{O}(W) \to \mathcal{O}(V)^G$ in all cases. Let us now look
more closely at $Q^\ast$ restricted to $W^\ast$, the linear functions
on $W$. For this let $\{ e_i \}$, $\{ f^i \}$, $\{ e_c \}$, and $\{
f^c \}$ be standard bases of $\mathbb{C}^n$, $(\mathbb{C}^n )^\ast$,
$\mathbb{C}^p$, and $(\mathbb{C}^p)^\ast$, respectively, and define
bases $\{ Z_c^i \}$ and $\{\tilde{Z}_i^c \}$ of $\mathrm{Hom}
(\mathbb{C}^p , \mathbb{C}^n)^\ast$ and $\mathrm{Hom} (\mathbb{C}^n ,
\mathbb{C}^p)^\ast$ by
\begin{displaymath}
    Z_c^i(\tilde{L}) = f^i(\tilde{L} \, e_c) \;, \quad
    \tilde{Z}_i^c(L) = f^c(L \, e_i) \;,
\end{displaymath}
where $i = 1, \ldots, n$ and $c = 1, \ldots, p\,$. Also, decomposing
$w \in W \subset \mathrm{End}(U)$ as in (\ref{eq:decompose-w}),
define a set of linear functions $x_c^{c^\prime}$, $y^{c^\prime c}$,
and $y_{c c^\prime}$ on $W$ by
\begin{displaymath}
    x_c^{c^\prime}(w) = f^{c^\prime}(A e_c) \;, \quad
    y^{c^\prime c}(w) = f^{c^\prime}(B f^c) \;, \quad
    y_{c c^\prime}(w) = (C e_{c^\prime})(e_c) \;.
\end{displaymath}
Notice that $y^{c^\prime c} = \pm y^{c c^\prime}$ and $y_{c c^\prime}
= \pm y_{c^\prime c}$ where the plus sign applies in the case of $G =
\mathrm{O}_n$ and the minus sign for $G = \mathrm{Sp}_n\,$. The set
of functions $\{ x_c^{c^\prime} \}$ is a basis of $W^\ast \simeq
\mathrm{End}(\mathbb{C}^p)^\ast$ for the case of $G = \mathrm{GL}_n
\,$. Expanding this set by including the set of functions $\{
y^{c^\prime c} , y_{c c^\prime} \}_{c \le c^\prime}$ we get a basis
of $W^\ast$ for $G = \mathrm{O}_n\,$. The same goes for $G =
\mathrm{Sp}_n$ if the condition on indices $c \le c^\prime$ is
replaced by $c < c^\prime$.
\begin{lem}
The pullback of algebras $Q^\ast : \, \mathcal{O}(W) \to \mathcal{O}
(V)^G$ restricted to the linear functions on $W$ realizes the
isomorphism of complex vector spaces $W^\ast \to \mathrm{S}^2
(V^\ast)^G$.
\end{lem}
\begin{proof}
Applying $Q^\ast$ to the chosen basis of $W^\ast$ we obtain the
expressions
\begin{displaymath}
    Q^\ast x_c^{c^\prime} = \tilde{Z}_i^{c^\prime} Z_c^i \;,
    \quad Q^\ast y^{c^\prime c} = \tilde{Z}_i^{c^\prime} \beta^{ij}
    \tilde{Z}_j^c \;, \quad
    Q^\ast y_{c c^\prime} = Z_c^i \beta_{ij} Z_{c^\prime}^j \;,
\end{displaymath}
where $\beta^{ij} = f^i(\beta^{-1} f^j)$ and $\beta_{ij} = (\beta
e_j)(e_i)$. Now the $p^2$ functions $\tilde{Z}_i^{c^\prime} Z_c^i$
are linearly independent and form a basis of $\mathrm{S}^2
(V^\ast)^{\mathrm{GL}_n}$. By including the $p(p \pm 1)$ linearly
independent functions $\tilde{Z}_i^{c^\prime} \beta^{ij}
\tilde{Z}_j^c$ and $Z_c^i \beta_{ij} Z_{c^\prime}^j\,$ we get a basis
of $\mathrm{S}^2 (V^\ast)^G$ for $G = \mathrm{O}_n$ resp.\
$\mathrm{Sp}_n\,$. Thus our linear map $Q^\ast : W^\ast \to
\mathrm{S}^2(V^\ast)^G$ is a bijection in all cases.
\end{proof}
\begin{prop}\label{lem:Q-surj}
The homomorphism $Q^\ast : \, \mathcal{O}(W) \to \mathcal{O}(V)^G$ is
surjective.
\end{prop}
\begin{proof}
Let $\mathbb{C}[W] = \mathrm{S}(W^\ast)$ and $\mathbb{C}[V]^G =
\mathrm{S}(V^\ast)^G$ be the rings of polynomial functions on $W$ and
$G$-invariant polynomial functions on $V$, respectively. Pulling back
functions by the $G$-invariant quadratic map $Q : \, V \to W$, we
have a homomorphism $Q^\ast : \, \mathbb{C}[W] \to \mathbb{C}[V]^G$.
This map $Q^\ast$ is surjective because $\mathbb{C}[V]^G = \mathrm{S}
(V^\ast)^G$ is generated by $\mathrm{S}^2(V^\ast)^G$ and $Q^\ast : \,
W^\ast \to \mathrm{S}^2(V^\ast)^G$ is an isomorphism.

Our holomorphic functions are expressed by power series with infinite
radius of convergence. Therefore the surjective property of $Q^\ast :
\, \mathbb{C}[W] \to \mathbb{C}[V]^G$ carries over to $Q^\ast : \,
\mathcal{O}(W) \to \mathcal{O}(V)^G$.
\end{proof}

In the sequel, we will establish a finer result, relating the
integral of an integrable function $Q^\ast F \in \mathcal{O}(V)^G$
along a real subspace of $V$ to an integral of $F \in \mathcal{O}(W)$
over a non-compact symmetric space in $W$. While Prop.\
\ref{lem:Q-surj} applies always, this relation between integrals
depends on the relative value of dimensions and will here be
developed only in the range $n \ge p$ (for $G = \mathrm{GL}_n$) or $n
\ge 2p$ (for $G = \mathrm{O}_n\,, \mathrm{Sp}_n$).

We begin by specifying the integration domain in $V$. Using the
standard Hermitian structures of $\mathbb{C}^n$ and $\mathbb{C}^p$,
let a real subspace $V_\mathbb{R} \subset V$ be defined as the graph
of
\begin{displaymath}
    \dagger : \, \mathrm{Hom}(\mathbb{C}^n , \mathbb{C}^p) \to
    \mathrm{Hom}(\mathbb{C}^p , \mathbb{C}^n) \;.
\end{displaymath}
Thus in order for $L \oplus \tilde{L} \in V$ to lie in $V_\mathbb{R}$
the linear transformation $\tilde{L} : \, \mathbb{C}^p \to
\mathbb{C}^n$ has to be the Hermitian adjoint $(\tilde{L} =
L^\dagger)$ of $L : \mathbb{C}^n \to \mathbb{C}^p$. Note that
$V_\mathbb{R} \simeq \mathrm{Hom}(\mathbb{C}^n , \mathbb{C}^p)$.

The real vector space $V_\mathbb{R}$ is endowed with a Euclidean
structure by the norm square
\begin{displaymath}
    \vert\vert L \oplus L^\dagger \vert\vert^2 :=
    \mathrm{Tr}(L\, L^\dagger) \;.
\end{displaymath}
Let then $\mathrm{dvol}_{ V_\mathbb{R}}$ denote the canonical volume
density of this Euclidean vector space $V_\mathbb{R}\,$. Our interest
will be in the integral over $V_\mathbb{R}$ of $f \, \mathrm{dvol}_{
V_\mathbb{R}}$ for $f \in \mathcal{O}(V)^G$. To make sure that the
integral exists, we will assume that $f$ is a Schwartz function along
$V_\mathbb{R}\,$.

Note that the anti-linear bijection $c_p : \, \mathbb{C}^p \to
(\mathbb{C}^p)^\ast$, $v \mapsto \langle v\, , \cdot \rangle$,
determines a Hermitian structure on $(\mathbb{C}^p)^\ast$ by $\langle
\varphi , \varphi^\prime \rangle := \langle c_p^{-1} \varphi^\prime ,
c_p^{-1} \varphi \rangle$. The canonical Hermitian structure of $U =
\mathbb{C}^p \oplus (\mathbb{C}^p)^\ast$ is then given by the sum
$\langle u \oplus \varphi , u^\prime \oplus \varphi^\prime \rangle =
\langle u , u^\prime \rangle + \langle \varphi , \varphi^\prime
\rangle$.

The following is a first step toward our goal of transferring the
integral $\int_{V_\mathbb{R}} f \, \mathrm{dvol}_{V_\mathbb{R}}$ to
an integral over a non-compact symmetric space in $W$.
\begin{lem}
The image of $V_\mathbb{R}$ under the quadratic map $Q$ lies in the
intersection of $W$ and the non-negative Hermitian operators. Thus
\begin{displaymath}
    Q(V_\mathbb{R}) \subset \left\{ \begin{array}{ll}
    \mathrm{Herm}^{\ge 0} \cap \mathrm{End}(\mathbb{C}^p) \;,
    &\quad G = \mathrm{GL}_n \;, \\
    \mathrm{Herm}^{\ge 0} \cap \mathrm{Sym}_s(\mathbb{C}^p
    \oplus (\mathbb{C}^p)^\ast) \;, &\quad G = \mathrm{O}_n \;, \\
    \mathrm{Herm}^{\ge 0} \cap \mathrm{Sym}_a(\mathbb{C}^p \oplus
    (\mathbb{C}^p)^\ast)\;, &\quad G = \mathrm{Sp}_n \;.
    \end{array} \right.
\end{displaymath}
\end{lem}
\begin{proof}
In the first case this is immediate from $Q(L \oplus L^\dagger) = L
\, L^\dagger = (L\, L^\dagger)^\dagger \ge 0\,$. To deal with the
other two cases we recall the expression
\begin{displaymath}
    Q(L \oplus L^\dagger) = \begin{pmatrix} L\, L^\dagger
    &L \, \beta^{-1} L^\mathrm{t}\\ \overline{L} \beta
    L^\dagger &\overline{L} \, L^\mathrm{t} \end{pmatrix} \;.
\end{displaymath}
We already know from Lemma \ref{lem:Q-into-W} that $Q(L \oplus
L^\dagger)\in \mathrm{Sym}_b(\mathbb{C}^p\oplus (\mathbb{C}^p)^\ast)$
where $b = s$ or $b = a\,$. The operator $Q(L \oplus L^\dagger)$ is
self-adjoint because $(L^\mathrm{t})^\dagger = \overline{L}$ and
$\beta^\dagger = \beta^{-1}$. It is non-negative because $\langle u
\oplus \varphi , Q(L\oplus L^\dagger) (u \oplus \varphi) \rangle =
\vert L^\dagger u +\beta^{-1} L^\mathrm{t} \varphi \vert^2 \ge 0\,$.
\end{proof}
\begin{rem}
The condition $n \ge p$ resp.\ $n \ge 2p$ emerging below, can be
anticipated as the condition for the $Q$-image of a generic element
in $V_\mathbb{R}$ to have full rank.
\end{rem}

\subsection{The symmetric space of regular $K$-orbits in $V_\mathbb{R}\;$}

Recall that our groups $G$ act on $V$ by $g .(L \oplus \tilde{L})=
L\, g^{-1} \oplus g \tilde{L}\,$. By the relation $(L \, g^{-1}
)^\dagger = g L^\dagger$ for unitary transformations $g \in G\,$, the
$G$-action on $V$ restricts to an action on $V_\mathbb{R}$ by the
unitary subgroup $K = \mathrm{U}_n\,$, $\mathrm{O}_n(\mathbb{R})$, or
$\mathrm{USp}_n\,$, of $G = \mathrm{GL}_n\,$, $\mathrm{O}_n
(\mathbb{C})$, resp.\ $\mathrm{Sp}_n\,$.

In this subsection we study the regular $K$-orbit structure of
$V_\mathbb{R}\,$. For this purpose we identify $V_\mathbb{R} \simeq
\mathrm{Hom}(\mathbb{C}^n , \mathbb{C}^p)$ by the $K$-equivariant
isomorphism given by $L \oplus L^\dagger \mapsto L\,$.

\subsubsection{$K = \mathrm{U}_n\,$.}

Here and elsewhere let $\mathrm{Hom}^\prime(A,B)$ denote the space of
homomorphisms of maximal rank between two vector spaces $A$ and $B$.
\begin{lem}\label{lem:3.8}
If $n \ge p$ then $\mathrm{Hom}^\prime(\mathbb{C}^n , \mathbb{C}^p) /
\mathrm{U}_n \simeq \mathrm{GL}_p / \mathrm{U}_p\,$ (diffeomorphism).
\end{lem}
\begin{proof}
Since a regular transformation $L \in \mathrm{Hom}^\prime
(\mathbb{C}^n , \mathbb{C}^p)$ is surjective, the space $\mathrm{im}
(L^\dagger)$ has dimension $p\,$. Thus the decomposition
$\mathbb{C}^n = \mathrm{ker}(L) \oplus \mathrm{im}(L^\dagger)$
defines an element of the Grassmannian $(\mathrm{U}_p \times
\mathrm{U}_{n-p}) \setminus \mathrm{U}_n$ of complex $p$-planes in
$\mathbb{C}^n$. Fixing some unitary basis of $\mathrm{im}
(L^\dagger)$, we can identify the restriction $L : \mathrm{im}
(L^\dagger) \to \mathbb{C}^p$ with an element of $\mathrm{GL}_p\,$.
In other words,
\begin{displaymath}
    \mathrm{Hom}^\prime(\mathbb{C}^n , \mathbb{C}^p) \simeq
    \mathrm{GL}_p \times_{\mathrm{U}_p} (\mathrm{U}_{n-p} \backslash
    \mathrm{U}_n ) \;,
\end{displaymath}
which gives the desired statement by taking the quotient by the right
$\mathrm{U}_n$-action.
\end{proof}

\subsubsection{$K = \mathrm{O}_n\;$}

To establish a similar result for the case of orthogonal symmetry, we
need the following preparation. (Here and in the remainder of this
subsection $\mathrm{O}_n \equiv \mathrm{O}_n(\mathbb{R})$ means the
real orthogonal group.) Recalling that we are given a symmetric
isomorphism $\delta : \, \mathbb{C}^n \to (\mathbb{C}^n)^\ast$, we
associate with $L \in \mathrm{Hom} (\mathbb{C}^n , \mathbb{C}^p)$ an
extended complex linear operator $\psi(L) \in \mathrm{Hom}
(\mathbb{C}^n , \mathbb{C}^p \oplus (\mathbb{C}^p)^\ast)$ by
\begin{displaymath}
    \psi(L) v = (L\, v) \oplus \overline{L}\, \delta v \;.
\end{displaymath}
\begin{lem}
The mapping $\psi : \, L \mapsto L \oplus (\overline{L} \circ
\delta)$ determines an $\mathrm{O}_n$-equivariant isomorphism
$\mathrm{Hom}(\mathbb{C}^n , \mathbb{C}^p) \to \mathrm{Hom}
(\mathbb{R}^n , \mathbb{R}^{2p})$ of vector spaces with complex
structure.
\end{lem}
\begin{proof}
Recall that $\overline{L} = c_p \circ L \circ c_n^{-1}$ where $c_p
:\, \mathbb{C}^p \to (\mathbb{C}^p)^\ast$ and $c_n :\, \mathbb{C}^n
\to (\mathbb{C}^n)^\ast$ are the canonical anti-linear isomorphisms
given by the Hermitian structures of $\mathbb{C}^p$ resp.\
$\mathbb{C}^n$. Writing $T := c_n^{-1} \delta$ we have
\begin{displaymath}
    \psi(L) = L \oplus (c_p \, L \, T) \;.
\end{displaymath}
Because $\delta$ is a symmetric isomorphism, the anti-unitary
operator $T : \mathbb{C}^n \to \mathbb{C}^n$ squares to $T^2 = 1\,$.
Let $\mathrm{Fix}(T) \subset \mathbb{C}^n$ denote the real subspace
of fixed points $v = Tv$, and define on $U = \mathrm{C}^p \oplus
(\mathbb{C}^p)^\ast$ an anti-linear involution $C$ by
\begin{displaymath}
    C (u \oplus \varphi) := (c_p^{-1} \varphi) \oplus c_p \, u \;.
\end{displaymath}
If $\mathrm{Fix}(C) \subset U$ denotes the real subspace of fixed
points of $C$, then from
\begin{displaymath}
    C\psi(L)v = L\, T v \oplus c_p\, L v \stackrel{v=Tv}{=} \psi(L)v
\end{displaymath}
we see that the $\mathbb{C}$-linear operator $\psi(L)$ maps
$\mathrm{Fix}(T) \simeq \mathbb{R}^n$ into $\mathrm{Fix}(C) \simeq
\mathbb{R}^{2p}$. Thus we may identify $\psi(L)$ with an element of
$\mathrm{Hom}(\mathbb{R}^n , \mathbb{R}^{2p})$. The correspondence $L
\mapsto \psi(L)$ is bijective and transforms multiplication by
$\sqrt{-1}$, $L \mapsto \mathrm{i} L\,$, into $\psi(L) \mapsto J
\psi(L)$ where $J : \, u + c_p\, u \mapsto \mathrm{i} u - \mathrm{i}
c_p \, u$ is the complex structure of the real vector space
$\mathbb{R}^{2p} \simeq \mathrm{Fix}(C)$.

By definition, the elements of the real orthogonal group
$\mathrm{O}_n$ commute with $T$. Thus $\psi(L) \circ k = \psi(L \,
k)$ for $k \in \mathrm{O}_n\,$, which means that $\psi$ is
$\mathrm{O}_n$-equivariant.
\end{proof}
As before, let $U = \mathbb{C}^p \oplus (\mathbb{C}^p)^\ast$ be
equipped with the Hermitian structure which is induced from that of
$\mathbb{C}^p$ by $\langle u \oplus \varphi, u^\prime \oplus
\varphi^\prime \rangle = \langle u , u^\prime \rangle + \langle
c_p^{-1} \varphi^\prime , c_p^{-1} \varphi \rangle$. Its restriction
to $\mathrm{Fix}(C) \simeq \mathbb{R}^{2p}$ is a Euclidean structure
defining the real orthogonal group $\mathrm{O}_{2p}\;$.
\begin{lem}
If $n\ge 2p$ then $\mathrm{Hom}^\prime(\mathbb{R}^n,\mathbb{R}^{2p})/
\mathrm{O}_n \simeq \mathrm{GL}_{2p}(\mathbb{R})/ \mathrm{O}_{2p}\;$.
\end{lem}
\begin{proof}
A regular linear operator $L : \, \mathbb{R}^n \to \mathbb{R}^{2p}$
determines an orthogonal decomposition $\mathbb{R}^n = \mathrm{ker}
(L) \oplus \mathrm{im}(L^\dagger)$ into Euclidean subspaces of
dimension $n-2p$ resp.\ $2p$ and hence a point of the symmetric space
$(\mathrm{O}_{2p} \times \mathrm{O}_{n-2p}) \setminus \mathrm{O}_n
\,$. Therefore, arguing in the same way as in the proof of Lemma
\ref{lem:3.8}, we have an identification
\begin{displaymath}
    \mathrm{Hom}^\prime(\mathbb{R}^n , \mathbb{R}^{2p})
    \simeq \mathrm{GL}_{2p}(\mathbb{R}) \times_{\mathrm{O}_{2p}}
    (\mathrm{O}_{n-2p} \backslash \mathrm{O}_n ) \;.
\end{displaymath}
The desired statement follows by taking the quotient by $\mathrm{O}_n
\,$.
\end{proof}
\begin{rem}
Although each of $\mathrm{GL}_{2p}(\mathbb{R})$ and $\mathrm{O}_{2p}$
has two connected components, their quotient $\mathrm{GL}_{2p}(
\mathbb{R}) / \mathrm{O}_{2p} = \mathrm{GL}_{2p}^+ (\mathbb{R}) /
\mathrm{SO}_{2p}$ is connected.
\end{rem}
For later purposes note that the anti-unitary map $C : \, U \to U$
combines with the Hermitian structure of $U$ to give the canonical
symmetric bilinear form of $U:$
\begin{displaymath}
    \langle C (u \oplus \varphi) , u^\prime \oplus \varphi^\prime
    \rangle = \varphi^\prime(u) + \varphi(u^\prime) = s(u \oplus
    \varphi , u^\prime \oplus \varphi^\prime) \;.
\end{displaymath}

\subsubsection{$K = \mathrm{USp}_n\;$}

In the final case to be addressed, we are given an alternating
isomorphism $\varepsilon : \, \mathbb{C}^n \to (\mathbb{C}^n)^\ast$
and hence an anti-unitary operator $T := c_n^{-1} \varepsilon : \,
\mathbb{C}^n \to \mathbb{C}^n$ which squares to $T^2 = -1$. Note $n
\in 2 \mathbb{N}\,$. The Hermitian vector space $\mathbb{C}^n$ now
carries the extra structure of a complex symplectic vector space with
symplectic form
\begin{displaymath}
    \omega(v,v^\prime) := \langle Tv , v^\prime \rangle
    = \overline{\langle T^2 v , T v^\prime \rangle}
    = - \omega(v^\prime , v) \;.
\end{displaymath}
The symmetry group of the Hermitian symplectic vector space
$\mathbb{C}^n$ is $K =\mathrm{USp}_n\,$.

To do further analysis in this situation, it is convenient to fix
some decomposition
\begin{displaymath}
    \mathbb{C}^n = P \oplus T (P)
\end{displaymath}
which is orthogonal with respect to the Hermitian structure of
$\mathbb{C}^n$ and Lagrangian w.r.t.\ the symplectic structure. The
latter means that $P$ and $T(P)$ are non-degenerately paired by
$\omega\,$, so that we have an isomorphism $T(P) \stackrel{\sim}
{\to} P^\ast$ by $Tv \mapsto \omega(Tv\, , \, \cdot) = - \langle v\,
, \, \cdot \rangle$.

Writing $U = \mathbb{C}^p \oplus (\mathbb{C}^p)^\ast$ we still define
$\psi : \, \mathrm{Hom}(\mathbb{C}^n , \mathbb{C}^p) \to \mathrm{Hom}
(\mathbb{C}^n,U)$ by
\begin{displaymath}
    \psi(L) = L \oplus \overline{L} \, \varepsilon
    = L \oplus c_p\, L\, T \;,
\end{displaymath}
and invoke the canonical Hermitian structure of $U$ to determine the
adjoint $\psi(L)^\dagger$. For future reference we note that the map
$L \mapsto \psi(L)$ is $\mathrm{USp}_n$-equivariant.
\begin{lem}
The decomposition $\mathbb{C}^n = \mathrm{ker}\, \psi(L) \oplus
\mathrm{im}\, \psi(L)^\dagger$ is a decomposition into Hermitian
symplectic subspaces.
\end{lem}
\begin{proof}
By the definition of the operation of taking the Hermitian adjoint,
the space $\mathrm{im}\, \psi(L)^\dagger$ is the orthogonal
complement of $\mathrm{ker}\, \psi(L)$ in the Hermitian vector space
$\mathbb{C}^n$. Since $U = \mathbb{C}^p \oplus (\mathbb{C}^p)^\ast$
is an orthogonal sum and $c_p : \,\mathbb{C}^p \to (\mathbb{C}^p
)^\ast$ is a bijection, the condition $0 = \psi(L) v = L\,v \oplus
c_p \,L\,T v$ implies that if $v$ is in the kernel of $\psi(L)$ then
so is $T v$. Thus $T$ preserves the subspace $\mathrm{ker}\,\psi(L)
$. Being anti-unitary, the operator $T$ then preserves also the
orthogonal complement $\mathrm{im}\, \psi(L)^\dagger$. It therefore
follows that $\omega$ restricts to a non-degenerate symplectic form
on both subspaces.
\end{proof}
Next, let an anti-unitary operator $C : \, U \to U$ with square $C^2
= -1$ be defined by
\begin{displaymath}
    C (u \oplus \varphi) = (c_p^{-1} \varphi) \oplus (- c_p\, u) \;.
\end{displaymath}
The associated symplectic structure of $U$ is given by the canonical
alternating form:
\begin{displaymath}
    - \langle C(u \oplus \varphi) , u^\prime \oplus \varphi^\prime
    \rangle = \varphi^\prime(u) - \varphi(u^\prime) = a(u \oplus
    \varphi , u^\prime \oplus \varphi^\prime) \;.
\end{displaymath}
A short computation shows that the complex linear operator $\psi(L) :
\, \mathbb{C}^n \to U$ satisfies the relation $\psi(L) = C \psi(L)
T^{-1}$. Let us therefore decompose $\psi(L)$ according to
\begin{displaymath}
    \psi(L) \, : \,\, P \oplus P^\ast \to \mathbb{C}^p \oplus
    (\mathbb{C}^p)^\ast \;.
\end{displaymath}
Recalling $T^2 = -1$ and the fact that the anti-unitary operator $T$
exchanges the subspaces $P$ and $P^\ast$, we then see that $\psi(L) =
C \psi(L) T^{-1}$ is already determined by its blocks $\alpha_1 :=
\psi(L) \vert_{P \to \mathbb{C}^p}$ and $\alpha_2 := \psi(L) \vert_{
P^\ast \to \mathbb{C}^p}:$
\begin{displaymath}
    \psi(L) = \begin{pmatrix} \alpha_1 &\alpha_2\\ - \bar\alpha_2
    &\bar\alpha_1 \end{pmatrix} \;.
\end{displaymath}
This means that the matrix expression of $\psi(L)$ with respect to
symplectic bases of $P \oplus P^\ast$ and $\mathbb{C}^p \oplus
(\mathbb{C}^p)^\ast$ consists of real quaternions $q \in \mathbb{H}:$
\begin{displaymath}
    q = q_0 \begin{pmatrix} 1 &0\\ 0 &1 \end{pmatrix} + q_1
    \begin{pmatrix} 0 &\mathrm{i}\\ \mathrm{i} &0 \end{pmatrix}
    + q_2 \begin{pmatrix} 0 &1\\ -1 &0 \end{pmatrix} + q_3
    \begin{pmatrix} \mathrm{i} &0\\ 0 &-\mathrm{i} \end{pmatrix}
    \quad (q_j \in \mathbb{R}) \;.
\end{displaymath}

Now assume that $n \ge 2p$ and $\psi(L)$ is regular. Then $\psi(L) :
\, \mathrm{im}\, \psi(L)^\dagger \to U$ is an isomorphism of
Hermitian symplectic vector spaces. On expressing this isomorphism
with respect to symplectic bases of $\mathrm{im}\, \psi(L)^\dagger$
and $U$, we can identify it with an element of $\mathrm{GL}_p(
\mathbb{H})$, the group of invertible $p \times p$ matrices with real
quaternions for their entries. Note that another characterization of
the elements $g$ of $\mathrm{GL}_p (\mathbb{H})$ as a subgroup of
$\mathrm{GL}(U)$ is by the equation $C g = g \,C\,$. The subgroup of
unitary elements in $\mathrm{GL}_p (\mathbb{H})$ is the unitary
symplectic group $\mathrm{USp}(U) \equiv \mathrm{USp}_{2p}\;$.

The rest of the argument goes the same way as before: a regular
transformation $\psi(L)$ is determined by a Hermitian symplectic
decomposition $\mathbb{C}^n = \mathrm{ker}\, \psi(L) \oplus
\mathrm{im}\, \psi(L)^\dagger$ together with a $\mathrm{GL}_p
(\mathbb{H})$-transformation from $\mathrm{im}\, \psi(L)^\dagger$ to
$U\,$; taking the quotient by the right action of $\mathrm{USp}_n$ we
directly arrive at the following statement.
\begin{lem}
If $n \ge 2p$ then the space of regular $\mathrm{USp}_n$-orbits in
the image of $\mathrm{Hom}(\mathbb{C}^n,\mathbb{C}^p)$ under $\psi$
is isomorphic to $\mathrm{GL}_{p}(\mathbb{H}) / \mathrm{USp}_{2p}\,
\,$.
\end{lem}

\subsection{Integration formula for $K$-invariant functions}
\label{sect:3.4}

Let us now summarize the results of the previous section. To do this
in a concise way covering all three cases at once, we will employ the
notation laid down in Table \ref{fig:1}.
\begin{table}
\begin{center}
\begin{tabular}{|p{1.5cm}|p{1.5cm}|p{1.5cm}|p{1.5cm}|p{1.8cm}|}
\hline $G_n$ & $K_n$ & $G_p$ & $K_p$ & $K_{n,p}$ \\
\hline
  $\mathrm{GL}_n(\mathbb{C})$ \newline
  $\mathrm{O}_n(\mathbb{C})$ \newline
  $\mathrm{Sp}_n(\mathbb{C})$
& $\mathrm{U}_n$ \newline
  $\mathrm{O}_n(\mathbb{R})$ \newline
  $\mathrm{USp}_n$
& $\mathrm{GL}_p(\mathbb{C})$ \newline
  $\mathrm{GL}_{2p}(\mathbb{R})$ \newline
  $\mathrm{GL}_p(\mathbb{H})$
& $\mathrm{U}_p$ \newline
  $\mathrm{O}_{2p}(\mathbb{R})$ \newline
  $\mathrm{USp}_{2p}$
& $\mathrm{U}_{n-p}$ \newline
  $\mathrm{O}_{n-2p}(\mathbb{R})$ \newline
  $\mathrm{USp}_{n-2p}$ \\
\hline
\end{tabular}\vskip 10pt
\caption{Meaning of the groups $K_n\,$, $G_p\,$, $K_p\,$, $K_{n,p}\,$
for the three choices of $G_n\,$.}\label{fig:1}
\end{center}
\end{table}
\begin{prop}
If $\mathrm{rank}(K_p) \le \mathrm{rank}(K_n)$ so that $K_p \subset
K_n\,$, the space of regular $K_n$-orbits in $\mathrm{Hom}
(\mathbb{C}^n , \mathbb{C}^p)$ is isomorphic to the non-compact
symmetric space $G_p / K_p\;$.
\end{prop}
Motivated by this result, our next goal is to reduce the integral of
a $K_n$-invariant function on $V_\mathbb{R} \simeq \mathrm{Hom}
(\mathbb{C}^n , \mathbb{C}^p)$ to an integral over $G_p / K_p\,$. To
prepare this step we introduce some further notations and definitions
as follows.

First of all, let $U_p$ denote the Hermitian vector space
\begin{displaymath}
    U_p = \left\{ \begin{array}{ll} \mathbb{C}^p &\quad G_n =
    \mathrm{GL}_n \;,\\ \mathbb{C}^p \oplus (\mathbb{C}^p)^\ast
    ;\, s &\quad G_n = \mathrm{O}_n \;,\\ \mathbb{C}^p \oplus
    (\mathbb{C}^p)^\ast ; \, a &\quad G_n = \mathrm{Sp}_n \;.
    \end{array} \right.
\end{displaymath}
In the second case $U_p$ carries a Euclidean structure (on
$\mathbb{R}^{2p} \simeq \mathrm{Fix}(C) \subset U_p$) by the
symmetric form $s\,$, in the third case it carries a symplectic
structure by the alternating form $a\,$. Then let us regard $U_p$
(assuming that, depending on the case, the inequality $p \le n$ or
$2p \le n$ is satisfied) as a subspace of $\mathbb{C}^n$ with
orthogonal complement $U_{n,p}\,$, thereby fixing an orthogonal
decomposition $\mathbb{C}^n = U_p \oplus U_{n,p}\,$. This
decomposition is Hermitian, Euclidean, or Hermitian symplectic,
respectively.

Let now $X_{p,\,n}$ denote the vector space of structure-preserving
linear transformations $U_p \oplus U_{n,p} \to U_p\,$. Using the
language of matrices one would say that
\begin{displaymath}
    X_{p,\,n} \simeq \left\{ \begin{array}{ll}
    \mathrm{Mat}_{p,\,n}(\mathbb{C})\;,&\quad G_n = \mathrm{GL}_n\;,\\
    \mathrm{Mat}_{2p,\,n}(\mathbb{R})\;,&\quad G_n = \mathrm{O}_n\;,\\
    \mathrm{Mat}_{p,\, n/2}(\mathbb{H})\;, &\quad
    G_n = \mathrm{Sp}_n\;. \end{array} \right.
\end{displaymath}
A special element of $X_{p,\,n}$ is the projector $\Pi : \, U_p
\oplus U_{n,p} \to U_p$ on the first summand. By construction, the
symmetry group of the kernel space $\mathrm{ker} (\Pi) = U_{n,p}$ is
$K_{n,p} \;$.

Next, we specify our normalization conventions for invariant measures
on the Lie groups and symmetric spaces at hand. For that purpose, let
$\psi$ denote the $K_n$-equivariant isomorphism discussed in the
previous subsection:
\begin{displaymath}
    \psi \, : \,\, \mathrm{Hom}(\mathbb{C}^n , \mathbb{C}^p)
    \to X_{p,\,n}\;, \quad L \mapsto \left\{ \begin{array}{ll}
    L\;, &\quad G_n = \mathrm{GL}_n \;, \\ L \oplus \overline{L}\,
    \delta \;, &\quad G_n = \mathrm{O}_n \;,\\ L \oplus \overline{L}
    \,\varepsilon\;, &\quad G_n = \mathrm{Sp}_n\;.\end{array}\right.
\end{displaymath}
To avoid making case distinctions, we introduce an integer $m$ taking
the value $m = 0, +1, -1$ for $G = \mathrm{GL}_n\,$, $\mathrm{O}_n
\,$, $\mathrm{Sp}_n\,$, respectively. Then from $\mathrm{Tr}_{
\mathbb{C}^p }\, L L^\dagger = \mathrm{Tr}_{ (\mathbb{C}^p)^\ast}\,
\overline{L} L^\mathrm{t}$ we have the relation
\begin{displaymath}
    \mathrm{Tr}_{\mathbb{C}^p}\, L L^\dagger = (1+|m|)^{-1}
    \mathrm{Tr}_{U_p}\, \psi(L) \psi(L)^\dagger \;,
\end{displaymath}
which transfers the Euclidean norm of the vector space $V_\mathbb{R}
\simeq \mathrm{Hom}(\mathbb{C}^n , \mathbb{C}^p)$ to a corresponding
norm on $X_{p,\,n\,}$. In view of the scaling implied by this
transfer, we equip the Lie algebra $\mathrm{Lie}(K_p) = T_e K_p$ with
the following trace form (or Euclidean structure):
\begin{displaymath}
    \mathrm{Lie}(K_p) \to \mathbb{R} \;, \quad A \mapsto \parallel A
    \parallel^2 := - (1+|m|)^{-1} \mathrm{Tr}_{U_p} \, A^2 \ge 0 \;.
\end{displaymath}
The compact Lie group $K_p$ is then understood to carry the invariant
metric tensor and invariant volume density given by this Euclidean
structure on $\mathrm{Lie}(K_p)$. The same convention applies to the
compact Lie groups $K_n$ and $K_{n,p}\,$. Please note that these
conventions are standard and natural in that they imply, e.g.,
$\mathrm{vol} (\mathrm{U}_1) = \mathrm{vol}(\mathrm{SO}_2) = 2\pi$.

By the symbol $dg_{K_p}$ we will denote the $G_p$-invariant measure
on the non-compact symmetric space $G_p / K_p\,$. In keeping with the
normalization convention we have just defined, the restriction of
$dg_{K_p}$ to the tangent space $T_o (G_p / K_p)$ at $o := K_p$ is
the Euclidean volume density determined by the trace form $B \mapsto
\parallel B \parallel^2 = (1+|m|)^{-1} \mathrm{Tr}_{U_p}\, B^2\,$,
which is positive for Hermitian matrices $B = B^\dagger$.

As a final preparation, we observe that the principal bundle $G_p \to
G_p / K_p$ is trivial in all cases. Recall also that the Euclidean
vector space $V_\mathbb{R} \simeq \mathrm{Hom}(\mathbb{C}^n ,
\mathbb{C}^p)$ comes with a canonical volume form (actually, a
density) $\mathrm{dvol}_{V_\mathbb{R}}\,$.
\begin{prop}\label{prop:BB-sector}
For $V = \mathrm{Hom}( \mathbb{C}^n , \mathbb{C}^p ) \oplus
\mathrm{Hom}( \mathbb{C}^p , \mathbb{C}^n )$ let $f \in \mathcal{O}
(V)^{G_n}$ be a holomorphic function on $V$ with the symmetry $f(L
\oplus \tilde{L}) = f(L\,h \oplus h^{-1} \tilde{L})$ for all $h \in
G_n\,$. Restrict $f$ to a $K_n$-invariant function $f_r$ on the real
vector subspace $V_\mathbb{R} \simeq \mathrm{Hom} (\mathbb{C}^n ,
\mathbb{C}^p)$ by $f_r(L) := f(L \oplus L^\dagger)$. If $f_r$ is a
Schwartz function, then
\begin{displaymath}
    \int\limits_{\mathrm{Hom}(\mathbb{C}^n,\mathbb{C}^p)} f_r(L)\,
    \mathrm{dvol}_{V_\mathbb{R}}(L) = \frac{\mathrm{vol}(K_n)}
    {\mathrm{vol}(K_{n,p})} \int\limits_{G_p / K_p}
    f_r \circ \psi^{-1}(g \Pi) \,J(g) \, dg_{K_p}\;,
\end{displaymath}
where the Jacobian function $J : \, G_p / K_p \to \mathbb{R}$ is
given by $J(g) = 2^{p^2 - pn} \vert \mathrm{Det}(g) \vert^{2n}$ for
$G_n = \mathrm{GL}_n$ and $J(g) = 2^{2p^2 - pn} | \mathrm{Det}(g)
|^n$ for $G_n = \mathrm{O}_n\, , \mathrm{USp}_n\;$.
\end{prop}
\begin{proof}
Convergence of the integral on both sides of the equation is
guaranteed by the requirement that the integrand $f_r$ be a Schwartz
function.

The first step is to transform the integral on the left-hand side to
the domain $X_{p,\,n}$ with integrand $(\psi^{-1})^\ast (f\, \mathrm
{dvol}_{ V_\mathbb{R}})$. Of course the space $X_{p,\,n }^\prime$ of
regular elements in $X_{p,\,n}$ has full measure with respect to
$(\psi^{-1})^\ast (\mathrm{dvol}_{ V_\mathbb{R}})$. Now choose a
section $s$ of the trivial principal bundle $\pi : \, G_p \to G_p /
K_p$ and parameterize $X_{p,\,n}^\prime$ by the diffeomorphism
\begin{displaymath}
    \phi\,:\,\, (G_p/K_p) \times (K_n / K_{n,p})\to X_{p,\,n}^\prime
    \;,\quad (x , k K_{n,p})\mapsto s(x) \Pi\, k^{-1} .
\end{displaymath}
Using $\phi\,$, transform the integral from $X_{p,\,n}^\prime$ to
$(G_p / K_p) \times (K_n / K_{n,p})$. Right $K_n$-trans\-lations (as
well as left $K_p$-translations) are isometries of $(\psi^{-1})^\ast
(\mathrm{dvol}_{V_\mathbb{R}})$, and varying $\psi(L) = s(x) \Pi\,
k^{-1}$ we get $\delta \psi(L) = \delta s(x) \Pi \, k^{-1} - s(x) \Pi
\, k^{-1} \delta k \, k^{-1}$. Therefore, the pullback of $\mathrm
{dvol}_{V_\mathbb{R} }$ by $\psi^{-1} \circ \phi$ is proportional to
the product of invariant measures of $G_p / K_p$ and $K_n / K_{n,p}$
times a Jacobian $j(x)$ which can be computed as the Jacobian of the
map
\begin{displaymath}
    \mathcal{L}_{s(x)} : \, X_{p,\,n} \to X_{p,\,n}\;, \quad
    L \mapsto s(x) L \;.
\end{displaymath}
In the case of $X_{p,\,n} = \mathrm{Hom}(\mathbb{C}^n,\mathbb{C}^p)$
this gives $j(x) = \vert \mathrm{Det}(s(x)) \vert^{2n}$. In the other
two cases the dimension of $U_p$ is doubled while the (real)
dimension of $X_{p,\,n}$ stays the same; hence $j(x) = \vert
\mathrm{Det}(s(x)) \vert^n$. In all cases we may replace $\vert
\mathrm{Det}(s(x)) \vert$ by $\vert \mathrm{Det}(g) \vert$ where $g$
is any point in the fiber $\pi^{-1}(x)$. Also, by the
$K_n$-invariance of the integrand $f$ one has $f \circ \psi^{-1}
(s(x) \Pi \, k^{-1}) = f \circ \psi^{-1}(g \Pi)$ independent of the
choice of $g \in \pi^{-1}(x)$. This already proves that the two
integrals on the left-hand and right-hand side are pro\-portional to
each other, with the constant of proportionality being independent of
$f$.

It remains to ascertain the precise value of this constant. Doing the
invariant integral over $K_n / K_{n,p}$ one just picks up the
normalization factor of volumes $\mathrm {vol}(K_n) / \mathrm{vol}
(K_{n,p})$. The remaining factor $2^{p^2 - pn}$ or $2^{2p^2 - pn}$ in
$J(g)$ is determined by the following consideration. Decomposing the
elements $\xi \in \mathfrak{k} \equiv \mathrm{Lie}(K_n)$ as
\begin{displaymath}
    \mathfrak{k} \ni \xi = \begin{pmatrix} A &B \\ -B^\dagger &D
    \end{pmatrix} \in \begin{pmatrix} \mathfrak{k} \cap
    \mathrm{End}(U_p) &\mathfrak{k} \cap \mathrm{Hom}(U_{n,p}\, ,U_p)
    \\ \mathfrak{k} \cap \mathrm{Hom}(U_p\, ,U_{n,p}) &\mathfrak{k}
    \cap \mathrm{End}(U_{n,p}) \end{pmatrix} \;.
\end{displaymath}
we have the norm square $\parallel \xi \parallel^2 = (1 + |m|)^{-1}
(- \mathrm{Tr}\, A^2 + 2\, \mathrm{Tr}\, B B^\dagger - \mathrm{Tr}\,
D^2)$. On the other hand, the differential of the mapping $\phi$ at
$(o,eK_{n,p}) \in (G_p/K_p) \times (K_n/K_{n,p})$ is
\begin{displaymath}
    H, \begin{pmatrix} A &B \\ - B^\dagger &0 \end{pmatrix}
    \mapsto (H+A) \oplus B \in X_{p,\,n} \cap \mathrm{End}(U_p)
    \oplus X_{p,\,n} \cap \mathrm{Hom}(U_{n,p}\,,U_p) \;,
\end{displaymath}
which gives the norm square $\parallel (H+A) \oplus B \parallel^2 =
(1 + |m|)^{-1} (\mathrm{Tr}\, H^2 - \mathrm{Tr}\, A^2 + \mathrm{Tr}\,
B B^\dagger)$. Thus the term $\mathrm{Tr}\, B B^\dagger$ gets scaled
by a factor of two, and by counting the number of independent
freedoms in $B \in \mathrm{Hom}(U_{n,p}\,, U_p)$ we see that the
Jacobian $J(g)$ receives an extra factor of $2^{p(n-p)}$ for the case
of $G_n = \mathrm{GL}_n$ and $2^{2p(n-2p)/2}$ for $G_n = \mathrm{O}_n
\,$, $\mathrm{Sp}_n\,$.
\end{proof}
\begin{rem}
For $n = 2p$ and $G_n = \mathrm{O}_n$ the space $K_n / K_{n,p} =
\mathrm{O}_n(\mathbb{R})$ consists of two connected components and
the volume factor means $\mathrm{vol} (K_n) / \mathrm{vol}(K_{n,p}) =
\mathrm{vol}(\mathrm{O}_n(\mathbb{R})) = 2\, \mathrm{vol}(\mathrm
{SO}_n (\mathbb{R}))$. On the other hand, for $n > 2p$ and the same
case the volume factor is that of the connected space $K_n / K_{n,p}
= \mathrm{O}_n(\mathbb{R})/\mathrm{O}_{n-2p} (\mathbb{R}) =
\mathrm{SO}_n(\mathbb{R})/\mathrm{SO}_{n-2p} (\mathbb{R})$.
\end{rem}
To finish this section, we cast Prop.\ \ref{prop:BB-sector} in a form
closer in spirit to the rest of paper.

Recall that either we have $G_p = \mathrm{GL}(U_p)$, or else $G_p
\subset \mathrm{GL}(U_p)$ is characterized by the commutation rule $g
\, C = C g\,$. Since $C^\dagger = C^{-1} = \pm C$, all our groups
$G_p$ are stabilized by the dagger operation. Thus there exists an
involution
\begin{displaymath}
    \theta \,:\,\, G_p \to G_p\;,\quad g \mapsto (g^{-1})^\dagger\;,
\end{displaymath}
which is actually a Cartan involution fixing the elements of the
maximal compact subgroup $K_p\,$. The mapping
\begin{displaymath}
    \gamma\, : \,\, G_p / K_p \to G_p \;, \quad g \mapsto g
    \theta(g^{-1}) = g g^\dagger \;,
\end{displaymath}
embeds the symmetric space into the group. To clarify the connection
with the setting of Sect.\ \ref{sect:Q-map}, let us understand the
image of this embedding as a subspace of $W$.
\begin{lem}
The Cartan embedding $\gamma : \, G_p / K_p \to G_p \subset
\mathrm{End}(U_p)$ projected to the positive Hermitian operators in
$W \subset \mathrm {End}(U_p)$ is a bijection.
\end{lem}
\begin{proof}
From $\psi(L) = L$ in the first case, and $\psi(L) = L \oplus
\overline{L}\, \beta$ in the last two cases, we immediately see that
the composition of mappings
\begin{displaymath}
    L\oplus L^\dagger \mapsto \psi(L)\mapsto \psi(L)\psi(L)^\dagger
    = Q(L \oplus L^\dagger) \in Q(V_\mathbb{R})
\end{displaymath}
is the quadratic map $Q : \, V \to W$ (Sect.\ \ref{sect:Q-map})
restricted to $V_\mathbb{R} \,$, and since $g \in G_p$ arises from
decomposing $\psi(L) = g \Pi \, k^{-1}$, the positive Hermitian
operator $g g^\dagger = \psi(L) \psi(L)^\dagger$ lies in
$Q(V_\mathbb{R}) \subset W$. Thus the embedding $G_p / K_p \to G_p$
is into $\mathrm{Herm}^+ \cap W$.

It remains to be shown that $\gamma :\, G_p / K_p \to \mathrm{Herm}^+
\cap W$ is one-to-one. In the case of $G_p = \mathrm{GL}(U_p)$, every
positive Hermitian operator $h \in \mathrm{Herm}^+ \cap W$ has a
unique positive Hermitian square root $\sqrt{h}\,$, and $h = \sqrt{h}
\, \theta (\sqrt{h})^{-1} = \sqrt{h}k\, \theta(\sqrt{h} k)^{-1}$ ($k
\in K_p$). Thus there exists a unique inverse $\gamma^{-1}(h) =
\sqrt{h} K_p \in G_p / K_p\,$.

To deal with the other two cases we recall the relation $\langle C_b
\, \cdot , \cdot^\prime \rangle = \pm b(\cdot , \cdot^\prime)$, i.e.,
$C \equiv C_b$ combines with the Hermitian structure of $U_p$ to give
the bilinear form $b\,$, where $b = s$ or $b = a\,$. This implies
that the symmetric transformations $w \in W = \mathrm{Sym}_b(U_p)$
are characterized by the commutation rule $C_b w = w^\dagger C_b\,$.
Indeed,
\begin{displaymath}
    \langle C_b w \, \cdot , \cdot^\prime \rangle = \pm b(w \,
    \cdot , \cdot^\prime) = \pm b(\cdot , w \, \cdot^\prime) =
    \langle C_b \, \cdot , w \, \cdot^\prime \rangle =
    \langle w^\dagger C_b \, \cdot , \cdot^\prime \rangle \;,
\end{displaymath}
and hence $W \cap \mathrm{Herm}$ are exactly the elements of
$\mathrm{End}(U_p)$ that commute with $C_b\,$. The desired statement
now follows from the definition $G_p = \{ g \in \mathrm{GL}(U_p) \, |
\, g\, C_b = C_b\, g \}$ because the squaring map on $\mathrm{Herm}^+
\cap \mathrm{GL}(U_p)$ remains a bijection when restricted to the set
of fixed points $\mathrm{Herm}^+ \cap G_p$ of the involution $w
\mapsto C_b \, w \, (C_b)^{-1}\,$.
\end{proof}
In the sequel we will often use the abbreviations
\begin{displaymath}
    n^\prime := (1 + |m|)^{-1} n \;, \quad \mathrm{Tr}^{\,\prime} :=
    (1 + |m|)^{-1} \mathrm{Tr}_{U_p} \quad (m = 0, 1, -1 \text{~for~}
    G= \mathrm{GL}, \mathrm{O}, \mathrm{Sp})\;.
\end{displaymath}
Let now $D_p := \mathrm{Herm}^+ \cap W$ denote the set of positive
Hermitian operators in $W$:
\begin{displaymath}
    D_p = \left\{ \begin{array}{ll} \mathrm{Herm}^+ \cap
    \mathrm{End}(\mathbb{C}^p)\;, &\quad G_n = \mathrm{GL}_n\;,
    \\ \mathrm{Herm}^+ \cap \mathrm{Sym}_s(\mathbb{C}^p \oplus
    (\mathbb{C}^p)^\ast)\;, &\quad G_n = \mathrm{O}_n \;, \\
    \mathrm{Herm}^+ \cap \mathrm{Sym}_a(\mathbb{C}^p \oplus
    (\mathbb{C}^p)^\ast)\;, &\quad G_n = \mathrm{Sp}_n \;.
    \end{array}\right.
\end{displaymath}
$D_p$ is equipped with a $G_p$-invariant measure $d\mu_{D_p}\,$. In
keeping with our general conventions, we normalize $d\mu_{D_p}$ so
that $(d\mu_{D_p})_o$ agrees with the Euclidean volume density of the
Euclidean vector space of Hermitian operators $H \in \mathrm{Lie}
(G_p)$ with norm square $\parallel H \parallel^2 = (1 + |m|)^{-1}
\mathrm{Tr}\, H^2 = \mathrm{Tr}^{\,\prime} H^2$. The Cartan embedding
$g \mapsto g \theta(g^{-1})$ on the Hermitian operators $g =
\mathrm{e}^H$ is the squaring map $\mathrm{e}^H \mapsto
\mathrm{e}^{2H}$. Thus, pushing the $G_p$-invariant measure
$dg_{K_p}$ forward by the Cartan embedding we get
$2^{-\mathrm{dim}(G_p / K_p)} d\mu_{D_p}\,$. Now since
\begin{displaymath}
    \mathrm{dim}(G_p / K_p) = \left\{ \begin{array}{ll} p^2 &\quad
    G_n = \mathrm{GL}_n \;,\\ p(2p+1) &\quad G_n =\mathrm{O}_n \;,
    \\ p(2p-1) &\quad G_n = \mathrm{Sp}_n \;, \end{array} \right.
\end{displaymath}
the following statement is a straightforward reformulation of Prop.\
\ref{prop:BB-sector}.
\begin{cor}\label{cor:BB-sector}
Given $f \in \mathcal{O}(V)^{G_n}$, and retaining the setup and the
conditions of Prop.\ \ref{prop:BB-sector}, define $F \in
\mathcal{O}(W)$ by $Q^\ast F = f\,$. Then
\begin{displaymath}
    \int_{V_\mathbb{R}} f \, \mathrm{dvol}_{V_\mathbb{R}} =
    2^{-p(n+m)} \frac{\mathrm{vol}(K_n)}{\mathrm{vol}(K_{n,p})}
    \int_{D_p} F(x)\, \mathrm{Det}^{n^\prime}(x)\, d\mu_{D_p}(x)\;.
\end{displaymath}
\end{cor}
In particular, since the function $x \mapsto F(x) = \mathrm{e}^{-
\mathrm{Tr}^{\,\prime} x}$ pulls back to $L \mapsto f(L) =
\mathrm{e}^{- \mathrm{Tr}\, L L^\dagger}$ and the Gaussian integral
$\int \mathrm{e}^{-\mathrm{Tr}\, L L^\dagger} d\mathrm{vol}_{
V_\mathbb{R}}(L)$ has the value $\pi^{pn}$, we infer the formula
\begin{equation}\label{eq:3.2}
    \int_{D_p} \mathrm{e}^{- \mathrm{Tr}^{\,\prime} x}\,
    \mathrm{Det}^{n^\prime}(x) \, d\mu_{D_p}(x) = (2\pi)^{p n} 2^{pm}
    \mathrm{vol}(K_{n,p}) / \mathrm{vol}(K_n) \;.
\end{equation}

\section{Lifting in the fermion-fermion sector}
\label{sect:FF-sector}\setcounter{equation}{0}

Having settled the case of $V_1 = 0$ (or $q = 0$) we now turn to the
complementary case where $V_0 = 0$ (or $p = 0$). Thus, in the present
section we consider
\begin{displaymath}
    V \equiv V_1 = \mathrm{Hom}(\mathbb{C}^n , \mathbb{C}^q) \oplus
    \mathrm{Hom}(\mathbb{C}^q ,\mathbb{C}^n) \simeq\mathbb{C}^{2qn}\;,
\end{displaymath}
in which case our basic algebra $\mathcal{A}_V$ becomes an exterior
algebra of dimension $2^{2qn}$:
\begin{displaymath}
    \mathcal{A}_V = \wedge (V^\ast)\simeq \wedge(\mathbb{C}^{2N}) \;,
    \quad N = q n \;.
\end{displaymath}
In the sequel, we will prove an analog of Prop.\ \ref{prop:BB-sector}
for this situation.

\subsection{Quadratic $G$-invariants}
\label{sect:FF-invs}

This subsection is closely analogous to Sect.\ \ref{sect:BB-invs},
the main difference being that the role of the symmetric algebra
$\mathrm{S} (V^\ast)$ is now taken by the exterior algebra $\wedge
(V^\ast)$. It remains true \cite{Howe1995} that for each of the
classical reductive complex Lie groups $G = \mathrm{GL}_n\,$,
$\mathrm{O}_n\,$, and $\mathrm {Sp}_n\,$, a basis of $\wedge^2
(V^\ast)^G$ is a generating system for $\wedge (V^\ast)^G$. Let us
therefore make another study of these quadratic invariants.

Recall that on the direct sum $U_q := \mathbb{C}^q \oplus
(\mathbb{C}^q)^\ast$ we have the canonical symmetric bilinear form
$s$ and the canonical alternating bilinear form $a\,$.
\begin{lem}\label{lem:FF-invs}
Let $V = \mathrm{Hom}(\mathbb{C}^n , \mathbb{C}^q) \oplus
\mathrm{Hom}(\mathbb{C}^q , \mathbb{C}^n)$ carry the $G$-action
induced from the fundamental action on $\mathbb{C}^n$ of $G =
\mathrm{GL}_n\,$, $\mathrm{O}_n\,$, or $\mathrm{Sp}_n\,$. If $U_q =
\mathbb{C}^q \oplus (\mathbb {C}^q)^\ast$, then the space of
quadratic invariants $\wedge^2 (V^\ast)^G$ is isomorphic as a complex
vector space to $W^\ast$ where $W = \mathrm{End}(\mathbb{C}^q)$ for
$G = \mathrm{GL}_n\,$, $W = \mathrm{Sym}_a(U_q)$ for $G = \mathrm
{O}_n\,$, and $W = \mathrm{Sym}_s (U_q)$ for $G = \mathrm{Sp}_n \,$.
\end{lem}
\begin{proof} (Sketch).
There is no conceptual difference from the proofs of Lemma
\ref{lem:BB-GL-invs} and Lemma \ref{lem:BB-OSP-invs}, and we
therefore give only a summary of the changes.

In the case of $G = \mathrm{GL}_n\,$, all quadratic invariants still
arise by composing the elements of $\mathrm{Hom}(\mathbb{C}^q ,
\mathbb{C}^n)$ with those of $\mathrm{Hom}(\mathbb{C}^n, \mathbb{C}^q
)$. Thus $\wedge^2(V)^G \simeq \mathrm{End} (\mathbb{C}^q)$ and,
since $G$ acts reductively and $\wedge^2(V)^G$ is paired with
$\wedge^2(V^\ast)^G$, we have $\wedge^2(V^\ast)^G \simeq \mathrm{End}
(\mathbb{C}^q)^\ast$.

In the other cases we use the isomorphism $\wedge^2 (V^\ast) \to
\mathrm{Alt}(V,V^\ast)$ given by
\begin{displaymath}
    \varphi^\prime \varphi - \varphi \varphi^\prime
    \mapsto \left( v \mapsto \varphi^\prime(\cdot)
    \varphi(v) - \varphi(\cdot) \varphi^\prime(v) \right) \;,
\end{displaymath}
which descends to an isomorphism $\wedge^2(V^\ast)^G \to
\mathrm{Alt}_G (V,V^\ast)$. Then, writing $U \equiv U_q$ we make the
$G$-equivariant identification $V \simeq U^\ast \otimes \mathbb{C}^n$
and have
\begin{displaymath}
    \wedge^2(V^\ast)^G \simeq \mathrm{Alt}_G (U^\ast \otimes
    \mathbb{C}^n , U \otimes (\mathbb{C}^n)^\ast ) \simeq \left\{
    \begin{array}{ll} \mathrm{Alt}(U^\ast , U)\;, &\quad G =
    \mathrm{O}_n \;, \\ \mathrm{Sym}(U^\ast , U)\;, &\quad G =
    \mathrm{Sp}_n \;. \end{array} \right.
\end{displaymath}
This leads to the desired statement by the isomorphisms
$\mathrm{Alt}(U^\ast , U) \simeq \mathrm{Alt}(U, U^\ast)^\ast \simeq
\mathrm{Sym}_a(U)^\ast$ and $\mathrm{Sym}(U^\ast , U) \simeq
\mathrm{Sym}(U, U^\ast)^\ast \simeq \mathrm{Sym}_s(U)^\ast$.
\end{proof}
As before, let the elements $w \in W$ be decomposed as
\begin{equation*}
    w = \begin{pmatrix} A &B \\ C &A^\mathrm{t} \end{pmatrix}
    \in \begin{pmatrix} \mathrm{End}(\mathbb{C}^q)
    &\mathrm{Hom}( (\mathbb{C}^q)^\ast, \mathbb{C}^q) \\
    \mathrm{Hom}( \mathbb{C}^q , (\mathbb{C}^q)^\ast)
    &\mathrm{End}((\mathbb{C}^q)^\ast) \end{pmatrix} \;.
\end{equation*}
Here $B = - B^\mathrm{t}$ and $C = - C^\mathrm{t}$ for the case of $G
= \mathrm {O}_n\,$, while $B = + B^\mathrm{t}$ and $C = +
C^\mathrm{t}$ for $G = \mathrm{Sp}_n\,$, and $B = C = 0$ for $G =
\mathrm{GL}_n\,$. By simple counting, the dimensions of $W$ are
$\mathrm{dim}\, W = q^2$, $q(2q-1)$, and $q(2q+1)$ for $G =
\mathrm{GL}_n\,$, $\mathrm{O}_n\,$, and $\mathrm{Sp}_n\,$,
respectively.

One can now reconsider the quadratic mapping $Q : \, V \to W$ defined
in Sect.\ \ref{sect:Q-map}, with the twist that the elements of $V$
in the present context are to be multiplied with each other in the
alternating sense of exterior algebras. However, what matters for our
purposes is not the mapping $Q$ but the pullback of algebras $Q^\ast
: \, \mathcal{O}(W) \to \wedge (V^\ast)^G$. Let us now specify the
latter at the level of the isomorphism $Q^\ast : \, W^\ast \to
\wedge^2(V^\ast)^G$.

For this let $\{ e_i \}$, $\{ f^i \}$, $\{ e_c \}$, and $\{ f^c \}$
be standard bases of $\mathbb{C}^n$, $(\mathbb{C}^n )^\ast$,
$\mathbb{C}^q$, and $(\mathbb{C}^q)^\ast$, respectively, and define
bases $\{ \zeta_c^i \}$ and $\{\tilde{\zeta}_i^c \}$ of $\mathrm{Hom}
(\mathbb{C}^q , \mathbb{C}^n)^\ast$ and $\mathrm{Hom} (\mathbb{C}^n ,
\mathbb{C}^q)^\ast$ by
\begin{displaymath}
     \zeta_c^i(\tilde{L}) = f^i(\tilde{L} \, e_c) \;, \quad
    \tilde{\zeta}_i^c(L) = f^c(L \, e_i) \;,
\end{displaymath}
where $L \in \mathrm{Hom}(\mathbb{C}^n , \mathbb{C}^q)$ and
$\tilde{L} \in \mathrm{Hom}(\mathbb{C}^q , \mathbb{C}^n)$. Of course
the index ranges are $i = 1, \ldots, n$ and $c = 1, \ldots, q\,$.
Then, decomposing $w \in W \subset \mathrm{End}(U)$ into blocks $A,
B, C$ as above, define a set of linear functions $x_c^{c^\prime},
y^{c^\prime c}, y_{c c^\prime} : \, W \to \mathbb{C}$ by
\begin{displaymath}
    x_c^{c^\prime}(w) = f^{c^\prime}(A e_c) \;, \quad
    y^{c^\prime c}(w) = f^{c^\prime}(B f^c) \;, \quad
    y_{c c^\prime}(w) = (C e_{c^\prime})(e_c) \;.
\end{displaymath}
Notice the symmetry relations $y^{c^\prime c} = \mp y^{c c^\prime}$
and $y_{c c^\prime} = \mp y_{c^\prime c}$ where the upper sign
applies in the case of $G = \mathrm{O}_n$ and the lower sign for $G =
\mathrm{Sp}_n\,$. The functions $\{ x_c^{c^\prime} \}$ constitute a
basis of $W^\ast \simeq \mathrm{End}(\mathbb{C}^q)^\ast$ for the case
of $G = \mathrm{GL}_n \,$. Inclusion of the set of functions $\{
y^{c^\prime c} , y_{c c^\prime} \}_{c < c^\prime}$ gives a basis of
$W^\ast$ for $G = \mathrm{O}_n\,$. By including also the functions
$y^{c^\prime c}$ and $y_{c c^\prime}$ for $c = c^\prime$ we get a
basis of $W^\ast$ for $G = \mathrm{Sp}_n\,$.

Recall that we are given a $G$-equivariant isomorphism $\beta : \,
\mathbb{C}^n \to (\mathbb{C}^n)^\ast$ which is symmetric ($\beta =
\delta$) for $G = \mathrm{O}_n$ and alternating $(\beta =
\varepsilon)$ for $G = \mathrm{Sp}_n\,$.
\begin{lem}
The isomorphism $W^\ast \to \wedge^2 (V^\ast)^G$ has a realization by
\begin{displaymath}
    Q^\ast x_c^{c^\prime} = \tilde{\zeta}_i^{c^\prime} \zeta_c^i
    \;, \quad Q^\ast y^{c^\prime c} = \tilde{\zeta}_i^{c^\prime}
    \beta^{ij} \tilde{\zeta}_j^c \;, \quad Q^\ast y_{c c^\prime}
    = \zeta_c^i \beta_{ij} \zeta_{c^\prime}^j \;,
\end{displaymath}
where $\beta_{ij} = (\beta e_j)(e_i)$ are the matrix entries of
$\beta : \, \mathbb{C}^n \to (\mathbb{C}^n)^\ast$, and $\beta^{ij} =
f^i(\beta^{-1} f^j)$.
\end{lem}
\begin{proof}
All functions $\tilde{\zeta}_i^{c^\prime} \zeta_c^i$, $\tilde
{\zeta}_i^{c^\prime} \beta^{ij} \tilde{\zeta}_j^c$ and $\zeta_c^i
\beta_{ij} \zeta_{c^\prime}^j$ are $G$-invariant in the pertinent
cases. The $q^2$ functions $\tilde{\zeta}_i^{c^\prime} \zeta_c^i$
form a basis of $\wedge^2 (V^\ast)^{\mathrm{GL}_n}$. By including the
$q(q-1)$ functions $\tilde{\zeta}_i^{c^\prime} \delta^{ij} \tilde{
\zeta}_j^c$ and $\zeta_c^i \delta_{ij} \zeta_{c^\prime}^j$ for $c <
c^\prime$, we get a basis of $\wedge^2 (V^\ast)^{\mathrm{O}_n\,}$.
Replacing $\beta = \delta$ by $\beta = \varepsilon$ and expanding the
index range to $c \le c^\prime$ we get a basis of $\wedge^2 (V^\ast)
^{\mathrm{Sp}_n}$. Thus the linear operator $Q^\ast$ takes one basis
to another one and hence is an isomorphism.
\end{proof}
Next, we review some useful representation-theoretic facts about
$\wedge(V^\ast)^G$.

\subsection{$G^\prime$-irreducibility of $\mathcal{A}_V^G$}

Taking $V \oplus V^\ast \simeq \mathbb{C}^{4N}$ to be equipped with
the canonical symmetric form $s(v \oplus \varphi , v^\prime \oplus
\varphi^\prime) = \varphi^\prime(v) + \varphi(v^\prime)$, one defines
the Clifford algebra $\mathrm{Cl}(V \oplus V^\ast)$ to be the
associative algebra generated by $V \oplus V^\ast \oplus \mathbb{C}$
with relations
\begin{displaymath}
    w w^\prime + w^\prime w = s(w,w^\prime) 1
    \qquad (w, w^\prime \in V \oplus V^\ast) \;.
\end{displaymath}
The linear span of the skew-symmetric quadratic elements $(w w^\prime
- w^\prime w)$ is closed under the commutator in $\mathrm{Cl}(V
\oplus V^\ast)$ and is canonically isomorphic to the Lie algebra of
the orthogonal group of the vector space $V \oplus V^\ast$ with
symmetric bilinear form $s\,$. By exponentiating this Lie algebra
inside the Clifford algebra, one obtains the spin group, a connected
and simply connected Lie group denoted by $\mathrm{Spin}(V\oplus
V^\ast) = \mathrm{Spin}_{4N}\,$.

Via their actions on $V$, the complex Lie groups $G = \mathrm{GL}_n
\,$, $\mathrm{O}_n\,$, and $\mathrm{Sp}_n\,$, are realized as
subgroups of $\mathrm{Spin}_{4N}\,$. The centralizer of $G$ in
$\mathrm{Spin}_{4N}$ is another complex Lie group, $G^\prime$, called
the Howe dual partner of $G$ \cite{Howe1995}. The list of such Howe
dual pairs is
\begin{displaymath}
    G \times G^\prime : \quad \mathrm{GL}_n \times \widetilde{
    \mathrm{GL}}_{2q} \;, \quad \mathrm{O}_n \times \mathrm{Spin}_{4q}
    \;, \quad \mathrm{Sp}_n \times \mathrm{Sp}_{4q} \;.
\end{displaymath}
Note that from $\mathrm{O}_n \subset \mathrm{GL}_n$ and
$\mathrm{Sp}_n \subset \mathrm{GL}_n$ one has $\widetilde{
\mathrm{GL}}_{2q} \subset \mathrm{Spin}_{4q}$ and $\widetilde{
\mathrm{GL}}_{2q} \subset \mathrm{Sp}_{4q}\,$. In the case of $n$
being odd, $\widetilde{\mathrm{GL}}_{2q}$ is a double covering of
$\mathrm{GL}_{2q}$ (see below).

A few words of explanation concerning the pairs $G \times G^\prime$
are in order. In the case of the first pair one regards the vector
space $V \oplus V^\ast$ as
\begin{displaymath}
    V \oplus V^\ast \simeq U_q \otimes \mathbb{C}^n \oplus (U_q)^\ast
    \otimes (\mathbb{C}^n)^\ast \;, \quad U_q = \mathbb{C}^q \oplus
    (\mathbb{C}^q)^\ast \;,
\end{displaymath}
and the centralizer of $G = \mathrm{GL}_n$ in $\mathrm{Spin}(V \oplus
V^\ast)$ is then seen to be $G^\prime = \mathrm{GL}(U_q) \equiv
\mathrm{GL}_{2q}$ (or a double cover thereof if $n$ is odd). In the
last two cases the $G$-equivariant isomorphism $\beta : \,
\mathbb{C}^n \to (\mathbb{C}^n)^\ast$ leads to an identification
\begin{displaymath}
    V \oplus V^\ast \simeq (U_q^{\vphantom{\ast}} \oplus U_q^\ast)
    \otimes \mathbb{C}^n \;.
\end{displaymath}
The symmetric bilinear form $s$ on $V \oplus V^\ast$ in conjunction
with $\beta$ induces a bilinear form on $U_q \oplus (U_q)^\ast$. For
$G = \mathrm{O}_n$ this form is the canonical symmetric bilinear form
$s$ and one has the centralizer $G^\prime = \mathrm{Spin}(
U_q^{\vphantom{\ast}} \oplus U_q^\ast ; s) \equiv \mathrm{Spin}_{4q}
\,$. For $G = \mathrm{Sp}_n$ the induced form is the canonical
alternating form $a$ and one has $G^\prime = \mathrm{Sp} (
U_q^{\vphantom{\ast}} \oplus U_q^\ast; a) \equiv \mathrm{Sp}_{4q}\,$.

Now the exterior algebra $\wedge (V^\ast)$ carries the spinor
representation of the Clifford algebra $\mathrm{Cl}(V \oplus
V^\ast)$. This is the representation which is obtained by letting
vectors $v \in V$ and linear forms $\varphi \in V^\ast$ operate by
contraction $\iota(v) : \, \wedge^k (V^\ast) \to \wedge^{k-1}
(V^\ast)$ and exterior multiplication $\varepsilon (\varphi) : \,
\wedge^k(V^\ast) \to \wedge^{k+1}(V^\ast)$.

By the inclusion $G \times G^\prime \subset \mathrm{Spin}_{4N}
\subset \mathrm{Cl}(V \oplus V^\ast)$ the spinor representation of
the Clifford algebra gives rise to a representation on $\mathcal{A}_V
= \wedge(V^\ast)$ of each Howe dual pair $G \times G^\prime$. It is
known \cite{Howe1995} that $\mathcal{A}_V$ decomposes as a direct sum
$\oplus_i \, (U_i \otimes U'_i)$ of irreducible $G \times G^\prime$
representations such that $U_i \not\simeq U_j$ and $U'_i \not\simeq
U'_j$ for $i \not= j\,$. In particular, the representation of
$G^\prime$ on the algebra of $G$-invariants $\mathcal{A}_V^G =
\wedge(V^\ast)^G$ is irreducible.

Next we observe that each of our Howe dual groups $G^\prime$ has rank
$2q\,$. Moreover, one can arrange for all of them to share the same
maximal torus. This is the Abelian group $T = (\mathbb{C}^\times)^q
\times (\mathbb{C}^\times)^q$ acting on $V = V_1 = \mathrm{Hom}(
\mathbb{C}^n , \mathbb{C}^q) \oplus \mathrm{Hom} (\mathbb{C}^q ,
\mathbb{C}^n)$ by diagonal transformations
\begin{displaymath}
    (t_1 , t_2) . (L \oplus \tilde{L}) \mapsto (t_1 L) \oplus
    (\tilde{L} \, t_2) \;.
\end{displaymath}
The induced action of elements $H = (H_1 \, , H_2)$ of the Cartan
algebra $\mathfrak{t} = \mathrm{Lie}(T) = \mathbb{C}^q \oplus
\mathbb{C}^q$ on the spinor module is by operators
\begin{equation}\label{eq:hat-H}
    \hat{H} = {\textstyle{\frac{1}{2}}} \sum\nolimits_c \left(
    (H_1)_c [ \iota(e_i^c) , \varepsilon(f_c^i) ] + (H_2)_c [
    \iota(e_c^i) , \varepsilon(f_i^c) ] \right) \;.
\end{equation}
Here $\{ e_i^c \}$ means the standard basis of $\mathrm{Hom}
(\mathbb{C}^n , \mathbb{C}^q)$, and $\{ e_c^i \}$ means the standard
basis of $\mathrm{Hom}(\mathbb{C}^q , \mathbb{C}^n)$, while $\{ f_c^i
\}$ and $\{ f_i^c \}$ are the corresponding dual bases. The factor of
$1/2$ in front of the sum reflects the fact that the spinor
representation is a ``square root'' representation.

The zero-degree component $\wedge^0(V^\ast) = \mathbb{C}$ -- the
`vacuum' in physics language -- is stabilized by the action of these
operators $\hat{H}\,$. Applying $H$ as $\hat{H}$ to $1 \in
\wedge^0(V^\ast)$ we get
\begin{displaymath}
    H . 1 = \lambda(H) 1 \;, \quad \lambda(H) = \frac{n}{2}
    \sum_{c=1}^q \big( (H_1)_c + (H_2)_c \big) \;.
\end{displaymath}
Note that the weight $\lambda$ is integral for even $n$, but
half-integral for odd $n$. (This is why the latter case calls for the
group $\mathrm{GL}_{2q}$ to be replaced by a double cover $G^\prime =
\widetilde{\mathrm{GL}}_{2q}\,$.) We will denote the integrated
weight or character by $\chi := \mathrm{e}^{\lambda \circ \ln}$.

\subsection{Berezin integral and lowest weight space}
\label{sect:4.3}

We are now going to think of the irreducible
$G^\prime$-representation space $\mathcal{A}_V^G$ as an irreducible
highest-weight module for the Lie algebra of $G^\prime$. To keep the
notation simple we omit the prime and denote this Lie algebra by
$\mathfrak{g} := \mathrm{Lie}(G^\prime)$. Thus
\begin{displaymath}
    \mathfrak{g} = \left\{ \begin{array}{ll}
    \mathfrak{gl}_{2q} \;, &\quad G = \mathrm{GL}_n \;, \\
    \mathfrak{o}_{4q} \;, &\quad G = \mathrm{O}_n \;, \\
    \mathfrak{sp}_{4q} \;, &\quad G = \mathrm{Sp}_n \;.
    \end{array} \right.
\end{displaymath}
The vacuum weight $\lambda$ is a highest weight for the $\mathfrak{g}
$-representation $\mathcal{A}_V^G$. We emphasize this fact by making
a change of notation $\mathcal{A}_V^G \equiv \mathcal{V} (\lambda)$.

The spinor module comes with a $\mathbb{Z}$-grading by the degree,
$\wedge (V^\ast) = \oplus_{k=0}^{2N} \wedge^k(V^\ast)$ where $N =
\frac{1}{2} \mathrm{dim}\, V = qn\,$. Since $G$ is defined on
$\mathbb{C}^n$ and acts on $V$, this grading carries over to the
algebra $\mathcal{A}_V^G = \mathcal{V} (\lambda)$:
\begin{displaymath}
    \mathcal{V}(\lambda) = \bigoplus\nolimits_{k \ge 0}
    \mathcal{V}(\lambda)_k \;.
\end{displaymath}
We denote the highest degree part by $\mathcal{V}(\lambda)
_\mathrm{top}\,$. The highest degree part $\wedge^{2N}(V^\ast)$ of
the spinor module is a complex line stable under the symmetry group
$G$. It is easy to check that $G$ in fact acts trivially on
$\wedge^{2N} (V^\ast)$, so $\mathcal{V}(\lambda)_\mathrm{top} =
\mathcal{V}(\lambda)_{2N} = \wedge^{2N} (V^\ast)$.

Now there exists a canonical generator $\Omega_V \in \wedge^{2N}(V)$
by the following principle. Since the trace form
\begin{displaymath}
    \mathrm{Hom}( \mathbb{C}^n , \mathbb{C}^q) \otimes
    \mathrm{Hom}(\mathbb{C}^q , \mathbb{C}^n) \to \mathbb{C} \;,
    \quad A \otimes B \mapsto \mathrm{Tr}\, AB \;,
\end{displaymath}
is non-degenerate, the vector spaces $\mathrm{Hom}( \mathbb{C}^n ,
\mathbb{C}^q)$ and $\mathrm{Hom}(\mathbb{C}^q , \mathbb{C}^n)$ are
canonically dual to each other. If $\{ e_1 , \ldots, e_N \}$ is any
basis of $\mathrm{Hom}(\mathbb{C}^n , \mathbb{C}^q)$, let $\{ f_1 ,
\ldots, f_N \}$ be the corresponding dual basis of $\mathrm{Hom}
(\mathbb{C}^q , \mathbb{C}^n )$. The exterior product
\begin{displaymath}
    \Omega_V = f_N \wedge e_N \wedge \cdots \wedge f_1 \wedge e_1
\end{displaymath}
then is independent of the choice of basis and only depends on how we
order the two summands in $V = \mathrm{Hom}(\mathbb{C}^n ,
\mathbb{C}^q) \oplus \mathrm{Hom}(\mathbb{C}^q , \mathbb{C}^n)$. For
definiteness, let us say that $\mathrm{Hom}(\mathbb{C}^n ,
\mathbb{C}^q)$ is the first summand. We then have a canonical element
$\Omega_V \in \wedge^{2N} (V)$, and by evaluating the canonical
pairing $\wedge^{2N}(V) \otimes \wedge^{2N} (V^\ast) \to \mathbb{C}$
with fixed argument $\Omega_V$ in the first factor, we get an
identification $\mathcal{V} (\lambda)_{2N} = \wedge^{2N} (V^\ast)
\simeq \mathbb{C}\,$.
\begin{defn}
The projection $\pi : \, \mathcal{V}(\lambda) \to \mathcal{V}
(\lambda)_{2N} \simeq \mathbb{C}$ is called the Berezin integral, and
is here denoted by $f \mapsto \Omega_V[f]$.
\end{defn}
Another way to view this projection is as follows. The vacuum
$\wedge^0(V^\ast)$ is the space of highest-weight vectors for
$\mathfrak{g}\,$, whereas the top degree part $\wedge^{2N}(V^\ast)$
is the space of \emph{lowest-weight vectors}. The latter are the
weight vectors of weight $-\lambda\,$. Indeed, going from zero to top
degree amounts to exchanging the operators $\varepsilon$ and $\iota$,
and since the expression (\ref{eq:hat-H}) is skew-symmetric in these,
the weight changes sign.

Now define the subgroup $H \subset G^\prime$ to be the intersection
of the stabilizer of $\mathcal{V}(\lambda)_0 = \wedge^0(V^\ast)$ with
the stabilizer of $\mathcal{V}(\lambda)_{2N} = \wedge^{2N} (V^\ast)$.
For $n \in 2\mathbb{N}$ these are the groups
\begin{displaymath}
    H = \left\{ \begin{array}{ll} \mathrm{GL}_q \times
    \mathrm{GL}_q \subset G^\prime =\mathrm{GL}_{2q} \;, &\quad
    G = \mathrm{GL}_n \;, \\ \mathrm{GL}_{2q} \subset G^\prime
    = \mathrm{Spin}_{4q} \;, &\quad G = \mathrm{O}_n\;, \\
    \mathrm{GL}_{2q} \subset G^\prime = \mathrm{Sp}_{4q} \;,
    &\quad G = \mathrm{Sp}_n \;. \end{array} \right.
\end{displaymath}
If $n$ is odd, we replace $H$ by the double cover forced on us by the
square root nature of the spinor representation or the highest weight
$\lambda$ being half-integral.

Let us now specify how the Lie algebra $\mathrm{Lie}(H)$ acts on the
spinor module $\wedge(V^\ast)$. (This will do as a temporary
substitute for the more detailed description of the $H$-action on
$W^\ast = \wedge^2 (V^\ast)^G$ given below.) In the first case, one
has $H = \mathrm{GL}(\mathbb{C}^q) \times \mathrm{GL}( (\mathbb{C}^q)
^\ast) \equiv \mathrm{GL}_q \times \mathrm{GL}_q$ and $X = (A , D)
\in \mathrm{Lie}(H)$ acts on $\wedge(V^\ast)$ as
\begin{displaymath}
    \hat{X} = {\textstyle{\frac{1}{2}}}[ \iota(A e_i^c) ,
    \varepsilon(f_c^i) ] + {\textstyle{\frac{1}{2}}}
    [\iota(e_c^i) , \varepsilon(D f_i^c) ] \;,
\end{displaymath}
where the notation of (\ref{eq:hat-H}) is being used. Note that
$\hat{X} . 1 = \frac{n}{2} (\mathrm{Tr}\, A + \mathrm{Tr}\, D)$. In
the last two cases, one uses $V \simeq \mathrm{Hom}( \mathbb{C}^n ,
U_q)$ and fixes a basis $\{ e_i^b \}_{i = 1, \ldots, n}^{b = 1,
\ldots, 2q}$ of $\mathrm{Hom}( \mathbb{C}^n , U_q)$, with dual basis
$\{ f_b^i \}$. With these conventions, an element $X$ of the Lie
algebra of $H = \mathrm{GL}(U_q) \equiv \mathrm{GL}_{2q}$ acts on
$\wedge(V^\ast)$ as $\hat{X} = \frac{1}{2}[\iota(X e_i^b),
\varepsilon(f_b^i)]$. Note $\hat{X}.1 = \frac{n}{2}\mathrm{Tr}\,X\,$.

By definition, the roots of $H$ are the roots of $\mathfrak{g}$ which
are orthogonal to $\lambda\,$. Since all groups $H$ are connected
subgroups of $G^\prime$ of maximal rank, they are in fact
characterized by their root systems. Note also that all of our groups
$H$ are reductive. Furthermore, the character $\chi : \, T \to
\mathbb{C}^\times$ extends to the character $\chi : \, H \to
\mathbb{C}^\times$, $h \mapsto \mathrm{Det}^{n/2}(h)$.

Being orthogonal to the highest weight $\lambda\,$, the root system
of $H$ is orthogonal also to the lowest weight $-\lambda\,$. It
follows that the space of lowest-weight vectors $\mathcal{V}
(\lambda)_{2N}$ is stable with respect to $H$: it is the
one-dimensional representation of $H$ corresponding to the reciprocal
character $\chi^{-1}(h) = \mathrm{Det}^{-n/2}(h)$. Since $H$ is
reductive and the $T$-weight space $\mathcal{V} (\lambda)_{2N} =
\wedge^\mathrm{top} (V^\ast)$ has dimension one, $\mathcal{V}
(\lambda)$ decomposes canonically as a $H$-representation space:
\begin{displaymath}
    \mathcal{V}(\lambda) = \mathcal{V}(\lambda)_{2N} \oplus U \;,
\end{displaymath}
where $U$ is the sum of all other $H$-subrepresentations in
$\mathcal{V}(\lambda)$. From $\mathrm{dim}\, \mathcal{V}
(\lambda)_{2N} = 1$ we then infer that the space of $H$-equivariant
homomorphisms $\mathrm{Hom}_H (\mathcal{V}(\lambda), \mathcal{V}
(\lambda)_{2N})$ is one-dimensional. Now the Berezin integral $\pi :
\, \mathcal{V} (\lambda) \to \mathcal{V}(\lambda)_{2N}$ is a non-zero
element of that space, and we therefore have the following result.
\begin{lem}
$\mathrm{Hom}_H (\mathcal{V}(\lambda),\mathcal{V}(\lambda)_{2N}) =
\mathbb{C}\, \pi\,$.
\end{lem}

\subsection{Parabolic induction}
\label{sect:4.4}

The $\mathfrak{g} $-representation $\mathcal{V}(\lambda)$ can be
constructed in another way, as follows. Decompose $\mathfrak{g}$ as
$\mathfrak{g} = \mathfrak{g}^- \oplus \mathfrak{h} \oplus
\mathfrak{g}^+$ where $\mathfrak{h} = \mathrm{Lie}(H)$ and
$\mathfrak{g}^\pm$ is the direct sum of the root subspaces of
$\mathfrak{g}$ corresponding to positive resp.\ negative roots not
orthogonal to $\lambda$. Since the highest weight $\lambda$ is the
weight of the vacuum with generator $1 \in \mathcal{V}(\lambda)_0\,$,
this implies that $\mathfrak{g}^+ . 1 = 0$ and $\mathfrak{g}^- . 1 =
\mathcal{V}(\lambda)_2\,$. Or, to put it in yet another way,
$\mathfrak{g}^+ \subset \mathfrak{g}$ is the subspace of elements
represented on the spinor module by operators of type $\iota \iota$,
while $\mathfrak{g}^- \subset \mathfrak{g}$ is the subspace of
operators of type $\varepsilon \varepsilon$.

Let $\mathfrak{p} := \mathfrak{h} \oplus \mathfrak{g}^+$. (The
notation is to convey that $\mathfrak{p}$ can be viewed as the Lie
algebra of a parabolic subgroup of $G^\prime\,$.) Since all roots of
$\mathfrak{h}$ are orthogonal to $\lambda$, the weight $\lambda : \,
\mathfrak{t} \to \mathbb{C}$ extends in the trivial way to a linear
function $\lambda : \, \mathfrak{h} \to \mathbb{C}\,$; the latter is
the function $\lambda(X) = (n/2) \mathrm{Tr}_{\mathbb{C}^{2q}} X$. We
further extend $\lambda$ trivially to all of $\mathfrak{p} =
\mathfrak{h} \oplus \mathfrak{g}^+$.

Let $\mathcal{U}(\mathfrak{p})$ be the universal enveloping algebra
of $\mathfrak{p}$ and denote by $V_\lambda := \mathcal{V}(\lambda)_0$
the one-dimensional $\mathcal{U}(\mathfrak{p})$-representation
defined by $X . v_\lambda = \lambda(X) v_\lambda$ for a generator
$v_\lambda \in V_\lambda$ and elements $X \in \mathfrak{p}\,$. Then
by the canonical left action of $\mathfrak{g}$ on $\mathcal{U}(
\mathfrak{g})$, the tensor product
\begin{displaymath}
    M(\lambda) := \mathcal{U}(\mathfrak{g}) \otimes_{\mathcal{U}
    (\mathfrak{p})} V_\lambda
\end{displaymath}
is a $\mathcal{U}(\mathfrak{g})$-representation of highest weight
$\lambda$ and highest-weight vector
\begin{displaymath}
    m_\lambda = 1 \otimes v_\lambda \;.
\end{displaymath}
This representation is called a generalized Verma module or the
universal highest-weight $\mathfrak{g}$-representation which is given
by parabolic induction from the one-dimensional representation
$V_\lambda$ of $\mathfrak{p}\,$. The module $M(\lambda)$ has the
following universal property.
\begin{lem}
Let $W$ be any $\mathfrak{g}$-module with a vector $w \not= 0$ such
that i) $X.w = \lambda(X) w$ for all $X \in \mathfrak{p}\,$, and ii)
$\mathcal{U}(\mathfrak{g}).w = W$. Then there exists a surjective
$\mathfrak{g}$-equivariant linear map $M(\lambda) \to W$ such that
$m_\lambda \mapsto w\,$.
\end{lem}
In particular, our irreducible $\mathfrak{g}$-representation
$\mathcal{V}(\lambda)$ is of this kind. Thus there exists a
surjective $\mathfrak{g}$-equivariant map
\begin{displaymath}
    p \, : \,\, M(\lambda) \to \mathcal{V}(\lambda) \;.
\end{displaymath}

\subsection{$H$-structure of $M(\lambda)$}
\label{sect:4.5}

The $\mathcal{U}(\mathfrak{g})$-representation $M(\lambda)$ has
infinite dimension and cannot be integrated to a representation of
$G^\prime$. As we shall now explain, however, the situation is more
benign for the subgroup $H \subset G^\prime$.

From $\mathfrak{g} / \mathfrak{p} \simeq \mathfrak{g}^-$ we have an
isomorphism of vector spaces $\mathcal{U}( \mathfrak{g})
\otimes_{\mathcal{U}(\mathfrak{p})} V_\lambda \simeq \mathcal{U}(
\mathfrak{g}^-)$. By making the identification as
\begin{displaymath}
    \mathcal{U}(\mathfrak{g}^-) \otimes V_\lambda \simeq M(\lambda)
    \;, \quad n \otimes v_\lambda \mapsto n m_\lambda \;,
\end{displaymath}
we will actually get more, as follows. If $\alpha$, $\beta$ are any
two roots such that $\alpha$ is orthogonal to $\lambda$ and $\beta$
is not, then $\alpha + \beta$ is not orthogonal to $\lambda\,$. From
this one directly concludes that if $h \in \mathfrak{h}$ and $n \in
\mathfrak{g}^-$, then $[h,n] \in \mathfrak{g}^-$, i.e.,
$\mathfrak{g}^-$ (or $\mathfrak{g}^+$, for that matter) is normalized
by $\mathfrak{h}\,$. This action of $\mathfrak{h}$ on
$\mathfrak{g}^-$ extends to an $\mathfrak{h}$-action on
$\mathcal{U}(\mathfrak{g}^-)$: supposing that $n = n_1 \cdots n_r \in
\mathcal{U}(\mathfrak{g}^-)$ where $n_i \in \mathfrak{g}^-$ ($i = 1,
\ldots, r$), we let
\begin{displaymath}
    \mathrm{ad}(h) n := \sum\nolimits_j n_1 \cdots n_{j-1}
    [ h , n_j ] n_{j+1} \cdots n_r \;,
\end{displaymath}
and by this definition we have the following commutation rule of
operators in $\mathcal{U}(\mathfrak{g})$:
\begin{displaymath}
    h n = \mathrm{ad}(h) n + n h \;.
\end{displaymath}
If we now let $\mathcal {U} (\mathfrak{h})$ act on $M(\lambda)$ by
the canonical left action and on $\mathcal{U}( \mathfrak{g}^-)
\otimes V_\lambda$ by
\begin{displaymath}
    h . (n \otimes v_\lambda) := (\mathrm{ad}(h) n) \otimes
    v_\lambda + \lambda(h) n \otimes v_\lambda \;,
\end{displaymath}
then we see that the identification $\mathcal{U}(\mathfrak{g}^-)
\otimes V_\lambda \stackrel{\sim}{\to} M(\lambda)$ by $n \otimes
v_\lambda \mapsto n m_\lambda$ is an isomorphism of $\mathcal{U}
(\mathfrak{h})$-representations.

Now every element in $\mathcal{U} (\mathfrak{g}^-)$ lies in an
$\mathfrak{h}$-representation of finite dimension. Basic principles
therefore entail the following consequence.
\begin{lem} The representation of the Lie algebra $\mathfrak{h}$
on $M(\lambda)$ can be integrated to a representation of the Lie
group $H$.
\end{lem}
In each of the three cases under consideration, $\mathfrak{g}^-$ is
commutative and we can identify $\mathcal{U}(\mathfrak{g}^-)$ with
the ring of polynomial functions $\mathbb{C}[W]$ on a suitable
representation $W \simeq (\mathfrak{g}^-)^\ast$ of $H\,$. From
$\mathfrak{g}^- \simeq \mathfrak{g}^- . 1 = \mathcal{V}(\lambda)_2 =
\wedge^2(V^\ast)^G$ we have $W = \wedge^2(V)^G$, the subspace of
$G$-fixed vectors in $\wedge^2(V)$. The space $W$ was described in
Lemma \ref{lem:FF-invs} where we saw that $W = \mathrm{End} (\mathbb
{C}^q)$, $\mathrm{Sym}_a(U_q)$, and $\mathrm{Sym}_s(U_q)$ for $G =
\mathrm {GL}$, $\mathrm{O}$, and $\mathrm{Sp}$, respectively. Note
that in all cases $W$ contains the identity element $\mathrm{Id} =
\mathrm {Id}_{\mathbb {C}^q}$ or $\mathrm{Id} = \mathrm{Id}_{U_q}\,$.

We now describe how the group $H$ acts on $W$. In the last two cases,
where $H = \mathrm{GL}(U_q) \equiv \mathrm{GL}_{2q}$ (or a double
cover thereof), we associate with each of the two bilinear forms $b =
s$ or $b = a$ an involution $\tau_b : \, H \to H$ by the equation
\begin{displaymath}
    b(\tau_b(g)x\, ,y) = b(x\, ,g^{-1}y) \qquad (x\, ,y \in U_q)\;.
\end{displaymath}
The action of $H$ on $W = \mathrm{Sym}_b (U_q)$ is then by twisted
conjugation, $g . w = g w \tau_b(g^{-1})$. In the notation of Sect.\
\ref{sect:orthsymp} we have $\tau_a(g^{-1}) = t_a \, g^\mathrm{t}
(t_a)^{-1}$ and $\tau_s(g^{-1}) = t_s\, g^\mathrm{t} (t_s)^{-1}$.
Note that the group of fixed points of $\tau_b$ in $H$ is the
symplectic group $\mathrm{Sp}(U_q) = \mathrm{Sp}_{2q}$ for $b = a$
and the orthogonal group $\mathrm{O}(U_q) = \mathrm{O}_{2q}$ for $b =
s\,$.

In the first case ($G = \mathrm{GL}_n$) the group $H$ is the subgroup
of $\mathrm{GL}(U_q)$ preserving the decomposition $U_q =\mathbb{C}^q
\oplus (\mathbb{C}^q)^\ast$. Here again it will be best to think of
the vector space $W$ as the intersection of the vector spaces for the
other two cases:
\begin{displaymath}
    W = \mathrm{End}(\mathbb{C}^q) \simeq
    \mathrm{Sym}_a(U_q)\cap \mathrm{Sym}_s(U_q) \;.
\end{displaymath}
This means that we think of $\mathrm{End}(\mathbb{C}^q)$ as being
embedded into $\mathrm{End}(U_q)$ as
\begin{displaymath}
    \mathrm{End}(\mathbb{C}^q) \to \begin{pmatrix} \mathrm{End}
    (\mathbb{C}^q) &\mathrm{Hom}((\mathbb{C}^q)^\ast,\mathbb{C}^q)\\
    \mathrm{Hom}(\mathbb{C}^q , (\mathbb{C}^q)^\ast) &\mathrm{End}(
    (\mathbb{C}^q)^\ast) \end{pmatrix} \;, \quad z \mapsto
    \begin{pmatrix} z &0\\ 0 &z^\mathrm{t} \end{pmatrix} =: w \;.
\end{displaymath}
For an element $g = (g_1 , g_2) \in H = \mathrm{GL}(\mathbb{C}^q)
\times \mathrm{GL}((\mathbb{C}^q)^\ast)$ one now has $\tau_a(g_1,
g_2)^{-1} = \tau_s(g_1 ,g_2)^{-1} = (g_2^\mathrm{t},g_1^\mathrm{t})$,
and the action of $H$ on $W$ by twisted conjugation is given by
\begin{displaymath}
    g.w = (g_1,g_2) . \begin{pmatrix} z &0\\ 0 &z^\mathrm{t}
    \end{pmatrix} = \begin{pmatrix} g_1^{\vphantom{\mathrm{t}}}
    z \, g_2^\mathrm{t} &0\\ 0 &g_2^{\vphantom{\mathrm{t}}}
    z^\mathrm{t} g_1^\mathrm{t} \end{pmatrix} \;.
\end{displaymath}
If we now define an involution $\tau_0$ by $\tau_0(g_1^{-1} ,
g_2^{-1}) = (g_2^\mathrm{t},g_1^\mathrm{t})$, then this action can be
written in the short form $g.w = g w \tau_0(g^{-1})$.

To sum up the situation, let $\tau = \tau_0$ for $G = \mathrm{GL}\,$,
$\tau = \tau_a$ for $G = \mathrm{O}\,$, and $\tau = \tau_s$ for $G =
\mathrm{Sp}\,$. Then the $H$-action on $W$ always takes the form
\begin{displaymath}
    g.w = g w \tau(g^{-1}) \;.
\end{displaymath}
In all three cases it is a well-known fact (see for example
\cite{Howe1995}) that the ring $\mathbb{C}[W]$ is multiplicity-free
as a representation space for $H$. It then follows that the universal
highest-weight representation $M(\lambda) \simeq \mathbb{C}[W]
\otimes V_\lambda$ is multiplicity-free.
\begin{lem}
Let $V_{-\lambda} = \mathcal{V}(\lambda)_{2N}$ be the one-dimensional
$H$-representation associated to the character $\chi^{-1} = \exp
\circ (-\lambda) \circ \ln\,$. Then the space
\begin{displaymath}
    \mathrm{Hom}_H( M(\lambda) , V_{-\lambda} )
\end{displaymath}
of $H$-equivariant homomorphisms from $M(\lambda)$ to $V_{-\lambda}$
has dimension one and is generated by $\pi\circ p\,$, the composition
of the projection $p : \, M(\lambda) \to \mathcal{V}(\lambda)$ with
the Berezin integral $\pi : \, \mathcal{V}(\lambda) \to V_{-\lambda}
= \mathcal{V}(\lambda)_{2N}\,$.
\end{lem}
\begin{proof}
Since $p$ is surjective, $\pi \circ p$ is non-trivial and the space
$\mathrm{Hom}_H (M(\lambda),V_{-\lambda})$ has at least dimension
one. On the other hand, since $M(\lambda)$ is multiplicity-free as an
$H$-representation, the dimension of $\mathrm{Hom}_H( M(\lambda),
V^\prime)$ cannot be greater than one for any irreducible
representation $V^\prime$ of $H\,$.
\end{proof}
\begin{cor}\label{cor:one-dim}
Let $P : \, M(\lambda) = \mathbb{C}[W] \otimes V_\lambda \to
V_{-\lambda}$ be any non-trivial $H$-equivariant linear mapping. Then
there exists a non-zero constant $c_P$ such that for every $f \in
\mathcal{V}(\lambda)$ and any lift $F \in p^{-1}(f) \subset
\mathbb{C}[W] \otimes V_\lambda$ one has
\begin{displaymath}
    P[F] = c_P \, \Omega[f] \;.
\end{displaymath}
\end{cor}
We are now going to realize $P$ by integration over a real domain in
$W$.

\subsection{Construction of $H$-equivariant homomorphisms}
\label{sect:4.6}

Let $o \equiv \mathrm{Id}$ denote the identity element of $W$. Then
in all three cases the $H$-orbit $H . o$ is open and dense in $W$ and
can be characterized as the complement of the zero set of a
polynomial $D:$
\begin{displaymath}
    H . o = \{w \in W \mid D(w)\not = 0 \} \;,
\end{displaymath}
where $D$ will be the Pfaffian function for the case of $G = \mathrm
{O}_n$ and will be the determinant function for $G = \mathrm{GL}_n$
and $G = \mathrm{Sp}_n\,$. Since $D$ does not vanish on $H.o\,$, the
map
\begin{displaymath}
    H.o \hookrightarrow W \oplus \mathbb{C} \;,
    \quad w \mapsto (w , D(w)^{-1})
\end{displaymath}
defines an inclusion, and one can view $H . o$ as the zero set of a
function on $W \oplus \mathbb{C}:$
\begin{displaymath}
    H.o = \{(w,t)\in W \oplus\mathbb{C} \mid D(w)\, t - 1 = 0 \} \;.
\end{displaymath}
Hence the ring of algebraic functions on $H.o\,$, namely $\mathbb{C}
[H.o]$, is the same as the ring $\mathbb{C}[W\oplus \mathbb{C}]$
factored by the ideal generated by the function $(w,t) \mapsto D(w)\,
t - 1$.

Let $\mathbb{C}(W)$ be the field of rational functions on $W$. Thus
an element $r \in \mathbb{C}(W)$ can be expressed as a quotient $r =
f / g$ of polynomial functions $f , g \in \mathbb{C}[W]$. Denote by
$\mathbb{C}[W]_D \subset \mathbb{C}(W)$ the subring of elements which
can be written as a quotient $r = f / D^n$,
\begin{displaymath}
    \mathbb{C}[W]_D = \left\{ w \mapsto f(w) / D^n(w) \mid f \in
    \mathbb{C}[W] \;, \, n \ge 0 \right\} \;.
\end{displaymath}
We now identify $\mathbb{C}[H.o]$ with $\mathbb{C}[W]_D:$ an element
in $\mathbb{C}[H.o]$ can be represented as a polynomial of the form
$(w,t) \mapsto f_0(w) + f_1(w)\, t + \ldots + f_s(w)\, t^s$, where
the $f_i \in \mathbb{C}[W]$, and the following map $\mathbb{C}[H .
o]\rightarrow \mathbb{C}[W]_D$ then defines an isomorphism of rings:
\begin{displaymath}
    \Big( (w,t) \mapsto f_0(w) + f_1(w)\, t + \ldots + f_s(w)\,t^s \Big)
    \mapsto \frac{f_0}{1} + \frac{f_1}{D} + \ldots + \frac{f_s}{D^s} \;.
\end{displaymath}

Our aim here is to construct an $H$-equivariant homomorphism
$M(\lambda) \to V_{-\lambda}$. By the isomorphisms of
$H$-representations: $(V_{-\lambda})^\ast \simeq V_\lambda$ and
$V_\lambda \otimes V_\lambda \simeq V_{2\lambda}$, this is the same
as constructing an $H$-invariant homomorphism
\begin{displaymath}
    M(\lambda) \otimes (V_{-\lambda})^\ast =
    \mathbb{C}[W] \otimes V_{2\lambda} \to \mathbb{C}
\end{displaymath}
to the trivial representation. For this purpose we identify the
representation $\mathbb{C}[W] \otimes V_{2\lambda}$ with the
following subspaces of $\mathbb{C}[W]_D\,$:
\begin{displaymath}
    \mathbb{C}[W] \otimes V_{2\lambda} \simeq \left\{
    \begin{array}{ll} \mathrm{Det}^{-n} \, \mathbb{C}[W] \subset
    \mathbb{C}[W]_\mathrm{Det}\;, &\quad G = \mathrm{GL}_n \;, \\
    \mathrm{Pfaff}^{\, -n} \, \mathbb{C}[W] \subset
    \mathbb{C}[W]_\mathrm{Pfaff}\;, &\quad G = \mathrm{O}_n \;,\\
    \mathrm{Det}^{-n/2} \, \mathbb{C}[W] \subset
    \mathbb{C}[W]_\mathrm{Det} \;, &\quad G = \mathrm{Sp}_n \;.
    \end{array} \right.
\end{displaymath}
In the first case, this identification is correct because $g = (g_1 ,
g_2) \in H$ operates on the determinant function $\mathrm{Det}^{-n} :
\, W = \mathrm{End}(\mathbb{C}^q) \to \mathbb{C}$ as
\begin{displaymath}
    (g . \, \mathrm{Det}^{-n})(z) = \mathrm{Det}^{-n} (g_1^{-1} z
    \, (g_2^\mathrm{t})^{-1}) = \mathrm{Det}^n(g_1) \mathrm{Det}^n(g_2)
    \mathrm{Det}^{-n}(z) \;,
\end{displaymath}
which is the desired behavior since $ \mathrm{Det}^n( g_1^{
\vphantom{\mathrm{t}}} g_2^\mathrm{t})$ agrees with the character
$\chi(g_1,g_2)^2$ associated with the $H$-representation
$V_{2\lambda}\, $. In the other two cases we have
\begin{displaymath}
    (g . \, \mathrm{Det}^{-n/2})(w) = \mathrm{Det}^{-n/2} (g^{-1} w
    \, \tau_b(g)^{-1}) = \mathrm{Det}^n(g) \mathrm{Det}^{-n/2}(w)
    \quad (b = a, s)\;.
\end{displaymath}
Again, this is as it should be since $\mathrm{Det}^n(g) = \chi(g)^2$
is the desired character for $V_{2\lambda}$. Here one should bear in
mind that $n$ is always an even number in the third case, and that
$\sqrt{\mathrm{Det}} = \mathrm{Pfaff}$ in the second case.

Recall now that $o \equiv \mathrm{Id}$ denotes the identity element
in $W$, let $H_o = \{ h \in H \mid h . o = o \}$ be the isotropy
group of $o\,$, and observe that $H . o$ is isomorphic to $H /
H_o\,$,
\begin{displaymath}
    H/H_o = \left\{ \begin{array}{ll} (\mathrm{GL}_q \times
    \mathrm{GL}_q ) / \mathrm{GL}_q \;, &\quad G =
    \mathrm{GL}_n \;, \\ \mathrm{GL}_{2q} / \mathrm{Sp}_{2q}\;,
    &\quad G = \mathrm{O}_n \;, \\ \mathrm{GL}_{2q} / \mathrm{O}_{2q}
    \;, &\quad G = \mathrm{Sp}_n \;. \end{array} \right.
\end{displaymath}
Then fix some maximal compact subgroup $K \subset H$ such that the
isotropy group $K_o \subset K$ is a maximal compact subgroup of the
stabilizer $H_o \subset H$. Since $H_o$ is a reductive group, we have
natural isomorphisms $ \mathbb{C}[W]_D = \mathbb{C}[H / H_o] =
\mathbb{C}[H]^{H_o}$. The maximal compact subgroup $K\subset H$ is
Zariski dense, i.e., a polynomial function that vanishes on $K$ also
vanishes on $H$, so $\mathbb{C}[H]= \mathbb{C}[K]$. For the same
reason, given any locally finite-dimensional $H$-representation, the
subspace of $H_o$-fixed points coincides with the $K_o$-fixed points.
Summarizing, we have
\begin{displaymath}
    \mathbb{C}[W] \otimes V_{2\lambda} \stackrel{i}{\hookrightarrow}
    \mathbb{C}[W]_D = \mathbb{C}[H / H_o] = \mathbb{C}[H]^{H_o} =
     \mathbb{C}[H]^{K_o} =\mathbb{C}[K]^{K_o} \;.
\end{displaymath}
The benefit from this sequence of identifications is that on the
space $\mathbb{C}[K]$ there exists a natural and non-trivial
$K$-invariant projection. Indeed, letting $dk$ be a Haar measure on
$K$, we may view $f \in \mathbb{C}[W]_D$ as a function on $K$ and
integrate:
\begin{displaymath}
    \mathbb{C}[W]_D \to \mathbb{C} \;, \quad
    f \mapsto \int_K f(k . o) \, dk \;.
\end{displaymath}
This projection is $K$-invariant, and its restriction to our space
$\mathbb{C}[W] \otimes V_{2\lambda}$ is still non-trivial. In the
first case this is because $\mathrm{Det}^n \in \mathbb{C}[W]$ and
hence $\mathbb{C}\, 1 \subset \mathrm{Det}^{-n} \, \mathbb{C}[W]$; in
the last two cases $\mathrm{Det}^{n/2} \in \mathbb{C}[W]$ and hence
$\mathbb{C}\, 1 \subset \mathrm{Det}^{-n/2} \, \mathbb{C}[W]$.

We can now reformulate Cor.\ \ref{cor:one-dim} to make the statement
more concrete.
\begin{prop}\label{thm:4.9}
For each case $G = \mathrm{GL}_n\,$, $\mathrm{O}_n\,$, or
$\mathrm{Sp}_n\,$, there is a choice of normalized Haar measure $dk$
so that the following holds for all $f \in \wedge(V^\ast)^G$. If $F
\in \mathbb{C}[W]$ is any lift w.r.t.\ the identification and
projection $\mathbb{C}[W] \simeq \mathbb{C}[W] \otimes V_\lambda \to
\wedge(V^\ast)^G$ then
\begin{displaymath}
    \Omega_V[f] = \int_K F(k.o) \mathrm{Det}^{-n}(k)\, dk \;.
\end{displaymath}
\end{prop}
Next, let us give another version of the `bosonization' formula of
Prop.\ \ref{thm:4.9} in order to get a better match with the
supersymmetric formula to be developed below. For that, notice that
$K_o = \mathrm{USp}_{2q}$ for $G = \mathrm{O}_n\,$, $K_o =
\mathrm{O}_{2q}(\mathbb{R})$ for $G = \mathrm{Sp}_n\,$, and $K_o =
\mathrm{U}_q$ acting by elements $(k,(k^{-1} )^\mathrm{t})$ for $G =
\mathrm{GL}_n\,$. From this we see that $\mathrm{Det}^{-n}(k) = 1$
for $k \in K_o$ in all cases. We can therefore push down the integral
over $K$ to an integral over the orbit $K.o \simeq K/K_o\,$. We
henceforth denote this orbit by $D_q := K.o$. Writing $y := k
\tau(k^{-1})$ for $G = \mathrm{O}$, $\mathrm {Sp}$ we have
$\mathrm{Det}^{-n}(k) = \mathrm{Det}^{-n/2} (k\tau(k)^{ -1}) =
\mathrm{Det}^{-n/2}(y)$. Similarly, letting $y := k_1 (k_2)
^\mathrm{t}$ for $G = \mathrm{GL}$ we have $\mathrm {Det}^{-n} (k) =
\mathrm{Det}^{-n}(k_1) \mathrm{Det}^{-n}(k_2)= \mathrm{Det}^{-n}
(y)$. Thus the relation $\mathrm{Det}^{-n}(k) = \mathrm{Det}^{-
n^\prime}(y)$ always holds if we set $n^\prime = (1+|m|)^{-1} n$,
i.e., $n^\prime = n$ for $G = \mathrm{GL}$ and $n^\prime = n/2$ for
$G = \mathrm{O}$, $\mathrm {Sp}$.

Let now $d\mu_{D_q}$ denote a $K$-invariant measure on the $K$-orbit
$D_q\,$. According to our general conventions, $d\mu_{D_q}$ is
normalized in such a way that $(d\mu_{D_q})_o$ coincides with the
Euclidean volume density on $T_o D_q = W \cap \mathrm{Lie}(K)$ which
is induced by the quadratic form $A \mapsto - \mathrm{Tr}
_{\mathbb{C}^p}\, A^2$ for $G = \mathrm{GL}$ and $A\mapsto - \frac{1}
{2} \mathrm{Tr}_{U_p}\, A^2$ for $G = \mathrm{O}$, $\mathrm{Sp}\,$.
We denote this by $A \mapsto - \mathrm{Tr}^{\,\prime} A^2$ for short.
Pushing the Haar measure $dk$ forward by $K \to K.o = D_q$ we obtain
$d\mu_{D_q}$ times a constant. Thus the formula of Prop.\
\ref{thm:4.9} becomes
\begin{equation}\label{eq:4.2}
    \Omega_V[f] = \mathrm{const} \times \int_{D_q} F(y)\,
    \mathrm{Det}^{-n^\prime}(y)\, d\mu_{D_q}(y) \;.
\end{equation}

To determine the unknown constant of proportionality, it suffices to
compute both sides of the equation for some special choice of $f$
(and a corresponding function $F$). If we choose $F(y) = \mathrm{e}^{
\mathrm{Tr}^{\,\prime} y}$, then $f$ is simply a Gaussian with
Berezin integral $\Omega_V[f] = 1$. However, the integral on the
right-hand side is not quite so easy to do. We postpone this
computation until the end of the paper, where we will carry it out
using a supersymmetric reduction technique based on relations
developed below. To state the outcome, we recall the definition of
the groups $K_{n,p}$ and let the sign of the positive integer $p$ now
be reversed; according to Table \ref{fig:1} of Sect.\ \ref{sect:3.4}
this means that $K_{n,-p} = \mathrm{U}_{n+p}\,$, $\mathrm{O}_{n+2p}
(\mathbb{R})$, and $\mathrm {USp}_{n+2p}$ for $G = \mathrm{GL}$,
$\mathrm{O}$, and $\mathrm{Sp}\,$.
\begin{lem}\label{lem:4.10}
$\int_{D_q} \mathrm{e}^{\mathrm{Tr}^{\,\prime} y}\, \mathrm{Det}^{
-n^\prime}(y) \, d\mu_{D_q}(y) = (2\pi)^{-qn} 2^{-qm} \mathrm{vol}
(K_{n,-q}) / \mathrm{vol}(K_n)$.
\end{lem}
\begin{rem}
The similarity of this formula with Eq.\ (\ref{eq:3.2}) is not an
accident; in fact, in Sect.\ \ref{sect:norm-int} we will establish
Lemma \ref{lem:4.10} by reduction to the latter result.
\end{rem}
Using Lemma \ref{lem:4.10} we can now eliminate the unknown constant
of proportionality from (\ref{eq:4.2}). To state the resulting
reformulation of Prop.\ \ref{thm:4.9}, we will use the surjective
mapping $Q^\ast:\, \mathcal{O}(W) \to \wedge(V^\ast)^G$ defined in
Sect.\ \ref{sect:FF-invs}.
\begin{thm}\label{thm:4.10}
For $f \in \wedge(V^\ast)^{G_n}$, if $F \in (Q^\ast)^{-1}(f) \in
\mathcal{O}(W)$ is any holomorphic function in the inverse image of
$f$, the Berezin integral $f \mapsto \Omega_V[f]$ can be computed as
an integral over the compact symmetric space $D_q \simeq K / K_o :$
\begin{displaymath}
    \Omega_V[f] = (2\pi)^{qn} 2^{qm} \frac{\mathrm{vol}(K_n)}
    {\mathrm{vol}(K_{n,-q})}\int_{D_q} F(y)\, \mathrm{Det}^{-n^\prime}
    (y)\,d\mu_{D_q}(y) \;,
\end{displaymath}
where $n^\prime = n$ for $G_n = \mathrm{GL}_n$ and $n^\prime = n/2$
for $G_n = \mathrm{O}_n\,$, $\mathrm{Sp}_n \,$.
\end{thm}

\subsection{Shifting by nilpotents}

In this last subsection, we derive a result which will be needed in
the supersymmetric context of Sects.\ \ref{sect:thm} and
\ref{sect:norm-int}. Let $\mathfrak {p} := T_o D_q$ be the tangent
space of $D_q = K . o$ at the identity $\mathrm{Id} = o\,$. Since $H
= K^\mathbb{C}$, $H_o = K_o^\mathbb {C}$, $D_q \simeq K/K_o\,$, and
$W$ is the closure of $H/H_o\,$, the tangent space $\mathfrak{p}$ is
a real form of $W$. More precisely,
\begin{displaymath}
    \mathfrak{p} = W \cap \mathrm{Lie}(K) \;.
\end{displaymath}
Linearizing $\tau(y) = \tau(k \tau(k)^{-1}) = \tau(k) k^{-1} =
y^{-1}$ at $y = o\,$, we note that $\xi \in \mathfrak{p}$ satisfies
$\tau_\ast(\xi) = - \xi$, where $\tau_\ast$ is the differential
$\tau_\ast := d\tau\vert_{o}\,$.

Let now $D_q$ be equipped with the canonical Riemannian geometry in
which $K$ acts by isometries and which is induced by the trace form
$-\mathrm{Tr}$ on $\mathfrak{p}\,$. With each $w \in W$ associate a
complex vector field $t_w \in \Gamma(D_q\, , \mathbb{C} \otimes T
D_q)$ by
\begin{displaymath}
    (t_w f)(y) = \frac{d}{dt} f(y + tw) \big\vert_{t = 0} \;.
\end{displaymath}
\begin{lem}\label{lem:divergence}
The vector field $t_w$ has the divergence
\begin{displaymath}
    \mathrm{div}(t_w)(y) = - (q-m/2)\, \mathrm{Tr}\,(y^{-1} w)\;,
\end{displaymath}
where $m = 0,+1,-1$ for $G = \mathrm{GL}\,$, $\mathrm{O}\,$,
$\mathrm{Sp}\,$.
\end{lem}
\begin{proof}
For $y \in D_q$ choose some fixed element $k \in K$ so that $y = k
\tau(k^{-1})$. Fixing an orthonormal basis $\{ e_\alpha \}$ of
$\mathfrak{p}\,$, define local coordinate functions $x^\alpha$ in a
neighborhood of the point $y$ by the equation
\begin{displaymath}
    y^\prime = k\, \mathrm{e}^{x^\alpha(y^\prime) e_\alpha}
    \tau(k^{-1})\;.
\end{displaymath}
By their construction via the exponential mapping, these are Riemann
normal coordinates centered at $y\,$. The Riemannian metric expands
around $y$ as $\sum_\alpha dx^\alpha \otimes dx^\alpha + \ldots$,
with vanishing corrections of linear order in the coordinates
$x^\alpha$.

Let $\partial_\alpha = \partial/\partial x^\alpha$ and express the
vector field $t_w$ in this basis as $t_w = t_w^\alpha \partial_\alpha
\,$. Differentiating the equation $k^{-1} (y^\prime + t w) \tau(k) =
\mathrm{e}^{x^\alpha(y^\prime + tw) e_\alpha}$ with the help of the
relation
\begin{displaymath}
    \frac{d}{dt} x^\alpha(y^\prime + tw) \big\vert_{t = 0} =
    (\mathcal{L}_{t_w} x^\alpha)(y^\prime) = t_w^\alpha (y^\prime)\;,
\end{displaymath}
where $\mathcal{L}_{t_w}$ is the Lie derivative w.r.t.\ the vector
field $t_w\,$, and then solving the resulting equation for
$t_w^\alpha\,$, one obtains the following expansion of $t_w^\alpha$
in powers of the $x^\beta$:
\begin{displaymath}
    t_w^\alpha(\cdot) = - \mathrm{Tr}\, (k^{-1} w \tau(k) e_\alpha)
    + \frac{1}{2} x^\beta(\cdot) \mathrm{Tr}\, \big( k^{-1} w \tau(k)
    (e_\alpha e_\beta + e_\beta e_\alpha) \big) + \ldots \;.
\end{displaymath}
Since the metric tensor is of the locally Euclidean form $\sum_\alpha
dx^\alpha \otimes dx^\alpha + \ldots$, the divergence is now readily
computed to be
\begin{displaymath}
    \mathrm{div}(t_w)(k \tau(k)^{-1}) = (\partial_\alpha t_w^\alpha)
    (k \tau(k)^{-1}) = \sum\nolimits_\alpha \mathrm{Tr}\, (k^{-1}
    w \tau(k) e_\alpha^2 ) \;.
\end{displaymath}
The sum of squares $e_\alpha^2$ is independent of the choice of basis
$e_\alpha\,$. Making any convenient choice, a short computation shows
that
\begin{displaymath}
    - \sum\nolimits_\alpha e_\alpha^2 = (q - m/2)\, \mathrm{Id} \;,
\end{displaymath}
where $m = 0,+1,-1$ for $G = \mathrm{GL}$, $\mathrm{O}$, $\mathrm
{Sp}\,$. The statement of the lemma now follows by inserting this
result in the previous formula and recalling $y = k \tau(k^{-1})$.
\end{proof}
\begin{rem}
A check on the formula for the sum of squares $- \sum_\alpha
e_\alpha^2$ is afforded by the relations $- \mathrm{Tr}\, e_\alpha^2
= 1$ and $\mathrm{dim}_\mathbb{R} \, \mathfrak{p} =
\mathrm{dim}_\mathbb{C} \, W = q^2$, $2q(q-1/2)$, $2q (q+1/2)$ for
$G= \mathrm{GL}$, $\mathrm{O}$, $\mathrm{Sp}\,$. Note also this:
defined by the equation $\mathrm{div}(t_w) d\mu_{D_q} =
\mathcal{L}_{t_w} d\mu_{D_q}$, the operation of taking the divergence
does not depend on the choice of scale for the metric tensor.
Therefore we were free to use a normalization convention for the
metric which differs from that used elsewhere in this paper.
\end{rem}
\begin{lem}\label{lem:shift}
Let $F : \, D_q \to \mathbb{C}$ be an analytic function, and let $N_0
= \oplus_{k \ge 1} \wedge^{2k}(\mathbb{C}^\bullet)$ be the nilpotent
even part of a (parameter) Grassmann algebra $\wedge
(\mathbb{C}^\bullet)$. Then for any $w \in N_0 \otimes W$ one has
\begin{displaymath}
    \int_{D_q} F(y+w) \, d\mu_{D_q}(y) = \int_{D_q} \frac{F(y)\,
    d\mu_{D_q}(y)}
    {\mathrm{Det}^{q+m/2}(\mathrm{Id}-y^{-1}w)} \;.
\end{displaymath}
\end{lem}
\begin{proof}
Let $t_w$ be the vector field generating translations $y \mapsto y +
sw$ ($s \in \mathbb{R}$). Since $w$ is nilpotent, the exponential
$\exp(s \mathcal{L}_{t_w})$ of the Lie derivative $\mathcal{L}_{t_w}$
is a differential operator of finite order. Applying it to the
function $F$ one has $(\mathrm{e}^{s \mathcal{L}_{t_w}} F)(y) =
F(y+sw)$.

Now for any density $\Omega$ on $D_q$ the integral $\int_{D_q}
\mathcal{L}_{t_w} \Omega$ vanishes by Stokes' theorem for the closed
manifold $D_q\,$. Therefore, partial integration gives
\begin{displaymath}
    \int_{D_q} F(y+sw)\, d\mu_{D_q}(y) = \int_{D_q} F(y)\,
    \mathrm{e}^{-s \mathcal{L}_{t_w}}\, d\mu_{D_q}(y) =
    \int_{D_q} F(y)\, J_s(y) \, d\mu_{D_q}(y) \;,
\end{displaymath}
where $J_s : \, D_q \to \mathbb{C} \oplus N_0$ is the function
defined by $\mathrm{e}^{-s\mathcal{L}_{t_w}} d\mu_{D_q} = J_s \,
d\mu_{D_q}\,$.

We now set up a differential equation for $J_s\,$. For this we
consider the derivative
\begin{displaymath}
    \frac{d}{ds} \left( \mathrm{e}^{-s\mathcal{L}_{t_w}} d\mu_{D_q}
    \right) = \mathrm{e}^{-s\mathcal{L}_{t_w}} \left( -
    \mathcal{L}_{t_w} d\mu_{D_q} \right) \;.
\end{displaymath}
By the relation $\mathcal{L}_{t_w} d\mu_{D_q} = \mathrm{div}(t_w)
d\mu_{D_q}$ we then get
\begin{displaymath}
    \frac{d}{ds} J_s\, d\mu_{D_q} = - \mathrm{e}^{-s\mathcal{L}_{t_w}}
    \left( \mathrm{div}(t_w) d\mu_{D_q} \right) = -
    \mathrm{e}^{-s\mathcal{L}_{t_w}}
    \left( \mathrm{div}(t_w) \right) J_s\, d\mu_{D_q} \;.
\end{displaymath}
Using the expression for $\mathrm{div}(t_w)$ from Lemma
\ref{lem:divergence} we obtain the differential equation
\begin{displaymath}
    \frac{d}{ds} \log J_s(y) = (q-m/2) \mathrm{Tr}\,(w(y-sw)^{-1})
    = - (q-m/2) \frac{d}{ds} \mathrm{Tr}\, \log(y-sw) \;.
\end{displaymath}
The solution of this differential equation with initial condition
$J_{s=0} = 1$ is
\begin{displaymath}
    J_s(y) = \frac{ \mathrm{Det}^{q-m/2}(y)}
    {\mathrm{Det}^{q-m/2}(y-sw)} \;,
\end{displaymath}
and setting $s = 1$ yields the statement of the lemma.
\end{proof}

\section{Full supersymmetric situation}
\label{sect:full-susy}
\setcounter{equation}{0}

We finally tackle the general situation of $V = V_0 \oplus V_1$ where
both $V_0$ and $V_1$ are non-trivial. The superbosonization formulas
(\ref{bosonize}, \ref{bos-other}) in this situation will be proved by
a chain of variable transformations resulting in reduction to the
cases treated in the two preceding sections. This proof has the
advantage of being constructive.

\subsection{More notation}

To continue the discussion in the supersymmetric context we need some
more notation. If $V = V_0 \oplus V_1$ is a $\mathbb{Z}_2$-graded
vector space, one calls $(V_0 \cup V_1) \setminus \{ 0 \}$ the subset
of homogeneous elements of $V$. A vector $v \in V_0 \setminus \{ 0
\}$ is called even and $v \in V_1 \setminus \{ 0 \}$ is called odd.
On the subset of homogeneous elements of $V$ one defines a parity
function $\vert \cdot \vert$ by $\vert v \vert = 0$ for $v$ even and
$\vert v \vert = 1$ for $v$ odd. Whenever the parity function $v
\mapsto \vert v \vert$ appears in formulas and expressions, the
vector $v$ is understood to be homogeneous even without explicit
mention.

There exist two graded-commutative algebras that are canonically
associated with $V = V_0 \oplus V_1\,$. To define them, let $T(V) =
\oplus_{k=0}^\infty T^k(V)$ be the tensor algebra of $V$, and let
$I_\pm(V) \subset T(V)$ be the two-sided ideal generated by
multiplication of $T(V)$ with all combinations $v \otimes v^\prime
\pm (-1)^{|v| |v^\prime|} v^\prime \otimes v$ for homogeneous vectors
$v, v^\prime \in V$. Then the \emph{graded-symmetric} algebra of $V =
V_0 \oplus V_1$ is the quotient
\begin{displaymath}
  \mathrm{S}(V) := T(V)/I_-(V) \simeq
  \mathrm{S}(V_0)\otimes \wedge(V_1)\;,
\end{displaymath}
which is isomorphic to the tensor product of the symmetric algebra of
$V_0$ with the exterior algebra of $V_1\,$. The
\emph{graded-exterior} algebra of $V$ is the quotient
\begin{displaymath}
  \wedge(V) := T(V)/I_+(V) \simeq
  \wedge(V_0)\otimes \mathrm{S}(V_1)\;.
\end{displaymath}
Here we have adopted Kostant's language and notation \cite{kostant}.

Recall that our goal is to prove an integration formula for
integrands in $\mathcal{A}_V^G$, the graded-commutative algebra of
$G$-equivariant holomorphic functions $f : \, V_0 \to \wedge
(V_1^\ast)$ with $V_0$ and $V_1$ given in (\ref{eq:2.3},
\ref{eq:2.4}). For that purpose we will view the basic algebra
$\mathcal{A}_V$ as a completion of the graded-symmetric algebra
\begin{displaymath}
  \mathrm{S}(V^\ast) = T(V^\ast)/I_-(V^\ast) \simeq
  \mathrm{S}(V_0^\ast)\otimes \wedge(V_1^\ast)\;.
\end{displaymath}
The latter algebra is $\mathbb{Z}$-graded by $\mathrm{S}(V^\ast) =
\oplus_{k \ge 0}\, \mathrm{S}^k (V^\ast)$ where
\begin{displaymath}
    \mathrm{S}^k(V^\ast) \simeq \bigoplus\nolimits_{l=0}^k \left(
    \mathrm{S}^l (V_0^\ast)\otimes \wedge^{k-l}(V_1^\ast)\right)\;.
\end{displaymath}
The action of $G$ on $V$ preserves the $\mathbb{Z}_2$-grading $V =
V_0 \oplus V_1\,$. Thus $G$ acts on $T(V^\ast)$ while leaving the
two-sided ideal $I_-(V^\ast)$ invariant, and it therefore makes sense
to speak of the subalgebra $\mathrm{S}(V^\ast)^G$ of $G$-fixed
elements in $\mathrm{S}(V^\ast)$.

It is a result of R.\ Howe -- see Theorem 2 of \cite{Howe1989} --
that for each of the cases $G = \mathrm{GL}_n\,$, $\mathrm{O}_n\,$,
and $\mathrm {Sp}_n\,$, the graded-commutative algebra $\mathrm{S}
(V^\ast)^G$ is generated by $\mathrm{S}^2(V^\ast)^G$. Hence, our
attention once again turns to the subspace $\mathrm{S}^2(V^\ast)^G$
of quadratic invariants.

\subsection{Quadratic invariants}

It follows from the definition of the graded-symmetric algebra
$\mathrm{S}(V^\ast)$ that the subspace of quadratic elements
decomposes as $\mathrm{S}^2(V_0^\ast \oplus V_1^\ast) = \mathrm{S}^2
(V_0^\ast) \oplus \wedge^2(V_1^\ast) \oplus (V_0^\ast \otimes
V_1^\ast)$. So, since $G$ acts on $V_0^\ast$ and $V_1^\ast$ we have a
decomposition
\begin{displaymath}
    \mathrm{S}^2(V_0^\ast \oplus V_1^\ast)^G = \mathrm{S}^2(V_0^\ast)^G
    \oplus \wedge^2(V_1^\ast)^G \oplus (V_0^\ast \otimes V_1^\ast)^G \;.
\end{displaymath}
To describe the components let us recall the notation $U_r =
\mathbb{C}^r \oplus (\mathbb{C}^r)^\ast$ for $r = p , q\,$.
\begin{lem}\label{lem:susy-invs}
$\mathrm{S}^2(V^\ast)^G$ is isomorphic as a $\mathbb{Z}_2$-graded
complex vector space to $W^\ast$ where the even and odd components of
$W = W_0 \oplus W_1$ are
\begin{eqnarray*}
    &&W_0 = W_{00} \times W_{11} = \left\{ \begin{array}{ll}
    \mathrm{End}(\mathbb{C}^p) \times \mathrm{End}(\mathbb{C}^q)
    \;, &\quad G = \mathrm{GL}_n \;, \\ \mathrm{Sym}_s(U_p) \times
    \mathrm{Sym}_a(U_q) \;, &\quad G = \mathrm{O}_n \;, \\
    \mathrm{Sym}_a(U_p) \times \mathrm{Sym}_s(U_q)\;, &\quad G =
    \mathrm{Sp}_n \;, \end{array} \right. \\ &&W_1 = \left\{
    \begin{array}{ll} \mathrm{Hom}(\mathbb{C}^q,\mathbb{C}^p)
    \oplus \mathrm{Hom}(\mathbb{C}^p,\mathbb{C}^q) \;, &\quad G =
    \mathrm{GL}_n \;, \\ \mathrm{Hom}(U_q \, , U_p) \;, &\quad G =
    \mathrm{O}_n \;, \mathrm{Sp}_n \;. \end{array} \right.
\end{eqnarray*}
\end{lem}
\begin{proof}
Writing $\mathrm{S}^2(V_0^\ast)^G \simeq W_{00}$ and $\wedge^2
(V_1^\ast)^G \simeq W_{11}\,$, the statement concerning the even part
$W_0 \subset W$ is just a summary of Lemmas \ref{lem:BB-GL-invs},
\ref{lem:BB-OSP-invs}, and \ref{lem:FF-invs}. Thus what remains to be
done is to prove the isomorphism $(V_0^\ast\otimes V_1^\ast)^G \simeq
W_1^\ast$ between odd components. Let us prove the equivalent
statement $(V_0 \otimes V_1)^G \simeq W_1\,$.

In the case of $G = \mathrm{GL}_n$ there are two types of invariant:
we can compose an element $\tilde{L} \in \mathrm{Hom}(\mathbb{C}^p ,
\mathbb{C}^n)$ with an element $K \in \mathrm{Hom}(\mathbb{C}^n ,
\mathbb{C}^q)$ to form $K \tilde{L} \in \mathrm{Hom}(\mathbb{C}^p ,
\mathbb{C}^q)$, or else compose $\tilde{K} \in \mathrm{Hom}
(\mathbb{C}^q , \mathbb{C}^n)$ with $L \in \mathrm{Hom}(\mathbb{C}^n
, \mathbb{C}^p)$ to form $L\, \tilde{K} \in \mathrm{Hom}(\mathbb{C}^q
, \mathbb{C}^p)$. This already gives the desired statement $(V_0
\otimes V_1)^{ \mathrm{GL}_n} \simeq \mathrm{Hom}(\mathbb{C}^q ,
\mathbb{C}^p) \oplus \mathrm{Hom}(\mathbb{C}^p , \mathbb{C}^q)$.

In the cases of $G = \mathrm{O}_n\,$, $\mathrm{Sp}_n$ we use
$\mathrm{Hom}(\mathbb{C}^n , \mathbb{C}^r) \simeq \mathrm{Hom}
((\mathbb{C}^r)^\ast , (\mathbb{C}^n)^\ast)$ and the $G$-equivariant
isomorphism $\beta : \, \mathbb{C}^n \to (\mathbb{C}^n)^\ast$ to make
the identifications
\begin{displaymath}
    V_0 \simeq \mathrm{Hom}(\mathbb{C}^n , U_p) \;, \quad
    V_1 \simeq \mathrm{Hom}(U_q \, , \mathbb{C}^n) \;.
\end{displaymath}
After this, the $G$-invariants in $V_0 \otimes V_1$ are seen to be in
one-to-one correspondence with composites $L\, \tilde{K} \in
\mathrm{Hom}(U_q \, , U_p)$ where $\tilde{K} \in \mathrm{Hom}(U_q \,
, \mathbb{C}^n)$ and $L \in \mathrm{Hom}(\mathbb{C}^n, U_p)$.
\end{proof}
\begin{rem}
Defining $\mathbb{Z}_2$-graded vector spaces $\mathbb{C}^{p|q} :=
\mathbb{C}^p \oplus \mathbb{C}^q$ and $U_{p|q} := U_p \oplus U_q$ we
could say that $\mathrm{S}^2(V)^{\mathrm{GL}_n} \simeq \mathrm{End}
(\mathbb{C}^{p|q})$, while $\mathrm{S}^2 (V)^ {\mathrm{O}_n} \simeq
\mathrm{S}^2(U_{p|q})$ and $\mathrm{S}^2 (V)^{\mathrm{Sp}_n} \simeq
\wedge^2(U_{p|q})$. We will not use these identifications here.
\end{rem}

\subsection{Pullback from $\mathcal{A}_W$ to $\mathcal{A}_V^G$}
\label{sect:pullback}

With $W = W_0 \oplus W_1$ as specified in Lemma \ref{lem:susy-invs},
consider now the algebra $\mathcal{A}_W$ of holomorphic functions
\begin{displaymath}
    F \, : \,\, W_0 \to \wedge(W_1^\ast) \;.
\end{displaymath}
At the linear level we have the isomorphism of Lemma
\ref{lem:susy-invs}, which we here denote by
\begin{displaymath}
    Q_2^\ast \, : \,\, W^\ast \to \mathrm{S}^2(V^\ast)^G \;.
\end{displaymath}
This extends in the natural way to an isomorphism of tensor algebras
\begin{displaymath}
    Q_T^\ast \,:\,\, T(W^\ast)\to T\big(\mathrm{S}^2(V^\ast)^G \big)\;.
\end{displaymath}
Since $Q_2^\ast$ is an isomorphism of $\mathbb{Z}_2$-graded vector
spaces, $Q_T^\ast$ sends the ideal $I_-(W^\ast) \subset T(W^\ast)$
generated by the graded-skew elements $w \otimes w^\prime - (-1)^{|w|
|w^\prime|} w^\prime \otimes w$ into the ideal $I_-(V^\ast) \subset
T(V^\ast)$ generated by the same type of element $v \otimes v ^\prime
- (-1)^{|v| |v^\prime|} v^\prime \otimes v$.

Now, taking the quotient of $T(V^\ast)$ by $I_-(V^\ast)$ is
compatible with the reductive action of $G$,
%
%
and it therefore follows that $Q_T^\ast$ descends to a mapping
\begin{displaymath}
    Q^\ast \, : \,\, \mathrm{S}(W^\ast) \to \mathrm{S}(V^\ast)^G \;.
\end{displaymath}
Because $Q^\ast(W^\ast) = Q_2^\ast(W^\ast) = \mathrm{S}^2(V^\ast)^G$
and $\mathrm{S}(V^\ast)^G$ is generated by $\mathrm{S}^2(V^\ast)^G$,
the map $Q^\ast \, : \,\, \mathrm{S}(W^\ast) \to \mathrm{S}
(V^\ast)^G$ is surjective.

The same holds true \cite{GSchwarz} at the level of our holomorphic
functions $\mathcal{A}_W$ and $\mathcal{A}_V^G$:
\begin{prop}
The homomorphism of algebras $Q^\ast : \, \mathcal{A}_W \to
\mathcal{A}_V^G$ is surjective.
\end{prop}

\subsection{Berezin superintegral form}

For a $\mathbb{Z}_2$-graded complex vector space such as our space
$V= V_0 \oplus V_1$ with dimensions $\mathrm{dim}\, V_0 = 2pn$ and
$\mathrm{dim}\, V_1 = 2qn\,$, we denote by $\mathrm{Ber}(V)$ the
complex one-dimensional space
\begin{displaymath}
    \mathrm{Ber}(V) = \wedge^{2p n} (V_0^\ast)
    \otimes \wedge^{2qn} (V_1) \;.
\end{displaymath}
Let now each of the Hermitian vector spaces $V_0$ and $V_1$ be
endowed with an orientation. Then there is a canonical top-form
$\tilde{\Omega}_{V_0} \in \wedge^{2pn} (V_0^\ast)$ and a canonical
generator $\Omega_{V_1} \in \wedge^{2qn} (V_1)$. Their tensor product
$\Omega_V := \tilde{\Omega}_{V_0} \otimes \Omega_{V_1} \in \mathrm
{Ber} (V)$ is called the (flat) Berezin superintegral form of $V$.
Such a form $\Omega_V$ determines a linear mapping
\begin{displaymath}
    \Omega_V \, : \,\, \mathcal{A}_V \to \Gamma ( V_0 \, ,
    \wedge^{2pn} (V_0^\ast)) \;, \quad f \mapsto \Omega_V[f] \;,
\end{displaymath}
from the algebra of holomorphic functions $f : \, V_0 \to \wedge
(V_1^\ast)$ to the space of top-degree holomorphic differential forms
on $V_0\,$. Indeed, if $v$ is any element of $V_0\,$, then by pairing
$f(v) \in \wedge(V_1^\ast)$ with the second factor $\Omega_{V_1}$ of
$\Omega_V$ we get a complex number, and subsequent multiplication by
the first factor $\tilde{\Omega}_{V_0}$ results in an element of
$\wedge^{2pn} (V_0^\ast)$.

In keeping with the approach taken in Sect.\ \ref{sect:BB-sector}, we
want to integrate over the real vector space $V_{0,\mathbb{R}}$
defined as the graph of $\dagger : \, \mathrm{Hom}(\mathbb{C}^n,
\mathbb{C}^p)\to\mathrm{Hom} (\mathbb{C}^p , \mathbb{C}^n)$. For
this, let $\mathrm{dvol}_ {V_{0,\mathbb{R}}}$ denote the positive
density $\mathrm{dvol}_{V_{0,\mathbb{R}}} := \vert \tilde{\Omega}
_{V_0} \vert$ restricted to $V_{0, \mathbb{R}}\,$. (This change from
top-degree forms to densities is made in anticipation of the fact
that we will transfer the integral to a symmetric space which in
certain cases is non-orientable; see the Appendix for more discussion
of this issue.)

The Berezin superintegral of $f \in \mathcal{A}_V$ over the
integration domain $V_{0,\mathbb{R}}$ is now defined as the two-step
process of first converting the integrand $f \in \mathcal{A}_V$ into
a holomorphic function $\Omega_{V_1}[f] : \, V_0 \to \mathbb{C}$ and
then integrating this function against $\mathrm{dvol}_{V_{0,
\mathbb{R}}}$ over the real subspace $V_{0,\mathbb{R}}:$
\begin{displaymath}
    f \mapsto \int_{V_{0,\mathbb{R}}} \Omega_{V_1}[f] \,
    \mathrm{dvol}_{V_{0,\mathbb{R}}} \;.
\end{displaymath}
Our interest in the following will be in this kind of integral for
the particular case of $G$-equivariant holomorphic functions $f : \,
V_0 \to \wedge(V_1^\ast)$ (i.e., for $f \in \mathcal{A}_V^G)$.

\subsection{Exploiting equivariance}

Recall from Sect.\ \ref{sect:BB-sector} the definition of the groups
$K_n\,$, $K_p\,$, $K_{n,p}\,$, and $G_p\,$. Recall also that $X_{p,
\,n} = \psi(\mathrm{Hom}(\mathbb{C}^n , \mathbb{C}^p))$ denotes the
vector space of structure-preserving linear transformations $\mathbb
{C}^n \to U_p\,$. To simplify the notation, let the isomorphism
$\psi$ now be understood, i.e., write $\psi(L) \equiv L\,$.

The subset of regular elements in $X_{p,\,n}$ is denoted by
$X_{p,\,n}^\prime\,$. Taking $\Pi \in X_{p,\,n}^\prime$ to be the
orthogonal projector $\mathbb{C}^n = U_p \oplus U_{n,p} \to U_p$ we
have the isomorphism
\begin{displaymath}
    G_p \times_{(K_p \times K_{n,p})} K_n \stackrel{\sim}{\to}
    X_{p,\,n}^\prime \;, \quad (g,k) \mapsto g\Pi \, k \;.
\end{displaymath}
Note that $X_{p,\,n}^\prime$ is a left $G_p$-space and a right
$K_n$-space. Note also the relations $\Pi \, k = k \Pi$ for $k \in
K_p$ and $\Pi \, k = 0$ for $k \in K_{n,p}\,$.

Since the compact subgroup $K_n \subset G$ acts on $V_{0, \mathbb
{R}}\,$, the given integrand $f \in \mathcal{A}_V^G$ restricts to a
function $f : \, X_{p,\,n}^\prime \to \wedge(V_1^\ast)$ which has the
property of being $K_n$-equivariant:
\begin{displaymath}
    f(L) = f(g\Pi\, k) = k^{-1} . f(g\Pi) \quad (k \in K_n)\;.
\end{displaymath}
Now notice that since the action of $G$ on $\wedge^{2qn}(V_1)$ is
trivial, the Berezin form $\Omega_{V_1}$ is invariant under $G$ and
hence invariant under the subgroup $K_n:$
\begin{displaymath}
    k (\Omega_{V_1}) = \Omega_{V_1} \circ k^\ast = \Omega_{V_1}
    \quad (k \in K_n) \;.
\end{displaymath}
Consequently, applying $\Omega_{V_1}$ to the $K_n$-equivariant
function $f$ we obtain
\begin{displaymath}
    \Omega_{V_1}[ f(g\Pi \, k)] = \Omega_{V_1}[ k^{-1} . f(g\Pi)]
    = \Omega_{V_1} [f(g\Pi)] \;,
\end{displaymath}
and this gives the following formula for the integral of $f$,
\begin{equation}\label{eq:5.1mrz}
    \int \Omega_{V_1}[f(L)]\, \mathrm{dvol}_{V_{0,\mathbb{R}}}(L)
    = \frac{\mathrm{vol}(K_n)}{\mathrm{vol}(K_{n,p})} \int\limits_{
    G_p / K_p} \Omega_{V_1} [f(g\Pi)]\, J(g) dg_{K_p} \;,
\end{equation}
as an immediate consequence of Prop.\ \ref{prop:BB-sector}.

Based on this formula, our next step is to process the integrand
$\Omega_{V_1}[f(g\Pi)]$.

\subsection{Transforming the Berezin integral}

It will now be convenient to regard the odd vector space $V_1 =
\mathrm{Hom}(\mathbb{C}^n , \mathbb{C}^q) \oplus \mathrm{Hom}
(\mathbb{C}^q , \mathbb{C}^n)$ for the case of $G = \mathrm{GL}_n$ as
\begin{displaymath}
    V_1 \simeq \mathrm{Hom}(\mathbb{C}^q , \mathbb{C}^n) \oplus
    \mathrm{Hom}((\mathbb{C}^q)^\ast , (\mathbb{C}^n)^\ast)
    \qquad (G = \mathrm{GL}) \;.
\end{displaymath}
In the other cases, using the isomorphism $\beta : \, \mathbb{C}^n
\to (\mathbb{C}^n)^\ast$ we make the identification
\begin{displaymath}
    V_1 \simeq \mathrm{Hom}(\mathbb{C}^q \oplus (\mathbb{C}^q)^\ast ,
    \mathbb{C}^n) = \mathrm{Hom}(U_q\, , \mathbb{C}^n) \;,
    \qquad (G = \mathrm{O} , \mathrm{Sp}) \;.
\end{displaymath}

Following Sect.\ \ref{sect:3.4} we fix an orthogonal decomposition
$\mathbb{C}^n = U_p \oplus U_{n,p}$ for $G = \mathrm{O}_n\,$,
$\mathrm{Sp}_n$ and $\mathbb{C}^n = \mathbb{C}^p \oplus
\mathbb{C}^{n-p}$ for $G = \mathrm{GL}_n\,$, which is Euclidean,
Hermitian symplectic, and Hermitian, respectively, and let this
induce a vector space decomposition $V_1 = V_\parallel \oplus
V_\perp$ in the natural way. For $G = \mathrm{O}_n\,, \mathrm{Sp}_n$
the summands are
\begin{displaymath}
    V_\parallel = \mathrm{Hom}(U_q \, , U_p) \;, \quad V_\perp =
    \mathrm{Hom}(U_q\,, U_{n,p}) \qquad (G =\mathrm{O},\mathrm{Sp})\;,
\end{displaymath}
and in the case of $G = \mathrm{GL}_n$ we have
\begin{eqnarray*}
    &&V_\parallel = \mathrm{Hom}(\mathbb{C}^q , \mathbb{C}^p)
    \oplus \mathrm{Hom}((\mathbb{C}^q)^\ast , (\mathbb{C}^p)^\ast)\;,
    \\ &&V_\perp = \mathrm{Hom}(\mathbb{C}^q,\mathbb{C}^{n-p})\oplus
    \mathrm{Hom}((\mathbb{C}^q)^\ast , (\mathbb{C}^{n-p})^\ast)
    \qquad (G = \mathrm{GL}) \;.
\end{eqnarray*}
From the statement of Lemma \ref{lem:susy-invs} we see that
$V_\parallel$ is isomorphic to $W_1$ as a complex vector space
(though not as a $K_p$-space) in all three cases.

The decomposition $V_1 = V_\parallel \oplus V_\perp$ induces a
factorization
\begin{equation}\label{eq:splitting}
    \wedge(V_1^\ast) \simeq \wedge(V_\parallel^\ast) \otimes
    \wedge(V_\perp^\ast)
\end{equation}
of the exterior algebra of $V_1^\ast$. In all three cases
($\mathrm{GL}$, $\mathrm{O}$, $\mathrm{Sp}$) the decomposition of
$V_1$ and that of $\wedge(V_1^\ast)$ is stabilized by the group $K_p
\times K_{n,p}\,$. We further note that $K_p \hookrightarrow K_n$
acts trivially on $\wedge(V_\perp^\ast)$ while $K_{n,p}
\hookrightarrow K_n$ acts trivially on $\wedge (V_\parallel^\ast)$.

To compute $\Omega_{V_1}[ f(g\Pi)]$, we are going to disect the
Berezin form $\Omega_{V_1}$ according to the decomposition
(\ref{eq:splitting}). For this, recall that if $V = A \oplus A^\ast$
is the direct sum of an $N$-dimensional vector space $A$ and its dual
$A^\ast$, then there exists a canonical generator $\Omega_V =
\Omega_{A \oplus A^\ast} \in \wedge^{2N}(A \oplus A^\ast)$ which is
given by $\Omega_{A \oplus A^\ast} = f_N \wedge e_N \wedge \ldots
\wedge f_1 \wedge e_1$ for any basis $\{ e_j \}$ of $A$ with dual
basis $\{ f_j \}$ of $A^\ast$. Note that $\Omega_{A \oplus A^\ast} =
(-1)^N \Omega_{A^\ast \oplus A}\,$.

The following statement is an immediate consequence of the properties
of $\wedge\,$.
\begin{lem}\label{lem:ABC}
If $A, B, C$ are vector spaces and $A = B \oplus C$ then
\begin{displaymath}
    \Omega_{A \oplus A^\ast} = \Omega_{B \oplus B^\ast} \wedge
    \Omega_{C \oplus C^\ast} \;.
\end{displaymath}
\end{lem}
Now $V_1$ and all of our spaces $V_\perp$ and $V_\parallel$ are the
direct sum of a vector space and its dual. Recall from Sect.\
\ref{sect:4.3} that $\Omega_{V_1} = \Omega_{\mathrm{Hom}(
(\mathbb{C}^q )^\ast , (\mathbb{C}^n)^\ast) \oplus \mathrm{Hom}
(\mathbb{C}^q , \mathbb{C}^n)}$, and let
\begin{displaymath}
    \Omega_{V_\parallel} = \Omega_{\mathrm{Hom}((\mathbb{C}^q)^\ast,
    (\mathbb{C}^p)^\ast)\oplus \mathrm{Hom}(\mathbb{C}^q,\mathbb{C}^p)}
    \;, \quad \Omega_{V_\perp} = \Omega_{
    \mathrm{Hom}((\mathbb{C}^q)^\ast , (\mathbb{C}^{n-p})^\ast)
    \oplus \mathrm{Hom}(\mathbb{C}^q,\mathbb{C}^{n-p})}
\end{displaymath}
for $G = \mathrm{GL}$, while in the case of $G = \mathrm{O},
\mathrm{Sp}$ the corresponding definitions are
\begin{displaymath}
    \Omega_{V_\parallel} = \Omega_{\mathrm{Hom}((\mathbb{C}^q)^\ast,
    \, U_p)\oplus \mathrm{Hom}(\mathbb{C}^q,\, U_p)} \;, \quad
    \Omega_{V_\perp} = \Omega_{\mathrm{Hom}((\mathbb{C}^q)^\ast ,\,
    U_{n,p}) \oplus \mathrm{Hom}(\mathbb{C}^q,\, U_{n,p})} \;.
\end{displaymath}
Here the vector spaces $\mathrm{Hom}( (\mathbb{C}^q)^\ast,\, U_p)$
and $\mathrm{Hom}(\mathbb{C}^q,\, U_p)$ are regarded as dual to each
other by the symmetric bilinear form $s : \, U_p \times U_p \to
\mathbb{C}$ for $G = \mathrm{O}$ and the alternating bilinear form $a
: \, U_p \times U_p \to \mathbb{C}$ for $G = \mathrm{Sp}\,$. This
means that for $G = \mathrm{O}$ we have
\begin{displaymath}
    \Omega_{V_\parallel} = \Omega_{\mathrm{Hom}((\mathbb{C}^q)^\ast,
    (\mathbb{C}^p)^\ast)\oplus\mathrm{Hom}(\mathbb{C}^q,\mathbb{C}^p)}
    \wedge \Omega_{\mathrm{Hom}((\mathbb{C}^q)^\ast,\mathbb{C}^p)
    \oplus \mathrm{Hom}(\mathbb{C}^q,(\mathbb{C}^p)^\ast)} \;,
\end{displaymath}
while the same Berezin form $\Omega_{V_\parallel}$ for $G =
\mathrm{Sp}$ has an extra sign factor $(-1)^{pq}$ due to the
alternating property of $a$ (cf.\ the sentence after the definition
of $\Omega_{W_1}$ in Eq.\ (\ref{eq:1.12})). The same conventions hold
good in the case of $\Omega_{V_\perp}$ -- we need only observe that
the given symmetric or alternating bilinear form on $U_{n,p}$ induces
such a form on $V_\perp\,$.

Now, applying Lemma \ref{lem:ABC} to the present situation we always
have
\begin{displaymath}
    \Omega_{V_1} = \Omega_{V_\parallel} \wedge \Omega_{V_\perp} \;.
\end{displaymath}

\subsubsection{Transformation of $\Omega_{V_\parallel}$}

Recall the isomorphism of vector spaces $V_\parallel \simeq W_1\,$,
which we now realize as follows. Using the identifications
$V_\parallel \simeq \mathrm{Hom}(U_q\, , U_p)$ for the case of $G =
\mathrm{O}$, $\mathrm{Sp}$ and $V_\parallel \simeq \mathrm{Hom}
(\mathbb{C}^q , \mathbb{C}^p) \oplus \mathrm{Hom}( (\mathbb{C}^q
)^\ast , (\mathbb{C}^p)^\ast)$ for $G = \mathrm{GL}\,$, we apply $g
\in G_p$ to $v \in V_\parallel$ to form $gv\,$, where $gv$ for $G =
\mathrm{GL}$ means $gv = g.(\tilde{L} \oplus L^\mathrm{t}) = (g
\tilde{L}) \oplus (\overline{g} L^\mathrm{t})$. Note that the mapping
$(g,v) \mapsto gv$ has the property of being $K_p$-invariant.

Given this isomorphism $g : \, V_\parallel \to W_1\,$, let $(g^{-1}
)^\ast : \, \wedge(V_\parallel^\ast) \to \wedge(W_1^\ast)$, $f
\mapsto g.f$, be the induced isomorphism preserving the pairing
between vectors and forms. Then
\begin{displaymath}
    \Omega_{V_1}[ f(g\Pi)] = \left( g(\Omega_{V_\parallel}) \wedge
    \Omega_{V_\perp}\right) [g . f(g\Pi)] \;,
\end{displaymath}
so our next step is to compute $g(\Omega_{V_\parallel})$. Here it
should be stressed that we define the Berezin form $\Omega_{W_1}$ by
the same ordering conventions we used to define
$\Omega_{V_\parallel}$ above.
\begin{lem}\label{lem:Om-para}
Under the isomorphism $V_\parallel \to W_1$ by $v \mapsto gv$ the
Berezin forms $\Omega_{W_1}$ and $\Omega_{V_\parallel}$ are related
by $g(\Omega_{V_\parallel}) = \mathrm{Det}^q(g g^\dagger)
\Omega_{W_1}\,$.
\end{lem}
\begin{proof}
Consider first the case of $G = \mathrm{O}$, $\mathrm{Sp}$, where
$V_\parallel = \mathrm{Hom}(U_q\,, U_p)$ and the same choice of
polarization $\mathrm{Hom}(U_q\,, U_p) = \mathrm{Hom}(
(\mathbb{C}^q)^\ast , U_p) \oplus \mathrm{Hom}(\mathbb{C}^q , U_p)$
determines both $\Omega_{V_\parallel}$ and $\Omega_{W_1}$. Applying
$g \in G_p$ to $\Omega_{V_\parallel} \in \wedge^{4pq}(V_\parallel)$
we obtain
\begin{displaymath}
    g (\Omega_{V_\parallel}) = \mathrm{Det}^{2q}(g) \Omega_{W_1} \;.
\end{displaymath}
The groups at hand are $G_p = \mathrm{GL}_{2p}(\mathbb{R})$, $\mathrm
{GL}_p(\mathbb{H})$, and $\mathrm{Det}(g)$ is real for these. Hence
\begin{displaymath}
    \mathrm{Det}^{2q}(g) = \big( \mathrm{Det}(g)
    \overline{\mathrm{Det}(g)}\, \big)^q = \mathrm{Det}^q
    (g g^\dagger) \;.
\end{displaymath}

In the case of $G = \mathrm{GL}$ we have $\Omega_{V_\parallel} =
\Omega_{\mathrm{Hom}((\mathbb{C}^q)^\ast , (\mathbb{C}^p)^\ast)
\oplus \mathrm{Hom}(\mathbb{C}^q , \mathbb{C}^p) }$. Transforming the
second summand by $\tilde{L} \mapsto g \tilde{L}$ for $g \in
\mathrm{GL}_p(\mathbb{C})$ we get the Jacobian $\mathrm{Det}^q(g)$,
transforming the first summand by $L^\mathrm{t} \mapsto \overline{g}
L^\mathrm{t}$ we get $\mathrm{Det}^q(\overline{g})$. Thus, altogether
we obtain again $g(\Omega_{V_\parallel}) = \mathrm{Det}^q(g)
\mathrm{Det}^q(\overline{g}) \Omega_{W_1} = \mathrm{Det}^q(g
g^\dagger) \Omega_{W_1}\,$.
\end{proof}

\subsubsection{Bosonization of $\Omega_{V_\perp}$}

We turn to the Berezin form $\Omega_{V_\perp}$ for the factor
$\wedge(V_\perp^\ast)$ of the decomposition (\ref{eq:splitting}).
Recall that the elements of this exterior algebra $\wedge(
V_\perp^\ast)$ are fixed under the action of $K_p\,$. Since $\Pi \, k
= 0$ for $k \in K_{n,p}\,$, the $K_n$-equivariance of $f \in
\mathcal{A}_V^G$ implies that $g . f(g\Pi)$ (for any fixed $g \in
G_p$) lies in $\wedge(W_1^\ast) \otimes \wedge(V_\perp^\ast)^{K_{n,
p}}$.

For future reference, we are now going to record a (bosonization)
formula for the Berezin integral $\Omega_{V_\perp}: \, \wedge
(V_\perp^\ast)^{K_{n,p}} \to \mathbb{C}\,$. For this notice that,
since the action of $K_{n,p}$ is complex linear we have $\wedge(
V_\perp^\ast)^{K_{n,p}} = \wedge(V_\perp^\ast)^{G_{n,p}}$, where
$G_{n,p}$ is the complexification of $K_{n,p}\,$. From Table
\ref{fig:1} of Sect.\ \ref{sect:3.4} we read off that $G_{n,p} =
\mathrm{GL}_{n-p} (\mathbb{C})$, $\mathrm{O}_{n-2p} (\mathbb{C})$,
and $\mathrm {Sp}_{n-2p}(\mathbb{C})$ for our three cases of
$\mathrm{GL}$, $\mathrm{O}$, and $\mathrm{Sp}$, respectively.

The subalgebra $\wedge(V_\perp^\ast )^{G_{n,p}}$ is generated, once
again, by $\wedge^2 (V_\perp^\ast)^{ G_{n,p}}$, the quadratic
invariants. Applying Lemma \ref{lem:FF-invs} with $V^\ast \equiv
V_\perp^\ast\,$, $G \equiv G_{n,p}\,$, and $W \equiv W_{11}$ we get
\begin{displaymath}
    \wedge^2(V_\perp^\ast)^{K_{n,p}} =
    \wedge^2(V_\perp^\ast)^{G_{n,p}} \simeq W_{11}^\ast \;.
\end{displaymath}

Now by the principles expounded in Sect.\ \ref{sect:FF-sector} we
lift a given element $f_\perp \in \wedge(V_\perp^\ast)^{G_{n,p}}$ to
a holomorphic function $F : \, W_{11} \to \mathbb{C}\,$. To formalize
this step, let
\begin{displaymath}
    P_\perp^\ast \,:\,\, \mathcal{O}(W_{11}) \to
    \wedge(V_\perp^\ast)^{K_{n,p}}
\end{displaymath}
be the surjective mapping which was introduced in Sect.\
\ref{sect:FF-invs} and denoted by the generic symbol $Q^\ast$ there.
For $F \in \mathcal{O}(W_{11})$ we then have with $n^\prime = n /
(1+|m|)$ the result
\begin{equation}\label{eq:Om-perp}
    \Omega_{V_\perp} [P_\perp^\ast F] = (2\pi)^{qn(1 - p/
    n^\prime)} 2^{qm} \frac{\mathrm{vol}(K_{n,p})}{\mathrm{vol}
    (K_{n,p-q})} \int_{D_q^1} \frac{F(y)\,d\mu_{D_q^1}(y)}
    {\mathrm{Det}^{n^\prime - p}(y)}\;,
\end{equation}
as an immediate consequence of the formula of Thm.\ \ref{thm:4.10}.
Here we refined our notation by writing $D_q^1$ for the compact
symmetric spaces $D_q$ of Sect.\ \ref{sect:4.6}. The non-compact
symmetric spaces $D_p$ introduced in Sect.\ \ref{sect:3.4} will
henceforth be denoted by $D_p^0\,$.

\subsection{Decomposition of pullback}

Recall from Sect. \ref{sect:pullback} that we have a pullback of
graded-commutative algebras $Q^\ast : \, \mathcal{A}_W \to
\mathcal{A}_V^G\,$. To go further, we should decompose $Q^\ast$
according to the manipulations carried out in the previous two
subsections. This, however, will only be possible in a restricted
sense, as some of our transformations require that the even part of
$w \in W$ be invertible.

We start with a summary of the sequence of operations we have carried
out so far. Recall that the elements of $\mathcal{A}_W$ are
holomorphic functions $F :\, W_0 \to \wedge(W_1^\ast)$, where $W_1$
and $W_0 = W_{00} \times W_{11}$ were described in Lemma
\ref{lem:susy-invs}. Since our domain of integration will be $D_p
\equiv D_p^0 \simeq G_p / K_p\,$, given $F\in \mathcal{A}_W$ let
$F_1$ denote $F$ restricted to $D_p^0 \subset W_{00}:$
\begin{equation}
    F_1 :\, D_p^0 \times W_{11} \to \wedge(W_1^\ast) \;.
\end{equation}
Now we use the Cartan embedding $G_p / K_p \to D_p^0 \subset G_p$ by
$g \mapsto g \theta(g^{-1}) = g g^\dagger$ to pull back $F_1$ in its
first argument from $D_p^0$ to $G_p / K_p\,$. Applying also the
mapping $P_\perp^\ast : \, \mathcal{O}(W_{11}) \to \wedge
(V_\perp^\ast )^{K_{n,p}}$ we go to the second function
\begin{equation}
    F_2 :\, G_p / K_p \to \wedge(W_1^\ast) \otimes
    \wedge(V_\perp^\ast)^{K_{n,p}} \;.
\end{equation}
In the next step, employing the isomorphism $V_\parallel \to W_1$, $v
\mapsto gv$ (pointwise for each coset $g K_p \in G_p / K_p$) we pull
back $F_2$ to a $K_p$-equivariant function
\begin{equation}
    F_3 \,:\,\, G_p \to \wedge(V_\parallel^\ast) \otimes
    \wedge(V_\perp^\ast)^{K_{n,p}} \;.
\end{equation}
Be advised that we are now at the level of the integrand $F_3(g) =
f(g\Pi)$ of (\ref{eq:5.1mrz}). In the final step, we pass to the
unique extension of $F_3$ to a $K_n$-equivariant function
\begin{equation}
    F_4 \,:\,\, X_{p,\,n}^\prime \stackrel{K_n-\textrm{eqvt}}
    {\longrightarrow} \wedge(V_1^\ast)
\end{equation}
by $F_4(L) = F_4(g\Pi k) := F_3(g)$. Let us give a name to this
sequence of steps.
\begin{defn}
We denote by $P^\ast$ the homomorphism of graded-commutative algebras
taking $F_1 :\, D_p^0 \times W_{11} \to \wedge(W_1^\ast)$ to the
$K_n$-equivariant function $F_4 : \, X_{p,\,n}^\prime \to \wedge
(V_1^\ast)$.
\end{defn}
The main point of this subsection will be to show that $Q^\ast$
(restricted to $D_p^0 \times W_{11}$) is the composition of $P^\ast$
with another homomorphism, $S^\ast$, which we describe next.

Consider first the case of $G = \mathrm{GL}$ where $W = W_0 \oplus
W_1$ and
\begin{eqnarray*}
    &&W_0 = W_{00} \oplus W_{11} = \mathrm{End}(\mathbb{C}^p)
    \oplus \mathrm{End}(\mathbb{C}^q) \;, \\ &&W_1 = W_{01}
    \oplus W_{10} = \mathrm{Hom}(\mathbb{C}^q,\mathbb{C}^p)
    \oplus \mathrm{Hom}(\mathbb{C}^p , \mathbb{C}^q) \;,
\end{eqnarray*}
and let $W_{00}^\prime$ denote the subset of regular elements in
$W_{00}\,$.
%
%
On $W^\prime := (W_{00}^\prime \times W_{11}) \times (W_{01} \oplus
W_{10})$ define a non-linear mapping $S :\, W^\prime \to W^\prime$ by
\begin{displaymath}
    S(x\, , y\, ;\sigma,\tau) =
    (x\, , y + \tau x^{-1} \sigma; \sigma,\tau)\;.
\end{displaymath}
This mapping is compatible with the structure of the
graded-commutative algebra $\mathcal{A}_W$ which is induced from the
$\mathbb{Z}_2$-grading $W = W_0 \oplus W_1$. Therefore, viewing the
entries of $\sigma$ and $\tau$ as anti-commuting generators, $S$
determines an automorphism $S^\ast : \, \mathcal{A}_W^\prime \to
\mathcal{A}_W^\prime$ of the superalgebra $\mathcal{A}_W^\prime$ of
holomorphic functions from $W_{00}^\prime \times W_{11}$ to $\wedge
(W_1^\ast)$. Adopting the supermatrix notation commonly used in
physics one would write
\begin{displaymath}
    (S^\ast F )\begin{pmatrix} x &\sigma\\ \tau &y \end{pmatrix} =
    F \begin{pmatrix} x &\sigma\\ \tau &y + \tau\, x^{-1} \sigma
    \end{pmatrix} \;.
\end{displaymath}

Next consider the case of $G = \mathrm{O}_n\,$, where $W_{00} =
\mathrm{Sym}_s(U_p)$, $W_{11} = \mathrm{Sym}_a(U_q)$, and $W_1 =
\mathrm{Hom}(U_q \, , U_p)$. Here we have the bilinear forms $s$ on
$U_p$ and $a$ on $U_q\,$, and these determine an isomorphism
$\mathrm{Hom}(U_q\, , U_p) \to \mathrm{Hom}(U_p\, , U_q)$, $\sigma
\mapsto \sigma^T$ by
\begin{displaymath}
    s(\sigma v\, , u) = a(v\, ,\sigma^T u) \quad
    (u \in U_p\, ,\,\, v\in U_q) \;.
\end{displaymath}
However, since $\sigma$ and $v$ in this definition are to be
considered as odd and $\sigma$ moves past $v\,$, the good isomorphism
to use (the 'supertranspose') has an extra minus sign:
\begin{displaymath}
    \sigma \mapsto \sigma^{\mathrm{s}T} := - \sigma^T \;.
\end{displaymath}
Restricting again to the regular elements $W_{00}^\prime$ of
$W_{00}\,$, define a mapping $S : \, W^\prime \to W^\prime$ on
$W^\prime = (W_{00}^\prime \times W_{11}) \times W_1$ by
\begin{displaymath}
    S(x\, , y\, ;\sigma) =
    (x\, , y + \sigma^{\mathrm{s}T} x^{-1} \sigma; \sigma)\;.
\end{displaymath}
From $x^{-1} \in \mathrm{Sym}_s(U_p)$ and the definition of the
transposition operation $\sigma \mapsto \sigma^{\mathrm{s}T}$ via the
bilinear forms $s$ and $a\,$, it is immediate that $\sigma^{
\mathrm{s}T} x^{-1} \sigma \in W_{11}\,$. Now for the same reasons as
before, $S$ determines an automorphism $S^\ast : \,
\mathcal{A}_W^\prime \to \mathcal{A}_W^\prime\,$.

The definitions for the last case $G = \mathrm{Sp}_n$ are the same as
for $G = \mathrm{O}_n$ but for the fact that the two bilinear forms
$s$ and $a$ exchange roles.

From here on we consider $S^\ast$ to be restricted to the functions
with domain $D_p^0 \times W_{11}\,$.
\begin{lem}
The homomorphism of superalgebras $Q^\ast : \, \mathcal{A}_W \to
\mathcal{A}_V^G$, when restricted to a homomorphism $Q^\ast$ taking
functions $D_p^0 \times W_{11} \to \wedge(W_1^\ast)$ to
$K_n$-equivariant functions $X_{p,\,n}^\prime \to \wedge(V_1^\ast)$,
decomposes as
\begin{displaymath}
    Q^\ast = P^\ast S^\ast \;.
\end{displaymath}
\end{lem}
\begin{proof}
Since the isomorphism $Q^\ast :\, W^\ast \stackrel{\sim}{\to}
\mathrm{S}^2 (V^\ast)^G$ determines $Q^\ast : \, \mathcal{A}_W \to
\mathcal{A}_V^G$, it suffices to check $Q^\ast = P^\ast S^\ast$ at
the level of the quadratic map $Q :\, V \to W$.

Let us write out the proof for the case of $G = \mathrm{GL}$ (the
other cases are no different). Recall that the quadratic map $Q : \,
V \to W$ in this case is given by
\begin{displaymath}
    Q \, :\,\, (L \oplus \tilde{L}) \oplus (K
    \oplus \tilde{K}) \mapsto \begin{pmatrix} L\, \tilde{L} & L\,
    \tilde{K} \\ K \tilde{L} &K \tilde{K} \end{pmatrix} \;.
\end{displaymath}
Now, fixing a regular element $(L\,,L^\dagger) \in X_{p,\,
n}^\prime\,$, we have an orthogonal decomposition
\begin{displaymath}
    \mathbb{C}^n = \mathrm{ker}(L) \oplus \mathrm{im}(L^\dagger)
\end{displaymath}
where $\mathrm{im}(L^\dagger) \simeq \mathbb{C}^p$ and
$\mathrm{ker}(L) \simeq \mathbb{C}^{n-p}$. Let $\Pi_L := L^\dagger (L
\, L^\dagger)^{-1} L$ denote the orthogonal projection $\Pi_L :\,
\mathbb{C}^n \to \mathrm{im}(L^\dagger)$. If we decompose $K$,
$\tilde{K}$ as
\begin{displaymath}
    K = K_\parallel(L) + K_\perp(L) \;, \quad K_\parallel(L) =
    K \, \Pi_L \;, \quad \tilde{K} = \tilde{K}_\parallel(L) +
    \tilde{K}_\perp(L)\;, \quad \tilde{K}_\parallel(L) = \Pi_L
    \, \tilde{K} \;,
\end{displaymath}
then our homomorphism $P^\ast$ is the pullback of algebras determined
by the map
\begin{eqnarray*}
    &&P \,:\,\, X_{p,\,n}^\prime \times V_1 \to W^\prime \;, \\
    &&((L\, , L^\dagger) , (K \oplus \tilde{K})) \mapsto
    \begin{pmatrix} L\, L^\dagger &L\, \tilde{K}_\parallel(L) \\
    K_\parallel(L) L^\dagger &K_\perp(L) \tilde{K}_\perp(L)
    \end{pmatrix} = \begin{pmatrix} L\, L^\dagger &L\,\tilde{K} \\
    K L^\dagger &K(\mathrm{Id}-\Pi_L) \tilde{K} \end{pmatrix} \;.
\end{eqnarray*}
When the second map $S : \, W^\prime \to W^\prime$ is applied to this
result, all blocks remain the same but for the $W_{11}$-block, which
transforms as
\begin{displaymath}
    K(\mathrm{Id} - \Pi_L) \tilde{K} \mapsto K(\mathrm{Id} - \Pi_L)
    \tilde{K} + (KL^\dagger) (L\,L^\dagger)^{-1} (L\, \tilde{K}) =
    K \tilde{K} \;.
\end{displaymath}
Thus $S \circ P$ agrees with $Q$ on $X_{p,\,n}^\prime \times V_1\,$,
which implies the desired result $Q^\ast = P^\ast S^\ast$.
\end{proof}
We now state an intermediate result en route to the proof of the
superbosonization formula. Let $f \in \mathcal{A}_V^G$ and $F \in
\mathcal{A}_W$ be related by $f = Q^\ast F$. We then do the following
steps: (i) start from formula (\ref{eq:5.1mrz}) for $\int
\Omega_{V_1} [P^\ast S^\ast F]\, \mathrm{dvol}_{V_{0, \mathbb{R}}}$ ;
(ii) transform the Berezin integral $\Omega_{V_1}[(P^\ast S^\ast F)(g
\Pi)]$ by Lemma \ref{lem:Om-para} for the part $\Omega_{V_\parallel}$
and equation (\ref{eq:Om-perp}) for $\Omega_{V_\perp}$; (iii) use
Cor.\ \ref{cor:BB-sector} to push the integral over $G_p / K_p$
forward to $D_p^0$ by the Cartan embedding; (iv) use $\mathrm{Det}^q
(g g^\dagger) = \mathrm{Det}^q(x)$. The outcome of these steps is the
formula
\begin{eqnarray}
    &&\int \Omega_{V_1}[f]\, \mathrm{dvol}_{V_{0,\mathbb{R}}}
    = 2^{(q-p)(n+m)} \pi^{\,qn} (2\pi)^{- pq(1+|m|)}
    \frac{\mathrm{vol}(K_n)}{\mathrm{vol}(K_{n,p-q})}
    \times \nonumber \\ &&\times \int_{D_p^0} \left( \int_{D_q^1}
    \Omega_{W_1} [S^\ast F(x\, ,y)]\,\mathrm{Det}^{p-n^\prime}(y)
    \, d\mu_{D_q^1}(y)\right) \mathrm{Det}^{q+n^\prime}(x)\,
    d\mu_{D_p^0}(x) \label{eq:5.8mrz} \;.
\end{eqnarray}
Let us recall once more that $n^\prime = n$ for $G = \mathrm{GL}$ and
$n^\prime = n/2$ for $G = \mathrm{O}$, $\mathrm{Sp}\,$.

\subsection{Superbosonization formula}
\label{sect:thm}

We are now in a position to reap the fruits of all our labors.
Introducing the notation $(S_x^\ast \, \mathrm{Det})(y) =
\mathrm{Det}(y + \tau x^{-1} \sigma)$ for $G = \mathrm{GL}$ and
$(S_x^\ast \, \mathrm{Det})(y) = \mathrm{Det}(y + \sigma^{\mathrm{s}
T} x^{-1} \sigma)$ for $G = \mathrm{O}$ and $G = \mathrm{Sp}\,$, we
note that the superdeterminant function $\mathrm{SDet}:\, D_p^0
\times D_q^1 \to \wedge(W_1^\ast)$ is given by
\begin{displaymath}
    \mathrm{SDet}(x\, ,y) = \frac{\mathrm{Det}(x)}
    {((S_x^\ast)^{-1} \mathrm{Det})(y)} \;.
\end{displaymath}
We define a related function $J : \, D_p^0 \times D_q^1 \to \wedge
%
%
(W_1^\ast)$ by
\begin{displaymath}
    J(x\, ,y) = \frac{\mathrm{Det}^q(x)\, \mathrm{Det}^{q-m/2}(y)}
    {((S_x^\ast)^{-1} \mathrm{Det}^{q-m/2-p})(y)} \;.
\end{displaymath}
\begin{thm}\label{thm:5.7}
Let $f : \, V_0 \to \wedge(V_1^\ast)^G$ be a $G$-equivariant
holomorphic function which restricts to a Schwartz function along the
real subspace $V_{0,\mathbb{R}}\,$. If $F :\, W_0 \to \wedge(
W_1^\ast)$ is any holomorphic function that pulls back to $Q^\ast F =
f\,$, then
\begin{eqnarray*}
    &&\int \widetilde{\Omega}_{V_1}[f]\, \mathrm{dvol}_{V_{0,\mathbb{R}}}
    = 2^{(q-p)m} \frac{\mathrm{vol}(K_n)}{\mathrm{vol}(K_{n,p-q})}
    \times \\ &&\times \int_{D_p^0} \int_{D_q^1} \widetilde{\Omega}_{W_1}
    [(J\cdot \mathrm{SDet}^{n^\prime}\cdot F)(x\,,y)]\,d\mu_{D_q^1}(y)\,
    d\mu_{D_p^0}(x)\;,
\end{eqnarray*}
where $n^\prime = n/(1+|m|) \ge p\,$, and $\widetilde{\Omega}_{V_1}
:= 2^{pn} (2\pi)^{-qn} \Omega_{V_1}$ and $\widetilde{\Omega}_{W_1} :=
(2\pi)^{ -pq (1+|m|)} \Omega_{W_1}$ are Berezin integral forms with
adjusted normalization.
\end{thm}
\begin{proof}
We first observe that in the present context the formula of Lemma
\ref{lem:shift} can be written as
\begin{displaymath}
    \int_{D_q^1} (S^\ast F)(x\, ,y)\, d\mu_{D_q^1}(y) =
    \int_{D_q^1} \frac{F(x\, ,y)\, \mathrm{Det}^{q-m/2}(y)}
    {((S_x^\ast)^{-1} \mathrm{Det}^{q-m/2})(y)}\, d\mu_{D_q^1}(y)\;.
\end{displaymath}
Our starting point now is equation (\ref{eq:5.8mrz}). We interchange
the linear operations of doing the ordinary integral $\int_{D_q}
(...)\, d\mu_{D_q}$ and the Berezin integral $\Omega_{W_1}[...]$. The
inner integral over $y$ is then transformed as
\begin{displaymath}
    \int_{D_q^1} (S^\ast F)(x\, ,y)\, \mathrm{Det}^{p-n^\prime}(y)\,
    d\mu_{D_q^1}(y) = \int_{D_q^1} \frac{F(x\,,y)\,\mathrm{Det}^{q
    -m/2}(y)\, d\mu_{D_q^1}(y)}{((S_x^\ast)^{-1}
    \mathrm{Det}^{n^\prime-p+q-m/2})(y)} \;.
\end{displaymath}
The factor $((S_x^\ast)^{-1} \mathrm{Det}^{-n^\prime})(y)$ combines
with the factor $\mathrm{Det}^{n^\prime}(x)$ of the outer integral
over $x$ to give the power of a superdeterminant:
\begin{displaymath}
    ((S_x^\ast)^{-1} \mathrm{Det}^{-n^\prime})(y) \,
    \mathrm{Det}^{n^\prime}(x) = \mathrm{SDet}^{n^\prime}(x\,,y) \;.
\end{displaymath}
Then, restoring the integrations to their original order (i.e.,
Berezin integral first, integral over $y$ second) we immediately
arrive at the formula of the theorem.
\end{proof}
\begin{rem}
The function $J(x\,,y)$ is just the factor that appears in the
definition of the Berezin measure $DQ$ in Sect.\ \ref{sect:1.2}.
Using supermatrix notation, this is seen from the following
computation:
\begin{eqnarray*}
    J(x\,,y) &=&
    \frac{\mathrm{Det}^q(x)\, \mathrm{Det}^{q-m/2}(y)}
    {\mathrm{Det}^{q-m/2-p}(y - \tau\, x^{-1} \sigma)} =
    \frac{\mathrm{Det}^q(x)\, \mathrm{Det}^p(y-\tau\,x^{-1}\sigma)}
    {\mathrm{Det}^{q-m/2}(1 - y^{-1}\tau\, x^{-1} \sigma)} \\ &=&
    \frac{\mathrm{Det}^q(x)\, \mathrm{Det}^{p}(y-\tau\,x^{-1}\sigma)}
    {\mathrm{Det}^{-q+m/2}(1 - x^{-1}\sigma\, y^{-1} \tau)} =
    \frac{\mathrm{Det}^q(x-\sigma\,y^{-1}\tau)\,
    \mathrm{Det}^{p}(y-\tau\,x^{-1}\sigma)}
    {\mathrm{Det}^{m/2}(1 - x^{-1}\sigma\, y^{-1} \tau)} \;.
    \end{eqnarray*}
We also have adjusted the normalization constants, so that $\mathrm{
dvol}_{V_{0,\mathbb{R}}} \otimes \widetilde{\Omega}_{V_1}$ agrees
with the Berezin superintegral form $D_{Z,\bar{Z};\zeta, \tilde{
\zeta}}$ of Eq.\ (\ref{eq:1.1}), and $d\mu_{D_p^0}\, d\mu_{D_q^1}
\otimes \widetilde{\Omega}_{W_1} \circ J$ agrees with $DQ$ as defined
in Eqs.\ (\ref{eq:DQ-u}, \ref{eq:DQ-os}). Thus, assuming the validity
of Thm.\ \ref{thm:4.10} we have now completed the proof of our main
formulas (\ref{bosonize}) and (\ref{bos-other}). To complete the
proof of Thm.\ \ref{thm:4.10} we have to establish the normalization
given by Lemma \ref{lem:4.10}.
\end{rem}

\subsection{Proof of Lemma \ref{lem:4.10}}\label{sect:norm-int}

Lemma \ref{lem:4.10} states the value of the integral
\begin{displaymath}
    \int_{D_q^1} \mathrm{e}^{\mathrm{Tr}^{\,\prime} y}
    \, \mathrm{Det}^{-n^\prime}(y)\, d\mu_{D_q^1}(y)
\end{displaymath}
over the compact symmetric space $D_q^1\,$. To verify that statement,
we are now going to compute this integral by supersymmetric reduction
to a related integral,
\begin{displaymath}
    \int_{D_q^0} \mathrm{e}^{-\mathrm{Tr}^{\,\prime} x}
    \, \mathrm{Det}^{n^\prime +q}(x)\, d\mu_{D_q^0}(x) \;,
\end{displaymath}
over the corresponding non-compact symmetric space $D_q^0\,$. For
that purpose, consider
\begin{equation}\label{eq:Konst}
    C_{n,\,q} := \int_{D_q^0} \int_{D_q^1} \Omega_{W_1} [ (J \cdot
    \mathrm{SDet}^{n^\prime + q})(x\, ,y)] \, \mathrm{e}^{\mathrm{Tr}^{
    \,\prime} y - \mathrm{Tr}^{\,\prime} x} d\mu_{D_q^1}(y)\,
    d\mu_{D_q^0}(x) \;,
\end{equation}
(for each of the three cases $G = \mathrm{GL}, \mathrm{O},
\mathrm{Sp}$) and first process the inner integral:
\begin{eqnarray*}
    &&\int_{D_q^1} \Omega_{W_1} [ (J \cdot \mathrm{SDet}^{n^\prime +
    q}) (x\,,y)] \, \mathrm{e}^{\mathrm{Tr}^{\,\prime} y}\,
    d\mu_{D_q^1}(y)  \\ &=& \int_{D_q^1} \Omega_{W_1} \left[
    \frac{\mathrm{Det}^{n^\prime + 2q}(x) \, \mathrm{Det}^{q-m/2}(y)}
    {((S_x^\ast)^{-1} \mathrm{Det}^{n^\prime + q - m/2})(y)} \right] \,
    \mathrm{e}^{\mathrm{Tr}^{\,\prime} y} \, d\mu_{D_q^1}(y) \\
    &=& \mathrm{Det}^{n^\prime + 2q}(x) \int_{D_q^1} \Omega_{W_1}
    [S_x^\ast (\exp \circ \mathrm{Tr}^{\,\prime}) (y)] \,
    \mathrm{Det}^{-n^\prime}(y) \, d\mu_{D_q^1}(y) \;.
\end{eqnarray*}
Here, after inserting the definitions of $\mathrm{SDet}^{n^\prime +
q}$ and $J$ for $p = q\,$, we again made use of the formula of Lemma
\ref{lem:shift}, reading it backwards this time.

The next step is to calculate the Berezin integral $\Omega_{W_1}$ of
$S_x^\ast (\exp \circ \mathrm{Tr}^{\,\prime}) (y)$. By the definition
of the shift operation $S_x^\ast$ this is a Gaussian integral. Its
value is
\begin{displaymath}
    \Omega_{W_1} [S^\ast (\exp \circ \mathrm{Tr}^{\,\prime}) (y)] =
    \mathrm{e}^{\mathrm{Tr}^{\,\prime} y}\, \mathrm{Det}^q (x^{-1})
\end{displaymath}
in all three cases. Inserting this result into the above expression
for $C_{n,\,q}$ we get the following product of two ordinary
integrals:
\begin{displaymath}
    C_{n,\,q} = \int_{D_q^0} \mathrm{e}^{-\mathrm{Tr}^{\,\prime} x}
    \,\mathrm{Det}^{n^\prime + q}(x) \, d\mu_{D_q^0}(x) \times
    \int_{D_q^1} \mathrm{e}^{\mathrm{Tr}^{\,\prime} y} \,
    \mathrm{Det}^{-n^\prime}(y) \, d\mu_{D_q^1}(y) \;.
\end{displaymath}
The first one is known to us from Eq.\ (\ref{eq:3.2}), while the
second one is the integral that we actually want. The formula claimed
for this integral in Lemma \ref{lem:4.10} is readily seen to be
equivalent to the statement that $C_{n,\,q} = (2\pi)^{(1+|m|) q^2}$.
Thus our final task now is to show that $C_{n,\,q} = (2\pi)^{(1+|m|)
q^2}$. This is straightforward to do by the localization technique
for supersymmetric integrals \cite{SusyLoc}, as follows.

To get a clear view of the supersymmetries of our problem, let us go
back to our starting point: the algebra $\mathcal{A}_V^G$ of
$G$-equivariant holomorphic functions $V_0 \to \wedge(V_1^\ast)$ of
the $\mathbb{Z}_2$-graded vector space $V = V_0 \oplus V_1$ for $V_0
= \mathrm{Hom}(\mathbb{C}^n , \mathbb{C}^p) \oplus \mathrm{Hom}
(\mathbb{C}^p , \mathbb{C}^n)$ and $V_1 = \mathrm{Hom}(\mathbb{C}^n ,
\mathbb{C}^q) \oplus \mathrm{Hom}(\mathbb{C}^q , \mathbb{C}^n)$.
There exists a canonical action of the Lie superalgebra $\mathfrak
{gl}_{p|q}$ on $\mathbb{C}^{p|q}$, hence on $V \simeq \mathrm{Hom}
(\mathbb{C}^n,\mathbb{C}^{p|q}) \oplus \mathrm{Hom}( \mathbb{C}^{p|q}
, \mathbb{C}^n)$, and hence on the algebra $\mathcal{A}_V^G$. To
describe this $\mathfrak{gl}_{p|q}$-action on $\mathcal{A}_V^G$, let
$\{ E_i^a \}$, $\{ \tilde{E}_a^i \}$, $\{ e_i^b \}$, and $\{
\tilde{e}_b^i \}$ with index range $i = 1, \ldots, n$ and $a = 1,
\ldots, p$ and $b = 1, \ldots, q$ be bases of $\mathrm{Hom}(
\mathbb{C}^n, \mathbb{C}^p)$, $\mathrm{Hom} (\mathbb{C}^p,
\mathbb{C}^n)$, $\mathrm{Hom}(\mathbb{C}^n, \mathbb{C}^q)$, and
$\mathrm{Hom}(\mathbb{C}^q, \mathbb{C}^n)$, in this order. If $\{
F_a^i \}$, $\{ \tilde{F}_i^a \}$, etc., denote the corresponding dual
bases, then the odd generators of $\mathfrak{gl}_{p|q}$ (the even
ones will not be needed here) are represented on $\mathcal{A}_V^G$ by
odd derivations
\begin{displaymath}
    d_b^a = \varepsilon(f_b^i) \delta(E_i^a) + \mu(\tilde{F}_i^a)
    \iota(\tilde{e}_b^i) \;, \quad \tilde{d}_a^b = \varepsilon
    (\tilde{f}_i^b)\delta(\tilde{E}_a^i)-\mu(F_a^i)\iota(e_i^b)\;,
\end{displaymath}
where the operators $\varepsilon(f)$, $\delta(v)$, $\mu(f)$, and
$\iota(v)$ mean exterior multiplication by the anti-commuting
generator $f$, the directional derivative w.r.t.\ the vector $v$,
(symmetric) multiplication by the function $f$, and alternating
contraction with the odd vector $v$. Clearly, all of these
derivations are $G$-invariant (for $G = \mathrm{GL}_n\,$,
$\mathrm{O}_n\,$, $\mathrm{Sp}_n\,$) and have vanishing squares
$(d_b^a)^2 = (\tilde{d}_a^b)^2 = 0\,$. Using the coordinate language
introduced in Sect.\ \ref{sect:1.1} one could also write
\begin{displaymath}
    d_b^a = \zeta_b^i \frac{\partial}{\partial Z_a^i} + \tilde{Z}_i^a
    \frac{\partial}{\partial \tilde{\zeta}_i^b} \;,\quad \tilde{d}_a^b
    = \tilde{\zeta}_i^b \frac{\partial}{\partial \tilde{Z}_i^a} - Z_a^i
    \frac{\partial}{\partial \zeta_b^i} \;.
\end{displaymath}
It will be of importance below that the flat Berezin superintegral
form $\mathrm{dvol}_{V_{0,\mathbb{R}}} \otimes \Omega_{V_1}$ is
$\mathfrak{gl}_{p|q}$-invariant, which means in particular that
\begin{displaymath}
    \int \Omega_{V_1}[d_b^a f]\, \mathrm{dvol}_{V_{0,\mathbb{R}}}
    = \int \Omega_{V_1}[\tilde{d}_a^b f]\, \mathrm{dvol}_{V_{0,
    \mathbb{R}}} = 0
\end{displaymath}
for any $f \in \mathcal{A}_V$ with rapid decay when going toward
infinity along $V_{0,\mathbb{R}}\,$.

Superbosonization involves the step of lifting $f \in \mathcal
{A}_V^G$ to $F \in \mathcal{A}_W$ by the surjective mapping $Q^\ast :
\, \mathcal{A}_W \to \mathcal{A}_V^G$. Now, since $W = \mathrm{S}^2
(V)^G = \big( T^2(V) / I_-(V) \big)^G$ and the Lie superalgebra
$\mathfrak{gl}_{p|q}$ acts on $T^2(V)$ by $G$-invariant derivations
stabilizing $I_-(V)$, we also have a $\mathfrak{gl}_{p|q}$-action by
linear transformations $W \to W$. Realizing this action by
derivations of $\mathcal{A}_W$ we obtain a $\mathfrak{gl}_{p|
q}$-action on $\mathcal{A}_W\,$. In particular, there exist such
derivations $\mathcal{D}_b^a$ and $\tilde{\mathcal{D}}_a^b$ that for
every $F \in \mathcal{A}_W$ we have
\begin{displaymath}
    Q^\ast \mathcal{D}_b^a F = d_b^a Q^\ast F \;, \quad Q^\ast
    \tilde{\mathcal{D}}_a^b F = \tilde{d}_a^b Q^\ast F \;.
\end{displaymath}
(In other words, our homomorphism of algebras $Q^\ast :\, \mathcal
{A}_W \to \mathcal{A}_V^G$ is $\mathfrak{gl}_{p|q}$-equivariant.)

For any positive integers $p$, $q$, $n$ with $n^\prime \ge p$
consider now the Berezin superintegral
\begin{displaymath}
    I_{p,\,q}^n [F] := \int_{D_p^0} \int_{D_q^1} \Omega_{W_1}
    [ J \cdot \mathrm{SDet}^{n^\prime} \cdot F] \, d\mu_{D_q^1}
    \, d\mu_{D_q^0} \;,
\end{displaymath}
which includes our integral $C_{n,\,q}$ of interest as a special case
by letting $p = q$ and $F(x\, ,y) = \mathrm{SDet}^q(x\, ,y)\,
\mathrm{e}^{\mathrm{Tr}^{\,\prime} y - \mathrm{Tr}^{\,\prime} x }$.
\begin{lem}
The odd derivations $\mathcal{D}_b^a$ and $\tilde{\mathcal{D}}_a^b$
are symmetries of $F \mapsto I_{p,\,q}^n [F]$; i.e., the integrals of
$\mathcal{D}_b^a F$ and $\tilde{\mathcal{D}}_a^b F$ vanish,
\begin{displaymath}
    I_{p,\,q}^n [\mathcal{D}_b^a F] = I_{p,\,q}^n
    [\tilde{\mathcal{D}}_a^b F] = 0 \quad (a = 1, \ldots, p\,;
    ~ b = 1, \ldots, q) \;,
\end{displaymath}
for any integrand $F \in \mathcal{A}_W$ such that $Q^\ast F
\vert_{V_{0,\mathbb{R}}}$ is a Schwartz function.
\end{lem}
\begin{proof}
While some further labor would certainly lead to a direct proof of
this statement, we will prove it here using the superbosonization
formula of Thm.\ \ref{thm:5.7} in reverse. (Of course, to avoid
making a circular argument, we must pretend to be ignorant of the
constant of proportionality between the two integrals, which will
remain an unknown until the proof of Lemma \ref{lem:4.10} has been
completed. Such ignorance does not cause a problem here, as we only
need to establish a null result.) Thus, applying the formula of Thm.\
\ref{thm:5.7} in the backward direction with an unknown constant, we
have
\begin{displaymath}
    I_{p,\,q}^n [\mathcal{D}_b^a F] = \mathrm{const} \times
    \int \Omega_{V_1}[Q^\ast \mathcal{D}_b^a F]
    \mathrm{dvol}_{V_{0,\mathbb{R}}} \;.
\end{displaymath}
We now use the intertwining relation $Q^\ast \mathcal{D}_b^a = d_b^a
Q^\ast$ of $\mathfrak{gl}_{p|q}$-representations. The integral on the
right-hand side is then seen to vanish because the integral form
$\mathrm{dvol}_{V_{0 ,\mathbb{R}}}\otimes\Omega_{V_1}$ is
$\mathfrak{gl}_{p|q}$-invariant. Thus $I_{p,\,q}^n [\mathcal{D}_b^a
F]= 0\,$. By same argument also $I_{p,\,q}^n [\tilde{\mathcal{D}}_a^b
F] = 0\,$.
\end{proof}
Thus we have $2pq$ odd $\mathcal{A}_W$-derivations (or vector fields)
$\mathcal{D}_b^a$ and $\tilde{\mathcal{D}}_a^b$ which are symmetries
of $I_{p,\,q}^n\,$. We mention in passing that for the cases of $G =
\mathrm{O}_n$ and $G = \mathrm{Sp}_n$ there exist further symmetries
which promote the full symmetry algebra from $\mathfrak{gl}_{p|q}$ to
$\mathfrak{osp}_{2p|2q}\,$. This fact will not concern us here.

Let now $p = q\,$. Then there exists a distinguished symmetry
\begin{displaymath}
    \mathcal{D} := \mathcal{D}_a^a = \mathcal{D}_1^1
    + \mathcal{D}_2^2 + \ldots + \mathcal{D}_q^q \;,
\end{displaymath}
which still satisfies $\mathcal{D}^2 = 0\,$.
\begin{lem}
Viewed as a vector field on the supermanifold of functions $D_q^0
\times D_q^1 \to \wedge(W_1^\ast)$, the numerical part of $\mathcal
{D}$ vanishes at a single point $o \equiv (\mathrm{Id}, \mathrm{Id})
\in D_q^0 \times D_q^1\,$.
\end{lem}
\begin{proof}
We sketch the idea of the proof for $G = \mathrm{GL}_n\,$. In that
case one verifies that $\mathcal{D}$ has the coordinate expression
\begin{displaymath}
    \mathcal{D} = \sigma_b^a \left( \frac{\partial}{\partial x_b^a} +
    \frac{\partial}{\partial y_b^a} \right) + \left( x_b^a - y_b^a
    \right) \frac{\partial}{\partial\tau_b^a} \;.
\end{displaymath}
The second summand, the numerical part of $\mathcal{D}$, is zero only
when the coordinate functions $x_b^a$ and $y_b^a$ are equal to each
other for all $a, b = 1, \ldots, q\,$. Since $D_q^1 = \mathrm{U}_q$
and $D_q^0$ is the set of positive Hermitian $q \times q$ matrices,
this happens only for $x = \mathrm{Id} \in D_q^0$ and $y =
\mathrm{Id} \in D_q^1\, $. The same strategy of proof works for the
cases of $G = \mathrm{O}_n\, , \mathrm{Sp}_n\,$.
\end{proof}
%
%

We are now in a position to apply the localization principle for
supersymmetric integrals \cite{SusyLoc}. Let $F \in \mathcal{A}_W$ be
a $\mathcal{D}$-invariant function which is a Schwartz function on
$D_q^0\,$. Choose a $\mathcal{D}$-invariant function $g_\mathrm{loc}
: \, D_q^0 \times D_q^1 \to \wedge(W_1^\ast)$ with the property that
$g_\mathrm{loc} = 1$ on some neighborhood $U(o) \subset V(o)$ of $o$
and $g_\mathrm{loc} = 0$ outside of $V(o)$. (Such ``localizing''
functions do exist.) Then according to Theorem 1 of \cite{SusyLoc} we
have
\begin{displaymath}
    I_{q,\,q}^n [F] = I_{q,\,q}^n [g_\mathrm{loc} \, F] \;,
\end{displaymath}
since $I_{q,\,q}^n$ is $\mathcal{D}$-invariant. (Although that
theorem is stated and proved for compact supermanifolds, the
statement still holds for our non-compact situation subject to the
condition that integrands be Schwartz functions.)

Taking $V(o)$ to be arbitrarily small we conclude that $F \to
I_{q,\,q}^n [F]$ depends only on the numerical part of the value of
$F$ at $o:$
\begin{displaymath}
    I_{q,\,q}^n [F] = \mathrm{const} \times \mathrm{num}(F(o)) \;.
\end{displaymath}
To determine the value of the constant for $G = \mathrm{GL}_n$ we
consider the special function
\begin{displaymath}
    F = \mathrm{e}^{- \frac{t}{2} (x_a^b x_b^a - y_a^b y_b^a +
    2 \sigma_a^b \tau_b^a)} \;.
\end{displaymath}
An easy calculation in the limit $t \to +\infty$ then gives $I_{q,
\,q}^n[F] = (2\pi)^{q^2} \mathrm{num}(F(o))$ due to our choice of
normalization for $d\mu_{D_q^0}$ and $d\mu_{D_q^1}\,$. The same
calculation for the cases of $G = \mathrm{O}_n\, , \mathrm{Sp}_n$
gives $I_{q,\,q}^n[F] = (2\pi)^{2 q^2} \mathrm{num}(F(o))$.

These considerations apply to the integrand in Eq.\ (\ref{eq:Konst})
with $\mathrm{num}(F(o)) = 1$. Thus we do indeed get $C_{n,\,q} =
(2\pi)^{ (1+|m|) q^2}$, and the proof of Lemma \ref{lem:4.10} is now
finished.

\section{Appendix: invariant measures}

In the body of this paper we never gave any explicit expressions for
the invariant measures $d\mu_{D_p^0}$ and $d\mu_{D_q^1}\,$. There was
no need for that, as these measures are in fact determined (up to
multiplication by constants) by invariance with respect to a
transitive group action, and this invariance really was the only
property that was required.

Nevertheless, we now provide assistance to the practical user by
writing down explicit formulas for our measures (or positive
densities) $d\mu_{D_p^0}$ and $d\mu_{D_q^1}\,$. For that purpose, we
will use the correspondence between densities and differential forms
of top degree. (Recall what the difference is: densities transform by
the absolute value of the Jacobian, whereas top-degree differential
forms transform by the Jacobian including sign.) Thus we shall give
formulas for the differential forms corresponding to $d\mu_{D_p^0}$
and $d\mu_{D_q^1}\,$. This is a convenient mode of presentation, as
it allows us to utilize complex coordinates for the complex ambient
spaces as follows.

Consider first the case of $G = \mathrm{GL}_n(\mathbb{C})$ where
$D_p^0 = \mathrm{Herm}^+ \cap \mathrm{End}(\mathbb{C}^p)$ and $D_q^1
= \mathrm{U} \cap \mathrm{End}(\mathbb{C}^q)$. Then for $r = p$ or $r
= q$ consider $\mathrm{End}(\mathbb{C}^r)$ and let $z_{c c^\prime} :
\, \mathrm{End}(\mathbb{C}^r) \to \mathbb{C}$ (with $c,c^\prime = 1,
\ldots, r$) be the canonical complex coordinates of $\mathrm{End}
(\mathbb{C}^r)$, i.e., the set of matrix elements with respect to the
canonical basis of $\mathbb{C}^r$. On the set of regular points of
$\mathrm{End}(\mathbb{C}^r)$ define a holomorphic differential form
$\omega^{(r)}$ by
\begin{displaymath}
    \omega^{(r)} = \mathrm{Det}^{-r}(z) \bigwedge_{c,\,
    c^\prime = 1}^r dz_{c c^\prime} \;,
\end{displaymath}
where $z = (z_{c c^\prime})$ is the matrix of coordinate functions.
By the multiplicativity of the determinant and the alternating
property of the wedge product, $\omega^{(r)}$ is invariant under
transformations $z \mapsto g_1 z\, g_2^{-1}$ for $g_1 , g_2 \in
\mathrm{GL}_r(\mathbb{C})$. The desired invariant measures (up to
multiplication by an arbitrary normalization constant) are
\begin{equation}
    d\mu_{D_p^0} \propto \omega^{(p)} \big\vert_{\mathrm{Herm}^+ \cap
    \mathrm{End}(\mathbb{C}^p)} \;, \quad
    d\mu_{D_q^1} \propto \omega^{(q)} \big\vert_{\mathrm{U} \cap
    \mathrm{End}(\mathbb{C}^q)} \;,
\end{equation}
where we restrict $\omega^{(r)}$ as indicated and reinterpret
$d\mu_{D_r^\bullet}$ as a positive density on the orientable manifold
$D_r^\bullet$ ($r = p , q$). For example, for $r = 1$ we have
$\omega^{(1)} = z^{-1} dz\,$. In this case we get an invariant
positive density $|dx|$ on the positive real numbers $\mathrm{Herm}^+
\cap \mathbb{C} = \mathbb{R}_+$ by setting $z = \mathrm{e}^x$ with $x
\in \mathbb{R}\, $, and a Haar measure $|dy|$ on the unit circle
$\mathrm{U} \cap \mathbb{C} = \mathrm{U}_1 = \mathrm{S}^1$ by setting
$z = \mathrm{e}^{\mathrm{i} y}$ with $0 \le y \le 2\pi$. Our
normalization conventions for the invariant measures $d\mu_{
D_r^\bullet}$ are those described in Sect.\ \ref{sect:1.2}.

We turn to the cases of $G = \mathrm{O}_n(\mathbb{C})$ and $G =
\mathrm{Sp}_n(\mathbb{C})$ and recall that the condition on elements
$M$ of the complex linear space $\mathrm{Sym}_b(\mathbb{C}^{2r})$ is
$M = t_b\, M^\mathrm{t} (t_b)^{-1}$. On making the substitution $M =
L \, t_b$ this condition turns into
\begin{displaymath}
    \begin{array}{ll} L = + L^\mathrm{t} &\text{~~for~~} b = s \;,
    \\ L = - L^\mathrm{t} &\text{~~for~~} b = a \;, \end{array}
\end{displaymath}
while the $\mathrm{GL}_{2r}(\mathbb{C})$-action on $\mathrm{Sym}_b
(\mathbb{C}^{2r})$ by twisted conjugation becomes $g . L = g L\,
g^\mathrm{t}$ in both cases. Define the coordinate function $z_{c
c^\prime} : \, \mathrm{Sym}_b (\mathbb{C}^{2r}) \to \mathbb{C}$ to be
the function that assigns to $M$ the matrix element of $M (t_b)^{-1}
= L$ in row $c$ and column $c^\prime$. We then have $z_{c c^\prime} =
z_{c^\prime c}$ for $b = s$ and $z_{c c^\prime} = - z_{c^\prime c}$
for $b = a\,$. As before let $z= (z_{c c^\prime})$ be the matrix made
from these coordinate functions (where the transpose $z^\mathrm{t} =
z$ for $b=s$ and $z^\mathrm{t}= -z$ for $b = a$). Then let top-degree
differential forms $\omega^{(r;\,b)}$ be defined \emph{locally} on
the regular points of $\mathrm{Sym}_b( \mathbb{C}^{2r})$ by
\begin{eqnarray*}
    &&\omega^{(r;\,s)} = \mathrm{Det}^{-r-1/2}(z)
    \bigwedge_{1 \le c \le c^\prime \le r} dz_{c c^\prime} \;, \\
    &&\omega^{(r;\,a)} = \mathrm{Det}^{-r+1/2}(z)
    \bigwedge_{1 \le c < c^\prime \le r} dz_{c c^\prime} \;.
\end{eqnarray*}
These are invariant under pullback by $L \mapsto g L\, g^\mathrm{t}$,
as the transformation behavior of $\mathrm{Det}^{-r \pm 1/2}$ is
contragredient to that of the wedge product of differentials in both
cases. We emphasize that this really is just a local definition so
far, as the presence of the square root factors may be an obstruction
to the global existence of such a form.

Now focus on the case of $G = \mathrm{O}_n(\mathbb{C})$. There,
restriction to the domains $D_{\delta,p}^0$ and $D_{\delta,\,q}^1$
gives the differential forms $\omega^{(p;\,s)} \big\vert_{\mathrm{
Herm}^+ \cap \mathrm{Sym}_s(\mathbb{C}^{2p})}$ and $\omega^{(q;\,a)}
\big\vert_{\mathrm{U} \cap \mathrm{Sym}_a (\mathbb{C}^{2q})}$. Both
of these are globally defined. Indeed, we can take the factor $L
\mapsto \mathrm{Det}^{-1/2}(L)$ in the first differential form to be
the reciprocal of the positive square root of the positive Hermitian
matrix $M = L\, t_s\,$, and the square root $L \mapsto \mathrm{Det}^{
1/2}(L)$ appearing in the second form makes global sense as the
Pfaffian of the unitary skew-symmetric matrix $L^\mathrm{t} = - L\,$.
Reinterpreting these differential forms as densities we arrive at a
$\mathrm{GL}_{2p}( \mathbb{R})$-invariant measure on $D_{\delta,p}^0$
and a $\mathrm{U}_{2q}$-invariant measure on $D_{\delta,\,q}^1$:
\begin{equation}
    d\mu_{D_{\delta,p}^0} \propto \omega^{(p;\,s)} \big\vert_{
    \mathrm{Herm}^+ \cap \mathrm{Sym}_s(\mathbb{C}^{2p})} \;, \quad
    d\mu_{D_{\delta,\,q}^1} \propto \omega^{(q;\,a)} \big\vert_{
    \mathrm{U} \cap \mathrm{Sym}_a(\mathbb{C}^{2q})} \;.
\end{equation}
Again, our normalization conventions for $d\mu_{D_{\delta,
r}^\bullet}$ are those of Sect.\ \ref{sect:1.2}.

In the final case of $G = \mathrm{Sp}_n(\mathbb{C})$ the roles of
$\omega^{(\bullet;\,s)}$ and $\omega^{(\bullet;\,a)}$ are reversed.
This immediately leads to a good definition of $d\mu_{D_{\varepsilon
, p}^0}$ for the non-compact symmetric space $D_{\varepsilon,p}^0\,$.
However, the remaining case of $D_{ \varepsilon,\,q}^1 = \mathrm{U}
\cap \mathrm{Sym}_s( \mathbb{C}^{2q})$ is problematic because there
exists no global definition of $\mathrm{Det}^{1/2}$ on the unitary
symmetric matrices. Thus the locally defined differential form
$\omega^{(q;\, s)}$ does \emph{not} extend to a globally defined form
on $D_{\varepsilon,\, q}^1\,$. (Please be advised that this is
inevitable, as the compact symmetric space $D_{\varepsilon,\,q}^1
\simeq \mathrm{U}_{2q} / \mathrm{O}_{2q}$ lacks the property of
orientability and on a non-orientable manifold any globally defined
top-degree differential form must have at least one zero and
therefore cannot be both non-zero and invariant in the required
sense.)

Of course $d\mu_{D_{\varepsilon,\,q}^1}$ still exists as a density on
the non-orientable manifold $D_{\varepsilon,\,q}^1\,$. The discussion
above is just saying that there exists no globally defined
differential form corresponding to $d\mu_{D_{\varepsilon,\, q}^1}\,$.
Locally, we have $d\mu_{D_{\varepsilon,\,q}^1} \propto \omega^{(q
;\,s)} \big\vert_{\mathrm{U} \cap \mathrm{Sym}_s (\mathbb{C}^{2q})}$.

\textbf{Note added in proof.} After submission of this paper, it was
brought to our attention that a bosonization formula for the
fermion-fermion sector with $\mathrm{U}_n$ symmetry had been
developed by Kawamoto and Smit \cite{kawamoto}. A supersymmetric
generalization was suggested in Ref.\ \cite{berruto}.

\end{document}